\documentclass[11pt]{article}
\usepackage{graphicx} %

\usepackage[utf8]{inputenc}
\usepackage{amsmath,amsthm,amssymb,xcolor}
\usepackage[margin=2.54cm]{geometry}
\usepackage{physics}
\usepackage[hidelinks]{hyperref}
\usepackage{float}
\usepackage{hyperref}
\usepackage{multirow}
\usepackage{enumitem}
\usepackage{cleveref}

\usepackage{thmtools}
\usepackage{thm-restate}

\usepackage{makecell}

\usepackage{lineno}

\usepackage{multirow}
\usepackage[ruled,linesnumbered,vlined,noend]{algorithm2e}
\crefname{algocf}{algorithm}{algorithms}
\Crefname{algocf}{Algorithm}{Algorithms}

\usepackage{pgfplots}
\usepackage{xcolor}
\pgfplotsset{compat=1.7}
\usetikzlibrary{tikzmark, arrows}
\usetikzlibrary{calc}
\usetikzlibrary{patterns,patterns.meta} 
\usetikzlibrary{intersections}
\usepackage{cite}
\usepackage{tkz-euclide}

\usepackage[table, dvipsnames]{xcolor}

\newtheorem{theorem}{Theorem}[section]
\newtheorem{lemma}[theorem]{Lemma}
\newtheorem{fact}[theorem]{Fact}

\usepackage{bm}

\usepackage{natbib}

\usepackage{dsfont}

\usepackage[disable]{todonotes}

\title{On the Power of Adaptivity in Testing Quantum States in Fidelity
}

\author{
Jan Seyfried\\
Centre for Quantum Technologies,\\
National University of Singapore\\
jan.seyfried@u.nus.edu
\and
Sayantan Sen\\
Centre for Quantum Technologies,\\
National University of Singapore\\
sayantan789@gmail.com
\and
Marco Tomamichel \\
Department of Electrical and Computer Engineering,\\
Centre for Quantum Technologies,\\
National University of Singapore\\
marco.tomamichel@nus.edu.sg
}

\date{}

\begin{document}

\maketitle
\newcommand{\eps}{\varepsilon}

\newcommand{\subD}[3]{{#1}^{#2}_{#3}}

\newcommand{\E}{\mathbb{E}}

\newcommand{\PrOutcome}[2]{\Pr[#1|#2]}

\newcommand{\cHSdistr}{c_{\textnormal{test}}^{\ell_2, C}}

\newcommand{\cHSquantum}{c_{\textnormal{test}}^{\ell_2, Q}}

\newcommand{\cFquantum}{c_{\textnormal{id}}^{F, Q}}

\newcommand{\cOpLearn}{c_{\textnormal{learn}}^{\opn}}

\newcommand{\cOpG}{c_{\textnormal{learn}}^{F,\opn}}

\newcommand{\cMtest}{c_{\textnormal{test}}^{M}}

\newcommand{\cOpMix}{c_{\textnormal{mix}}^{\opn}}

\newcommand{\opn}{\textnormal{op}}

\newcommand{\bv}{s}
\newcommand{\eqM}{\hat{M}}

\newcounter{mynotes}
\setcounter{mynotes}{0}

\newcommand{\ssnotes}[1]{\addtocounter{mynotes}{1}{{}}\todo[color=blue!20!white]{[\arabic{mynotes}] \scriptsize  {Sayantan: {\sf {#1}}}}}

\newcommand{\jsnotes}[1]{\addtocounter{mynotes}{1}{{}}\todo[color=cyan!20!white]{[\arabic{mynotes}] \scriptsize  {Jan: {\sf {#1}}}}}

\newcommand{\mtnotes}[1]{\addtocounter{mynotes}{1}{{}}\todo[color=green!20!white]{[\arabic{mynotes}] \scriptsize  {Marco: {\sf {#1}}}}}

\newtheorem{proposition}{Proposition}[section]
\newtheorem{claim}{Claim}[section]
\newtheorem{remark}[theorem]{Remark}
\newtheorem{definition}[theorem]{Definition}

\newcommand{\blue}[1]{\textcolor{black}{#1}}
\newcommand{\red}[1]{\textcolor{black}{#1}}

\newcommand{\green}[1]{\textcolor{green!80!black}{#1}}

\newcommand{\snote}[1]{{\color{orange} Sayantan: #1}}

\begin{abstract}
We study the problems of quantum state certification, equivalence testing and independence testing. In certification, given samples of an unknown quantum state $\rho$ and the description of a state $\sigma$, the goal is to test whether $\rho=\sigma$, or whether $\rho$ and $\sigma$ are far in a given distance measure. In equivalence testing, $\sigma$ is also unknown and only accessible via samples. Independence testing decides whether $\rho_{AC}=\rho_A\otimes\rho_C$, or is far from being a product. The sample complexities of these problems are now well-understood for a decision gap $\varepsilon$ in trace distance: in the single-copy measurement setting with $d$-dimensional states, all three tasks can be solved using the same non-adaptive approach, which uses $\Theta(d^{3/2}/\varepsilon^2)$ samples and is optimal in general, even without adaptivity.

In this work, we consider decision gaps expressed in fidelity and study possible separations between these problems and how adaptivity can help. We prove that certification with respect to fidelity for a state $\sigma$ of rank $r$ does not benefit from adaptivity and requires $\widetilde{\Theta}(r^{3/2}/\varepsilon)$ samples. For equivalence testing and independence testing, we provide adaptive algorithms using $\widetilde{O}(\min\{d^{3/2}/\varepsilon^2,d^{9/4}/\varepsilon\})$ and $\widetilde{O}(\min\{(d_Ad_C)^{3/2}/\varepsilon^2,d_A^{9/4}d_C^{3/4}/\varepsilon\})$ samples, for $d_A\geq d_C$, respectively. Our main technique is a framework that uses partial learning and a reduction to testing in $\ell_2$-distance, adapted from the distribution testing literature. We show that adaptivity matters for equivalence testing in fidelity by proving that $\widetilde{\Omega}(1/\varepsilon^2)$ samples are necessary in the non-adaptive case even for qubits, showing a separation from certification.

\end{abstract}
\newpage
\tableofcontents
\section{Introduction}

In this paper, we study three fundamental problems in quantum property testing:
quantum \emph{state certification}, \emph{equivalence testing}, and \emph{independence testing}. For state certification, we are given the classical description of a quantum state $\sigma \in \mathbb{C}^{d \times d}$, together with sample access to an unknown quantum state $\rho \in \mathbb{C}^{d \times d}$. The goal is to distinguish, with high probability, between the case when $\rho=\sigma$ and the case where $\rho$ is \emph{far} from $\sigma$ with respect to a given distance measure $L$, i.e.,
\begin{align}
    \rho=\sigma \quad\text{or}\quad L(\rho,\sigma)\geq \eps.
\end{align}
The parameter $\eps$ is usually referred to as the \emph{decision gap} (and if the distance between $\rho$ and $\sigma$ is less than $\eps$, the algorithm may output either value). Quantum equivalence testing studies the same question, but in the setting where both states are unknown, and only accessible via samples. The third problem, independence testing, is aimed at testing the structure of a given state. Here, we are given samples of a state $\rho_{AC}\in \mathbb{C}^{d_Ad_C \times d_Ad_C}$ and need to decide whether it has product structure, $\rho_{AC}=\rho_A\otimes\rho_C$ or is far from being a product state\footnote{We assume in the following that it is possible to generate samples of $\rho_A\otimes\rho_C$ given samples from $\rho_{AC}$, for example via partial traces, even though this technically requires two copies of $\rho_{AC}$.}.

These problems are the natural analogues of testing problems in distribution testing, where we receive i.i.d.\ samples from an unknown probability distribution and the goal is to distinguish whether two distributions are the same or far under suitable distance measures. In the classical setting, these testing problems have been studied extensively, see \cite{goldreich1998property,DBLP:conf/focs/BatuFRSW00,gs009,canonne2022topics} and the references therein. It is natural to expect that techniques used for distribution testing can, in some form, be generalized to the quantum setting, which is one of the key motivations behind our work. The problem of quantum property testing has been studied in great detail, see the surveys~\cite{montanaro2016survey,DBLP:journals/sigact/ArunachalamW17,DBLP:journals/corr/abs-2305-20069} and the references therein.

An important aspect of the setting is the choice of measurement model. In the most powerful \emph{multi-copy} (or \emph{coherent}) model, one is allowed to
perform an arbitrary joint measurement on the joint state $\rho^{\otimes n}$. This model has been extensively studied in the literature~\cite{buadescu2019quantum, scharnhorst2025optimal, odonnell_certification}. While this model is mathematically interesting, it requires storing all copies coherently and
measuring them simultaneously, which is often practically infeasible for large $n$.

A more restricted, although arguably more practically feasible access model is the setting where only \emph{single-copy} (or \emph{incoherent})
measurements are considered, that is, the copies of $\rho$ are measured individually. A further distinction in the single-copy model is the separation between \emph{adaptive} and \emph{non-adaptive} protocols. In the adaptive setting, the measurement performed on a copy may depend on the outcomes of earlier measurements. In contrast, in the \emph{non-adaptive} setting, all measurements are fixed in advance, but possibly randomized. %

The three testing problems we are interested in have been studied extensively in the single-copy setting for decision gaps expressed in terms of the \emph{trace distance}. For both certification and equivalence testing, that is, 
\begin{align}
    \text{distinguishing whether}\quad \rho=\sigma \quad\text{or}\quad \|\rho-\sigma\|_{\tr}\geq\eps,
\end{align}
a sequence of works~\cite{bubeck_entanglement_2020,chen_toward_2021,chen2022tightboundsquantumstate} established tight bounds of $\widetilde{\Theta}\!\left({d^{3/2}}/{\eps^2}\right)$ copies. Surprisingly, it has also been shown that adaptive single-copy measurements give no general advantage over non-adaptive strategies. Further, the direct application of this algorithm to states $\rho_{AC}$ and $\rho_A\otimes\rho_C$ is generally optimal for independence testing, which intuitively is due to the fact that the hard instances are based on the maximally mixed state \cite{yu:LIPIcs.ITCS.2021.11}.

In contrast, for distribution testing, there are clear separations between the sample complexities of the different testing tasks. This makes intuitive sense, as there is a natural hierarchy between them, with equivalence testing being more general than certification and mutual information testing, and algorithms can exploit the additional guarantees. It is natural to ask whether testing quantum states can show similar separations between problems, and also the role adaptivity plays in it.

As a first step, this work studies the aforementioned problems with respect to a guarantee expressed in \emph{fidelity}, another natural notion of distance/closeness between quantum states. Interestingly, the role of adaptivity has already been studied for the task of \emph{tomography}, i.e., the task of learning a quantum state to a given precision $\eps$.

In a recent seminal work, \cite{chen_when_2023} showed that adaptivity does not help when learning in trace distance, mirroring the behavior of the property testing regime described earlier. However, for quantum state tomography with respect to fidelity, adaptive single-copy measurements are generally more powerful than non-adaptive ones, achieving sample complexities of $\Theta(d^3/\eps)$ instead of the non-adaptive $\Theta(d^3/\eps^2)$. It is also important to note that this marks a natural trade-off between fidelity and trace distance: the latter gives a stronger guarantee, but is often more difficult to learn or test for. %
Meanwhile, the corresponding role of adaptivity for quantum property-testing tasks has remained much less understood. As mentioned earlier, for certification, equivalence, and independence testing with respect to trace distance, it is known that adaptivity does not provide any advantage over non-adaptive measurements. 
This motivates the following question:
\begin{center}
\emph{Does adaptivity help in  
testing quantum states with respect to fidelity,
and if yes, how much?}    
\end{center}
We address the first part of the question by showing a separation: For certification in fidelity, adaptivity is not needed to achieve the optimal $\widetilde{\Theta}(d^{3/2}/\eps)$. However, we prove a separation for equivalence testing by providing lower bounds of $\widetilde{\Omega}(1/\eps^2)$ in the non-adaptive case. Moreover, we demonstrate that for equivalence testing, adaptivity can be used more effectively than via tomography by providing an equivalence testing algorithm in $\widetilde{O}(\min\{d^{3/2}/\eps^2,d^{9/4}/\eps\})$. As described below, we further show that for the more specific task of independence testing, we can use adaptivity to be even more efficient.

When it comes to the techniques used to derive algorithms, we would ideally like to benefit from parallels between distribution testing and quantum property testing, and suitably leverage existing ideas to the quantum setting. One powerful framework, which we will outline in more detail in \Cref{sec:overview}, reduces testing problems in different distance measures to a subroutine of testing in the $\ell_2$-distance \cite{DBLP:conf/focs/DiakonikolasK16}. It has been shown to give sample-optimal algorithms for different testing tasks, in particular for testing problems with guarantees and decision gaps in fidelity. The approach can be understood as a partial-learning and testing framework. One of our key questions is:
\begin{center}
    \emph{Can the $\ell_2$-testing framework be generalized to quantum and give sample optimal algorithms?}
\end{center}
Our interest in this framework is based on the fact that it relies on certain subroutines for testing and partial learning which are already mostly available in the quantum literature. This helps in the design and analysis of testing algorithms. The reliance on well-defined subroutines also makes the framework easy to adapt to other settings as well, such as joint measurements.

As an application of our algorithms and the underlying framework, we use our independence tester to test for \emph{quantum mutual information}. For a bipartite state $\rho_{AC} \in \mathbb{C}^{d_{AC} \times d_{AC}}$, the mutual information
\begin{align}
    I(A:C)_\rho
    =D(\rho_{AC}\|\rho_A\otimes\rho_C),
\end{align}
where $D(\cdot\|\cdot)$ denotes the relative entropy. The mutual information is a fundamental quantity in (quantum) information theory and appears throughout quantum communication, entanglement theory, many-body physics, and learning theory.
This is the quantum analogue of classical mutual information testing, a problem that is well studied in the literature~\cite{DBLP:conf/focs/DiakonikolasK16, canonne_testing_2018, pmlr-v291-seyfried25a}.

In \emph{mutual information testing}, given samples of an unknown bipartite state $\rho_{AC}$, the goal is to distinguish with high probability whether $I(A:C)_\rho = 0$ (i.e., $\rho_{AC}=\rho_A\otimes\rho_C$), or $I(A:C)_\rho \geq \eps$. Our algorithm achieves a sample complexity of $\widetilde{O}(\min\{(d_Ad_C)^{3/2}/\eps^2,d_A^{9/4}d_C^{3/4}/\eps\})$, which, for certain regimes, improves on applying mere equivalence testing by a factor of $\widetilde{O}(d_C^{3/2})$.

While it remains an open question whether our algorithms for equivalence testing and independence testing are optimal, we conjecture that optimal approaches will need to take advantage of the different problem settings in a similar manner, which we in particular expect to be reflected in different sample complexities. For example, we already note here that equivalence testing can always be used directly for independence testing: for trace distance, this is optimal, but our approach suggests that in fidelity, one can be more efficient.
The reason for our conjecture is that testing in fidelity appears to be much more sensitive to the structure of the states than trace distance, as we will also outline in \Cref{sec:overview}. We further note that in classical distribution testing, where the equivalent problems have all been fully resolved, we see separations in the sample complexities between all problems for both the trace distance (i.e., total variation distance), and the fidelity. We further believe that the study of tailored algorithms for these basic testing routines may be valuable also for the development of other algorithms which aim to be more sample-efficient: this may be achieved by providing a different set of guarantees instead of a stringent guarantee in trace distance or fidelity, which may be chosen to allow for an efficient use of the structure of the problem.

\color{blue}

\color{black}

\subsection{Our Results}

Let us start with our result on certification. Here, given the complete description of a quantum state $\sigma \in \mathbb{C}^{d \times d}$ and samples (copies) of an unknown quantum state $\rho \in \mathbb{C}^{d \times d}$, we show that certification with respect to fidelity can be done with $\widetilde{\Theta}(d^{3/2}/ \eps)$ samples of $\rho$, where $\eps \in (0,1)$ is the distance parameter. If the rank of $\sigma$ is $r$, this can be done with only $\widetilde{\Theta}(r^{3/2}/ \eps)$ samples of $\rho$, independent of the ambient dimension $d$. %
Concretely, we prove the following.

\begin{restatable}[Certification]{theorem}{fididtesting}\label{lemma:fid_id_testing}
Let $\rho, \sigma \in \mathbb{C}^{d \times d}$ be an unknown and a known quantum states, respectively, and $0<\eps \leq \eps_0$ be a parameter, for $\eps_0\ll 1$ a suitably small constant. Moreover, $\sigma$ is of rank $r$. %
Then, using non-adaptive single-copy measurements, deciding whether
\begin{align}
        \rho=\sigma
        \quad\text{or}\quad
        F(\rho,\sigma)\leq 1-\varepsilon
\end{align}
with probability at least $0.99$, $\widetilde{O}(r^{3/2}/\varepsilon)$ samples of $\rho$ sufficient. On the other hand, there exist $\sigma$ such that $\widetilde{\Omega}(r^{3/2}/\eps)$ samples are necessary for certification, even for adaptive protocols.   
\end{restatable}

Note that the success probability in the above theorem is arbitrary and can be improved to any $1- \delta$ for $\delta \in (0,1)$, with an additional multiplicative $\log 1/\delta$ factor by using a majority voting argument. The testing algorithm underlying the above result is non-adaptive. The lower bound follows from a reduction to mixedness testing~\footnote{A quantum state $\rho \in \mathbb{C}^{d \times d}$ is called \emph{maximally mixed} if $\rho= \mathds{1}/d$. In mixedness testing, the goal is to distinguish whether an unknown state $\rho$ is maximally mixed or far from it.}.

Next we move to equivalence testing.  Unlike certification, here both quantum states $\rho$ and $\sigma$ are unknown, and the goal is to distinguish whether $\rho=\sigma$ or the fidelity between them is at most $1- \eps$. We show that $\widetilde{O}(\min\{d^{3/2}/\varepsilon^2,d^{9/4}/\varepsilon\})$ samples of $\rho$ and $\sigma$ are sufficient for fidelity testing. Unlike our certification result, the testing algorithm is adaptive.

\begin{restatable}[Equivalence Testing]{theorem}{fideqtesting}\label{lemma:fid_eq_testing}
Let $\rho, \sigma \in \mathbb{C}^{d \times d}$ be two unknown quantum states and $\eps \in (0,1)$ be a parameter. Then, using single-copy adaptive measurements, deciding whether
\begin{align}
        \rho=\sigma
        \quad\text{or}\quad
        F(\rho,\sigma)\leq 1-\varepsilon
\end{align}
with probability at least $0.99$, $\widetilde{O}(\min\{d^{3/2}/\varepsilon^2,d^{9/4}/\varepsilon\})$ samples of $\rho$ and $\sigma$ are sufficient.    
\end{restatable}

As mentioned before, since certification and equivalence testing with respect to trace distance require $\Theta(d^{3/2}/\eps^2)$ samples, the above result shows improvement over testing in trace distance for the regime $\varepsilon<1/d^{3/4}$. A comparison with prior results is presented in \Cref{tab:overview_cert_equiv}. Note that while the result of \cite{chen_toward_2021} is stated for quantum state certification, their testing algorithm can be used for quantum equivalence testing as well and achieves the same bound.

\begin{table}[h!]
    \centering
    \renewcommand{\arraystretch}{1.3}
    \setlength{\arrayrulewidth}{0.5pt}
    \begin{tabular}{|c|c|c|c|c|}
        \hline
        & \multicolumn{2}{c|}{Identity Testing (Certification)}
        & \multicolumn{2}{c|}{Equivalence Testing} \\
        \hline
        & $\ell_1/\text{Trace distance}$
        & $D_H^2/\text{Fidelity}$
        & $\ell_1/\text{Trace distance}$
        & $D_H^2/\text{Fidelity}$ \\
        \hline

        \rule{0pt}{3ex}\multirow{ 2}{*}{Classical}
        & $\Theta\left(\frac{d^{1/2}}{\varepsilon^2}\right)$
        & $\Theta\left(\frac{d^{1/2}}{\varepsilon}\right)$
        & $\widetilde\Theta\left(\max\left\{\frac{d^{1/2}}{\varepsilon^2},\frac{d^{2/3}}{\varepsilon^{4/3}}\right\}\right)$
        & $\widetilde\Theta\left(\min\left\{\frac{d^{3/4}}{\varepsilon},\frac{d^{2/3}}{\varepsilon^{4/3}}\right\}\right)$ \\
        
        & \cite{paninski_liam}
        & \cite{paninski_liam, sublinearly}
        & \cite{cdvv}
        & \cite{DBLP:conf/focs/DiakonikolasK16} 
        \\
        \hline
        \rule{0pt}{3ex}
        \multirow{ 2}{*}{\makecell{Quantum\\(adaptive)}}
        & $\Theta\left(\frac{d^{3/2}}{\varepsilon^2}\right)$
        & \cellcolor{yellow!15}$\widetilde{\Theta}\left(\frac{d^{3/2}}{\varepsilon}\right)$
        & $\widetilde{\Theta}\left(\frac{d^{3/2}}{\varepsilon^2}\right)$
        & \cellcolor{yellow!15}$\widetilde{O}\left(\min\left\{\frac{d^{3/2}}{\varepsilon^2}, \frac{d^{9/4}}{\varepsilon}\right\}\right)$ \\
        
        & \cite{chen_toward_2021}
        & \cellcolor{yellow!15}[this work]
        & \cite{chen_toward_2021}
        & \cellcolor{yellow!15}[this work] \\
        \hline
    \end{tabular}

    \caption{Comparison between sample complexities for testing identity and equivalence, as a function of the dimension $d$ and precision $\varepsilon$.}
    \label{tab:overview_cert_equiv}
\end{table}

We also study testing for correlations between subsystems in a quantum state. As mentioned before, mutual information is the natural measure for such problems. More formally, given sample access to a bipartite quantum state $\rho_{AC}$, we want to distinguish whether the subsystems are independent, or the mutual information is at least $\eps$.%

\begin{restatable}[Mutual Information Testing]{theorem}{mitesting}\label{lemma:mi_testing}
Let $\rho_{AC} \in \mathbb{C}^{d_{AC} \times d_{AC}}$ be an unknown bipartite quantum state with $d_A\geq d_C$, and $\eps \in (0,1)$ be a parameter. Then, using single-copy adaptive measurements on $\rho_{AC}$ and $\rho_A\otimes\rho_C$, distinguishing between
    \begin{align}
        \rho_{AC}=\rho_A\otimes\rho_C\quad\text{and}\quad I(A:C)_{\rho}\geq\varepsilon
    \end{align}
with probability at least $0.99$ can be achieved using $\widetilde{O}(\min\{(d_Ad_C)^{3/2}/\varepsilon^2,d_A^{9/4}d_C^{3/4}/\varepsilon\})$ samples.
       
\end{restatable}

Of course, independence testing might also be expressed with respect to different distance measures. The problem has been studied (and resolved) with respect to the trace distance, see \Cref{tab:overview_independence} for an overview and a comparison to the problem in distribution testing.

\begin{table}[h!]
    \centering
    \renewcommand{\arraystretch}{1.3}
    \setlength{\arrayrulewidth}{0.5pt} %
    \begin{tabular}{|c|c|c|}
        \hline
        & \multicolumn{2}{c|}{Independence Testing} \\
        \hline
        & $\ell_1$/Trace distance
        & $I(A:C)$ (i.e., $D_H^2$/Fidelity) \\
        \hline
        
        \rule{0pt}{4ex}\multirow{ 2}{*}{Classical}
        & $\widetilde\Theta\left(
            \max\left\{
                \frac{(d_A d_C)^{1/2}}{\varepsilon^2},
                \frac{d_A^{2/3}d_C^{1/3}}{\varepsilon^{4/3}}
            \right\}
        \right)$
        & $\widetilde\Theta\left(
            \min\left\{
                \frac{d_A^{3/4}d_C^{1/4}}{\varepsilon},
                \frac{d_A^{2/3}d_C^{1/3}}{\varepsilon^{4/3}}
            \right\}
        \right)$ \\
        
        & \cite{DBLP:conf/focs/DiakonikolasK16}
        & \cite{pmlr-v291-seyfried25a} \\
        \hline
        
        \rule{0pt}{4ex}\multirow{ 2}{*}{\makecell{Quantum\\(adaptive)}}
        & $\Theta\left(\frac{(d_A d_C)^{3/2}}{\varepsilon^2}\right)$
        & \cellcolor{yellow!15}$\widetilde{O}\left(
            \min\left\{
                \frac{(d_A d_C)^{3/2}}{\varepsilon^2},
                \frac{d_A^{9/4}d_C^{3/4}}{\varepsilon}
            \right\}
        \right)$ \\
        
        & \cite{yu:LIPIcs.ITCS.2021.11}
        & \cellcolor{yellow!15}[this work] \\
        \hline
    \end{tabular}

    \caption{Comparison between sample complexities for independence testing, without loss of generality assuming $d_A \ge d_C$.}
    \label{tab:overview_independence}
\end{table}

Using a known reduction introduced by Flammia and O'Donnell \cite{Flammia2024quantumchisquared}, we reduce this problem to independence testing in fidelity. The same paper also gives an algorithm for mutual information testing, but in the joint-measurement setting (although their techniques can also be applied to the single-copy setting). Conceptually, they perform a testing-by-learning approach, where they first learn $\rho_A$ and $\rho_C$, resulting in $\hat\rho_A$ and $\hat\rho_C,$ and then perform a certification task between $\rho_{AC}$ and $\hat\rho_A\otimes\hat\rho_C$. Importantly, they incorporate an overhead in the precision learned results in a dimensional scaling which is worse than learning. Our approach uses only partial learning and avoids this overhead. Moreover, for the special case where $d_A$ and $d_C$ are (roughly) of the same size, $d_A=\widetilde{\Theta}(d_C)$, the sample complexities for independence testing give us 
\begin{align}
    \text{trace distance:~}\Theta\left(\frac{d_A^3}{\eps^2}\right),\quad \text{fidelity:~}\widetilde{\Theta}\left(\frac{d_A^3}{\eps}\right).
\end{align}
This is exactly the cost of learning the subsystem, and in general, our algorithm is more efficient than learning the larger subsystem. The optimality for the fidelity problem follows since any improvement would imply a better trace distance algorithm, which was already known to be optimal. In this special case, we further observe the `unconditional' $1/\eps$ improvement we also see in tomography when comparing trace distance and fidelity.

\ssnotes{please change this part}

\color{black}

Now we discuss our lower bound result. We show that equivalence testing with respect to fidelity using single-copy non-adaptive measurements requires $\widetilde{\Omega}(1/\eps^2)$ samples.

\begin{restatable}[Equivalence Testing Lower Bound for Qubits]{theorem}{lbnonadaptive}\label{theo:lbnonadaptive}
Let $\rho, \sigma \in \mathbb{C}^{d \times d}$ be two unknown quantum states and $\eps \in (0,1)$ be a parameter. Then, using single-copy non-adaptive measurements, deciding whether
\begin{align}
        \rho=\sigma
        \quad\text{or}\quad
        F(\rho,\sigma)\leq 1-\varepsilon
\end{align}
with probability at least $0.99$, $\widetilde{\Omega}(1/\eps^2)$ samples of $\rho$ and $\sigma$ are necessary in general, even for qubits.    
\end{restatable}

We formally prove this result in \Cref{sec:lower_equiv_fidelity_qubits}. Combined with our adaptive upper bound for equivalence testing, this indeed demonstrates the power of adaptivity in equivalence testing with respect to fidelity.

\color{black}

\color{blue}

\color{black}

Finally, we also revisit some results from Markov chain testing, where we can use our subroutine for equivalence testing. These results follow a line similar to the work of \cite{DBLP:journals/tit/GaoY25}, who studied the same problems under multi-copy measurements. The main idea is that for certification and equivalence testing, it is sufficient to test marginals of the Markov chain (see \Cref{lemma:cmi_decomposition}).

\begin{restatable}[Markov Chain Testing]{theorem}{lemmarkov}\label{lem:markov}

Let $\rho_{ABC}$ and $\sigma_{ABC}$ be two quantum Markov chains over $H_A \otimes H_B \otimes H_C$. In order to distinguish whether  
\begin{equation}
\rho_{ABC}=\sigma_{ABC}
\quad \text{or}\quad
F(\rho_{ABC},\sigma_{ABC})\leq 1- \eps,
\end{equation}
with probability at least $0.99$, we have the following results:
\begin{enumerate}
    \item[(i)] If the Markov chain $\sigma_{ABC}$ is known and we have sample access to $\rho_{ABC}$, then 
\begin{equation}
\widetilde{O}\left(\max\left\{\frac{d_A^{3/2} d_B^{3/2}}{\eps}, \frac{d_B^{3/2} d_C^{3/2}}{\eps}\right\}\right)
\end{equation}
samples of $\rho_{ABC}$ are sufficient for certification.    
    
    \item[(ii)] If both the Markov chains $\rho_{ABC}$ and $\sigma_{ABC}$ are unknown and we have sample access to them, then 
    \begin{equation}  \widetilde{O}\left(\max\left\{\min\left\{\frac{(d_Ad_B)^{3/2}}{\varepsilon^2},\frac{(d_Ad_B)^{9/4}}{\varepsilon}\right\},\min\left\{\frac{(d_Bd_C)^{3/2}}{\varepsilon^2},\frac{(d_Bd_C)^{9/4}}{\varepsilon}\right\}\right\}\right)
    \end{equation}
    samples are sufficient for equivalence testing.
\end{enumerate}
\end{restatable}

\subsection{Discussion}
In this work we explored how the sample complexities of testing problems change if we switch from decision gaps expressed in trace distance to guarantees in fidelity. The main conclusion is that, unlike for trace distance where certification and equivalence testing can be solved optimally with the same non-adaptive algorithm, differences between them emerge for testing in fidelity. Comparing unknown quantum states appears to be conceptually harder than certification: the sample optimal certification algorithm we provide for testing in fidelity heavily relies on knowledge of $\sigma$, and our equivalence testing algorithm is costlier because it needs to extract some information before it can test. Nevertheless, our equivalence tester can do better than reusing the trace distance algorithm or applying (adaptive) tomography in some regime, and our algorithm for independence testing improves on the testing-by-learning approach studied in the literature \cite{Flammia2024quantumchisquared} for the single-copy setting.

It is not surprising that, as for trace distance, certification does not need adaptivity. For equivalence testing, however, we prove that it can be helpful, which is a separation from the setting in trace distance. Our testing algorithms also provide a natural intuition: our approach for testing by reduction to $\ell_2$ distance requires certain knowledge of the state, and adaptivity allows us to first learn this information, and then test efficiently. One of the main technical challenges of this work was to handle the impact of learning this needed information imperfectly, which required additional care. Surprisingly, while the learning itself is costly, the additional overhead we have to pay to account for the limited precision is only logarithmic.

We see the fact that our approach for independence testing is optimal for the case where $d_A$ and $d_C$ are balanced as evidence that this approach might be optimal more generally, up to logarithmic factors.

Importantly, general lower bounds for non-adaptive algorithms including scaling in the dimension $d$ remain open. It seems natural to conjecture that, analogous to tomography, non-adaptive protocols for equivalence testing in fidelity may not have an advantage over trace distance, such that $\widetilde{\Omega}(d^{3/2}/\eps^2)$ samples are required, and analogously for independence testing. It is also not clear whether our equivalence tester is optimal, or even whether adaptive equivalence testing in fidelity is fundamentally harder than certification. 

While we study certification with respect to a guarantee on the rank of the known state, the sample complexity can be much more fine-grained: the question of instance-optimality has received significant attention in the literature \cite{chen_toward_2021,chen2022tightboundsquantumstate}. We expect that similar approaches might be applied for testing in fidelity as well.

We therefore state the following open questions:
\begin{itemize}
    \item If adaptivity is not allowed, is there an asymptotic advantage in equivalence testing or independence testing in fidelity over trace distance?
    \item Are the presented sample complexities for equivalence testing in fidelity and independence testing optimal, and if not, is (adaptive) equivalence testing generally harder than certification?
    \item Are instance-optimal bounds for certification in fidelity analogous to the existing results for trace distance?
\end{itemize}

\section{Overview of our Results}\label{sec:overview}
In this section, we aim to give an overview of our results and techniques. 
The top-level framework is inspired by an approach from distribution testing, which we refer to as the `\emph{testing in $\ell_2$}'-framework. The main idea is to use an equivalence tester in $\ell_2$ distance as a fundamental subroutine, that is, an algorithm which takes samples from two distributions $P$ and $Q$ and tests whether $P=Q$ or $\|P-Q\|_2\geq\eps$, for some parameter $\eps \in (0,1)$. This testing algorithm can be modified to only compare a subset $S$ of the support of $P$ and $Q$. Various testing problems in distribution testing can be reduced to instances of equivalence testing in $\ell_2$, for example, certification, equivalence testing, independence testing, and even conditional independence testing. This reduction can also handle different distance measures, such as the variation distance or the Hellinger distance, and the approach gives optimal sample complexities (up to poly-logarithmic factors), while simultaneously giving elegant explanations of the subtleties in these testing problems, notably in the various case distinctions which appear %
naturally when trying to bound distance measures such as the variation distance or the squared Hellinger distance tightly by the $\ell_2$ distance.

In the quantum setting, similar ideas have been present in the literature \cite{chen_toward_2021}. Known certification and equivalence testers in the trace distance use an equivalence tester in $\ell_2$-distance as a subroutine, and the notion of splitting the problem into several pieces (a process we refer to as \emph{bucketing}) has been used for instance-optimal state certification in the single-copy setting\ssnotes{cite}. Perhaps surprisingly, certification and equivalence testing in the single-copy setting have not been studied before for guarantees in fidelity. Our goals are twofold:
\begin{itemize}
    \item We aim to understand whether this testing framework, which has proven to be very versatile in distribution testing, is able to solve testing problems for quantum states with single-copy measurements in a sample-optimal way, or if there are more fundamental testing subroutines. In this work, we present, to the best of our knowledge, the first testing algorithms for certification, equivalence testing, and mutual information (i.e., independence) testing with respect to the fidelity.
    
    \item Learning in fidelity is known to benefit from adaptivity~\cite{chen_when_2023}. An impressive line of work~\cite{bubeck_entanglement_2020,chen_toward_2021, chen2022tightboundsquantumstate} has settled that adaptivity does not help for certification in trace distance, and hence also not for equivalence testing, since the optimal certification algorithm does not use the knowledge of $\sigma$. Here we demonstrate that adaptivity does indeed help for certain testing problems in fidelity, by proving a separation between the non-adaptive and adaptive equivalence testing in fidelity.
\end{itemize}

We now discuss our different results. The algorithms and underlying ideas naturally increase in complexity, as we add new building blocks in each step.

\subsection{State Certification in Fidelity}
As mentioned before, we use an $\ell_2$ tester as a fundamental building block. This tester takes samples from $\rho$, and a description of $\sigma$, as well as a projection $\Pi$ and distinguishes between the cases: 
\begin{align}
    \|\Pi(\rho-\sigma)\Pi\|_2^2\geq \eps\quad\text{or}\quad \Pi\rho\Pi=\Pi\sigma\Pi.
\end{align}
The main approach for a reduction from fidelity testing to $\ell_2$ distance is a `bucketing' via the measured $\chi^2$-distance in the following sense: we define a set of $K=O(\log(d/\eps))$ (not necessarily disjoint) projections $\{\Pi_i\}_{i \in [K]}$, and then bound the (in-)fidelity between $\rho$ and $\sigma$ by a sum of $\ell_2$-norms of the distance restricted to $\Pi_i$, weighted by the smallest eigenvalue of $\Pi_i\sigma\Pi_i$, expressed as follows:
\begin{align}
    1-F(\rho,\sigma)\leq \frac{1}{2}D_{\chi^2}(\rho\|\sigma)\leq\sum_{i,j\in [K]}\frac{\|\Pi_{ij}(\rho-\sigma)\Pi_{ij}\|_2^2}{\sigma_{i,\min}+\sigma_{j,\min}},
\end{align}
where $\Pi_{ij}$ projects to the joint support of $\Pi_i$ and $\Pi_j$. When $F(\rho,\sigma)\leq 1-\eps$, using the pigeonhole principle, we know that at least one invocation of our subroutine would output reject with high probability. Thus using our testing subroutine on each term to precision $O(\eps(\sigma_{i,\min}+\sigma_{j,\min})/K^2)$ is sufficient for our purpose. Since $\sigma$ is known, when the rank $r$ of $\sigma$ is small (say $r \ll d$), this approach leads to a sample complexity of $\widetilde{\Theta}(r^{3/2}/\eps)$, which is completely independent of $d$.
The optimality of our result directly follows from the lower bound result on trace distance, and adaptivity is not needed to achieve an optimal sample complexity, up to polylogarithmic factors.

\subsection{Equivalence Testing in Fidelity}
\label{sec:overview_equiv_fidelity}
Unlike certification, in equivalence testing, we are given sample access to two unknown states $\rho$ and $\sigma$ and are interested in distinguishing whether $\rho=\sigma$, or $F(\rho,\sigma)\leq 1-\varepsilon$ with high probability. Ideally, we would like to reuse our approach for certification, and reduce the problem to comparing the states `piecewise' in $\ell_2$ distance. The additional challenge compared to before is that we no longer know how to perform the bucketing, since both $\rho$ and $\sigma$ are unknown. One potential approach could be to first learn the state $\sigma$ and then follow this route. However, learning the full state would incur the cost of tomography, which is higher. As a result, we settle for a trade-off, which is also used in distribution testing: we first learn the state ``approximately'' such that we can bucket all eigenvalues down to a precision $\eta$ for some suitable threshold $\eta$, thereby allowing for a bucketing of all large eigenvalues above the threshold $\eta$. The eigenvalues smaller than $\eta$ will be tested together in one (different) bucket. However, this will make the testing part more expensive: the sample complexity of the $\ell_2$ testing algorithm scales with the maximal eigenvalue $\sigma_{\max}$, of a bucket, the precision to which we need to test with the minimal eigenvalue $\sigma_{\min}$. For certification, these were of the same order, so we did not need to perform additional steps. However, here we need to balance the trade-off between the costs of these two tasks (approximately learning and testing), which will define our threshold $\eta$ and the final sample complexity. The learning subroutine comes from the tomography in fidelity result by \cite{chen_when_2023}, and provides a guarantee in operator norm. This is exactly the kind of precision we require to form our buckets. However, a suitable definition and a careful analysis of the buckets are needed to show that we do not mistakenly incorporate too much of the wrong subspaces in our buckets. The bulk of the proof is used to rigorously capture this.

It is crucial to note that the algorithm for equivalence testing in trace distance follows directly from the (non-instance-optimal) certification algorithm, since this algorithm doesn't use any knowledge of $\sigma$! In the case of fidelity, knowledge of $\sigma$ is essential for the optimal algorithm. Without this knowledge, the algorithmic challenges become considerably more involved, and it might be unavoidable that performing equivalence testing in fidelity is more expensive than certification. The best known prior sample complexity was 
\begin{align}
    O\left(\min\left\{\frac{d^{3/2}}{\eps^2},\frac{d^3}{\eps}\right\}\right),
\end{align}
i.e., either performing the trace distance algorithm, or performing (adaptive) full state tomography. Our algorithm also splits into two regimes, and essentially reduces the second term by a factor of $d^{3/4}$. We want to stress again that our case distinction naturally arises from our approach, and is not the result of combining two different algorithms.
The iterative process uses a logarithmic number of adaptive phases. It is natural to ask whether, similar to the tomography task, adaptivity is necessary here, and whether it is crucial to beating equivalence testing in trace distance. As pointed out earlier, the answer to the first question is affirmative. %

\subsection{Non-adaptive Lower Bounds for Equivalence Testing in Qubits}
Here we show that even for qubits, equivalence testing without adaptivity is weaker compared to adaptivity. We generalize the standard approach of comparing two sets of quantum states \cite{bubeck_entanglement_2020}, and show that distinguishing them is impossible unless the number of samples $N$ has a certain size. %

The hard instances are natural generalizations of hard instances in distribution testing: the eigenvalues for states in one set are $[1-\eps,\eps]$, and $[1-2\eps,2\eps]$ for the other. We then pick a random unitary to make the basis random. For distributions, seeing samples from the `light' area allows us to separate the two cases: seeing such a sample takes roughly $\widetilde{O}(1/\eps)$ samples. If we have quantum states and non-adaptive measurements, an additional challenge arises: a measurement basis which is poorly aligned with the eigenstates of the state will have an overlap between both eigenvalues, and smear out the heavy eigenvalue, which means that most measurements reveal only very little information. For example, we expect only an $\widetilde{O}(\eps)$ fraction of the measurements to be well-aligned with the unknown eigenbasis, which intuitively explains the sample complexity of $\widetilde{\Omega}(1/\eps^2)$.

\ssnotes{please update this part if possible}

\subsection{Mutual Information Testing}
As a natural application of our equivalence testing algorithm, we study the problem of mutual information testing. Due to the aforementioned reduction \cite{Flammia2024quantumchisquared}, this problem reduces to independence testing. On a technical level, we perform independence testing with a guarantee in fidelity as follows: 
\begin{align}
    F(\rho_{AC},\rho_A\otimes \rho_C)\leq 1-\widetilde{\Omega}(\eps)
    \quad\text{or}\quad
    \rho_{AC}=\rho_A\otimes \rho_C,
\end{align}
and then we relate this to mutual information testing by incurring only logarithmic factors. The primary idea is to leverage the approach from \cite{pmlr-v291-seyfried25a}, which solved the problem for distribution testing. In the classical setting, as a first step, the independence testing problem can be reduced to equivalence testing, which is direct if we assume that we are able to process two samples of $\rho_{AC}$ into the product state $\rho_{A}\otimes\rho_{C}$ using partial traces: independence testing then reduces to an instance of equivalence testing between $\rho_{AC}$ and $\rho_{A}\otimes\rho_{C}$. However, we can go further and ask whether we can benefit from the additional structure in our problem, since we know that one of our states has product structure. The answer to this question is yes indeed, as used by \cite{pmlr-v291-seyfried25a}: when learning the buckets approximately, we learn them individually for $\rho_A$ and $\rho_C$, and compose them to obtain the bucketing for $\rho_A\otimes\rho_C$. This proves to be more efficient than simply bucketing $\rho_A\otimes\rho_C$ as a whole. This leads to the final sample complexity of $\widetilde{O}(\min\{(d_Ad_C)^{3/2}/\varepsilon^2,d_A^{9/4}d_C^{3/4}/\varepsilon\})$.

\section{Preliminaries}
\label{sec:prelim}

For an integer $n$, $[n]$ denotes the set $\{1, \ldots, n\}$, and $[n]_0$ denotes the set $\{0,1, \ldots, n\}$. For integers $i,j$, $\delta_{ij}$ denotes the Kronecker delta function and $\delta_{ij}=1$ if and only if $i=j$ and $0$ otherwise. For a matrix $X$, we denote $\|X\|_{\opn}$ as its operator norm. For concise expressions and readability, we use the asymptotic notations of $\widetilde{O}(\cdot), \widetilde{\Omega}(\cdot)$, and $\widetilde{\Theta}(\cdot)$, where we hide polylogarithmic dependencies on the parameters. We assume that the reader has some familiarity with quantum computing and refer interested readers to the textbook \cite{nielsen_chuang} for further background.

Let us start by formalizing the notion of measurement.

\begin{definition}[Positive operator valued measurement (POVM)]
A POVM $\mathcal{M}$ is a (possibly continuous) set of positive operators $\mathcal{M} = \{M_z\}$ satisfying $\sum_z M_z = \mathbf{1}_d$. Measuring a quantum state $\rho$ using $\mathcal{M}$ means sampling from the distribution over the labels, where we observe $z$ with probability $\tr[\rho M_z]$.
\end{definition}

Now let us define various distance and divergence measures which will be used in this work.

\begin{definition}[Various distance and divergence measures]
Let $\rho, \sigma \in \mathbb{C}^{d \times d}$ be two quantum states. Then we define the following distances between $\rho$ and $\sigma$:
\begin{itemize}
    \item[(i)] Trace distance: %
    $\|\rho-\sigma\|_{\tr}:=\frac{1}{2}\|\rho-\sigma\|_1$.

    \item[(ii)] Hilbert-Schmidt distance: %
    $\|\rho-\sigma\|_{\textnormal{HS}}:=\|\rho-\sigma\|_2$.

    \item[(iii)] Fidelity: %
    $F(\rho,\sigma) :=\|\sqrt{\rho}\sqrt{\sigma}\|_1$.

    \item[(iv)] Measured $\chi^2$-Divergence: %
    $D_{\chi^2}(\rho \|\sigma) :={\displaystyle 2\sum_{i,j}\frac{|\bra{\sigma_i}\rho\ket{\sigma_j}-\delta_{ij}\sigma_{i}|^2}{\sigma_i+\sigma_j}}$ where $\sigma= \sum_i \sigma_i \dyad{\sigma_i}$.

    \item[(v)] Bures distance: %
    $D_B(\rho, \sigma):= \sqrt{2 (1- F(\rho, \sigma))}$.

\end{itemize}

\end{definition}

\color{blue}

\color{black}

\begin{lemma}[Fuchs-van de Graaf Inequalities~\cite{fuchs1999cryptographic}]\label{lem:fucsvandegraff}
For two quantum states $\rho, \sigma \in \mathbb{C}^{d \times d}$, the following relations hold:
$$1-F(\rho,\sigma)\leq \frac{1}{2}\|\rho-\sigma\|_1\leq \sqrt{1-F(\rho,\sigma)^2}.$$    
\end{lemma}

We will use several crucial relations between these distance measures.

\begin{lemma}[{\cite[Prop.\ 2.31]{Flammia2024quantumchisquared}}]
\label{lemma:fid_chi2_bound}
Let $\rho, \sigma \in \mathbb{C}^{d\times d}$ be two quantum states. Then the following holds:
    \begin{align}
        1-F(\rho,\sigma)\leq \frac{1}{2}D_{\chi^2}(\rho\|\sigma).
    \end{align}
\end{lemma}

\begin{remark}
We note that the Schatten norms satisfy the Hölder inequality, i.e.,
\begin{equation}
    \forall p,q,r\in [1,\infty]: \frac{1}{p}+\frac{1}{q}=\frac{1}{r}
    \quad\implies\quad 
    \forall A,B\in \mathbb{C}^{d \times d} : \|AB\|_r\leq \|A\|_p\|B\|_q.
\end{equation}    
\end{remark}

\begin{definition}
The {\em squared Hellinger} distance between two distributions $P$ and $Q$ over a set $\mathcal{D}$ is defined as:
\begin{equation}
    D_H^2(P,Q)= \frac{1}{2} \sum_{x \in \mathcal{D}} \left(\sqrt{P(x)}- \sqrt{Q(x)}\right)^2.
\end{equation}    
\end{definition}

\begin{lemma}
    \label{lemma:l2_proj_sum}
        Let $\sigma \in \mathbb{C}^{d \times d}$ be a quantum state of dimension $d$ with eigenvectors $\{\ket{\sigma_i}\}_{i=1}^d$ and eigenvalues $\{\sigma_i\}_{i=1}^d$, and $S, R\subseteq[d]$. Define $\Pi_S=\sum_{i\in S}\dyad{\sigma_i}$ and $\Pi_R$ analogously. Then the following holds for any quantum state $\rho \in \mathbb{C}^{d \times d}$: 
        \begin{equation}
            \|\Pi_S(\rho-\sigma)\Pi_R\|_2^2=\sum_{i\in S, j \in R}|\bra{\sigma_i}\rho\ket{\sigma_j}-\delta_{ij}\sigma_{i}|^2.
        \end{equation}
    \end{lemma}
    \begin{proof}
        We calculate the following:
        \begin{align}\|\Pi_S(\rho-\sigma)\Pi_R\|_2^2&=\tr\left[\Pi_S(\rho-\sigma)\Pi_R(\rho-\sigma)\Pi_S\right]
            \\
            &=\sum_{i\in S, j \in R}\bra{\sigma_i}(\rho-\sigma)\dyad{\sigma_j}(\rho-\sigma)\ket{\sigma_i}
            \\
            &=\sum_{i\in S, j \in R, i\neq j}\bra{\sigma_i}\rho\dyad{\sigma_j}\rho\ket{\sigma_i}+\sum_{i \in S\cap R}\bra{\sigma_i}(\rho-\sigma)\dyad{\sigma_i}(\rho-\sigma)\ket{\sigma_i}
            \\
            &=\sum_{i\in S, j \in R, i\neq j}|\bra{\sigma_i}\rho\ket{\sigma_j}|^2+\sum_{i \in S\cap R}(\bra{\sigma_i}\rho\ket{\sigma_i}-\sigma_i)^2
            \\
            &=\sum_{i\in S,j\in R}|\bra{\sigma_i}\rho\ket{\sigma_j}-\delta_{ij}\sigma_{i}|^2.\qedhere
        \end{align}
    \end{proof}

\begin{remark}[Minimal eigenvalues]
\label{rem:min_ev}
    We can reduce certification (identity testing) in fidelity (parameter $\varepsilon$) to fidelity testing with eigenvalues in $\widetilde{\Omega}(\varepsilon/d)$ with respect to parameter $\varepsilon/2$.

\end{remark}

\begin{proof}
    This is a standard technique, we follow the argument used in \cite[Corollary 6.17]{buadescu2019quantum}. There, the Bures distance $D_B(\rho,\sigma)=\sqrt{2(1-F(\rho,\sigma))}$, a metric which also satisfies $D_B^2(\rho,\sigma)\leq 2\|\rho-\sigma\|_1$, is used with a depolarization channel. The guarantee $1-\varepsilon\geq F(\rho,\sigma)$ translates to $2\varepsilon\leq D_B^2(\rho,\sigma)$ such that (let $\tau=\mathbf{1}_d/d$)
    \begin{align}
        \sqrt{2\varepsilon}\leq D_B(\rho,\sigma)&\leq D_B((1-\nu)\rho +\nu \tau,(1-\nu)\sigma+\nu \tau)+D_B((1-\nu)\sigma+\nu \tau,\sigma)+D_B((1-\nu)\rho +\nu \tau,\rho)
        \\
        &\leq D_B((1-\nu)\rho +\nu \tau,(1-\nu)\sigma+\nu \tau)+\sqrt{2\nu\|\tau-\sigma\|_1}+\sqrt{2\nu\|\tau-\rho\|_1}
        \\
        &\leq D_B((1-\nu)\rho +\nu \tau,(1-\nu)\sigma+\nu \tau)+4\sqrt{\nu}.
    \end{align}
    Thus, as long as $((\sqrt{2}-1)/4)^2\varepsilon\geq \nu$, we still have 
    \begin{align}
        F((1-\nu)\rho +\nu \tau,(1-\nu)\sigma+\nu \tau)\leq 1-\varepsilon/2,
    \end{align}
    with minimal eigenvalues $\nu\tau_{\min}$ (which for the depolarization channel means $\nu/d$).
\end{proof}

An integral part of our proof is the following result.
\begin{lemma}[{\cite[Lemma 6.4]{chen_toward_2021}}]
\label{prel:lemma:exp_z_to_l2}
    For any Hermitian $M\in \mathbb{C}^{d\times d}$ and Haar-random $U\in U(d)$, we have 
    \begin{equation}
        \mathbb{E}\left[\tr\left[\sum_{i=1}^d\left(U_i^{\dagger}MU_i\right)^2\right]\right]=\frac{\tr[M]^2+\|M\|_2^2}{d+1}.
    \end{equation}
    In addition, if $\tr[M]=0$, then 
    \begin{equation}
        \label{eq_lemma_weingarten_moment_two}
        \mathbb{E}\left[\tr\left[\sum_{i=1}^d\left(U_i^{\dagger}MU_i\right)^2\right]^2\right]\leq\frac{(1+o(1))\|M\|_2^4}{d^2},
    \end{equation}
    where the $o(1)$ decays with $O(1/d)$.
\end{lemma}

\begin{remark}
The factor $o(1)$ in \Cref{prel:lemma:exp_z_to_l2} can be bounded by $O(1/d)$. This follows directly from the proof in \cite[Lemma 6.4]{chen_toward_2021}, where, $Z:=\sum_{i=1}^d\left(U_i^{\dagger}MU_i\right)^2$ is bounded by
    \begin{align}
\mathbb{E}\left[\tr\left[Z\right]^2\right]&\leq d\frac{O(\|M\|_{2}^4)}{d(d+1)(d+2)(d+3)}+(d^2-d)\left(\frac{(d^4-8d^2+6)\|M\|_2^4}{d^8-14d^6+49d^4-36d^2}+\frac{O(\|M\|_2^4)}{d^5}\right)
        \\
        &\leq \frac{\|M\|_2^4}{d^2}+\frac{O(\|M\|_2^4)}{d^3}.
    \end{align}
\end{remark}

We will also need the following result from classical equivalence testing, which allows to compare two unknown probability distributions on a given subset $S$. Given a distribution $P$ over $[d]$ and a set $S \subseteq [d]$, $\subD{P}{S}{}$ denotes the (pseudo)-distribution from $P$ where, $\subD{P}{S}{} (i) = P(i)$, for $i\in S$ and $\subD{P}{S}{} (i) =0$ otherwise.
\begin{lemma}[{\cite[Lemma 5.6]{pmlr-v291-seyfried25a}}]
\label{prelim:lemma:equivalence_l2}
Let $P$ and $Q$ be two unknown distributions over $[d]$, and $S\subseteq [d]$. Moreover, let us assume that $\|\subD{Q}{S}{}\|_2 \leq b$ for some $b\in \mathbb{R}$. Given sample access to $P$ and $Q$, in order to distinguish if $\subD{P}{S}{}=\subD{Q}{S}{}$
or $\|\subD{P}{S}{}-\subD{Q}{S}{}\|_2 \geq \eps$ with probability at least $0.999$, $\cHSdistr\max\{\frac{b}{\eps^2}, \frac{1}{\eps},\sqrt{d}\}$ samples are sufficient, where $P^S$ and $Q^S$ denote the (pseudo)-{distribution} induced on the set $S$ and $\cHSdistr$ is a fixed suitable constant.
\end{lemma}

\subsection{Subroutine: Testing Equivalence on Subspaces in $\ell_2$-Distance}
\label{sec:prelim_test_ltwo}
We begin our discussion by generalizing the identity tester in $\ell_2$ distance discussed in \cite{chen_toward_2021} to compare states reduced to a subspace. They do this implicitly in the proof of their Lemma 6.10.

\begin{lemma}
\label{idenl2:lemma:equiv_proj}
Let $\rho, \sigma \in \mathbb{C}^{d \times d}$ be two unknown quantum states, $\Pi$ be a known projection mapping to a subspace of dimension $d_{\Pi}$ such that $\| \Pi\sigma\Pi \|_{\opn} \leq \sigma_{\Pi}^{\max}$ and $\eps \in (0,1)$ be a parameter. Given sample access to $\rho$ and $\sigma$, in order to distinguish between $\Pi\rho\Pi=\Pi\sigma\Pi$ and $\|\Pi(\rho-\sigma)\Pi\|_2^2\geq \varepsilon$ with probability at least $0.99$,  \begin{equation}   N=\cHSquantum\max\left\{\frac{\sqrt{d_{\Pi}}\tr[\Pi\sigma\Pi]}{\varepsilon},\sqrt{\frac{d_{\Pi}}{\varepsilon}}\right\}\leq \cHSquantum\max\left\{\frac{\sqrt{d_{\Pi}}\min\{d_{\Pi}\sigma_{\Pi}^{\max},1\}}{\varepsilon},\sqrt{\frac{d_{\Pi}}{\varepsilon}}\right\}
    \end{equation}
samples of $\rho$ and $\sigma$ each, where $\cHSquantum$ is a suitable universal constant. %

\end{lemma}

Note that using a majority vote over $O(\log(1/\delta))$ repetitions amplifies the probability of success to $1-0.01\delta$.

\begin{proof}
First, note that we may without loss of generality assume that $d_{\Pi}\geq d_0$, for $d_0$ a suitably chosen constant, which will be necessary to show that certain random variables are close to their expectation value. For a discussion, see \Cref{remark:min_dim} below.
    
We now define a POVM by selecting a random unitary $U$ on the support of $\Pi$, and obtain the POVM from the eigenvectors of $U$, $\{\ket{U}_i\}_{i=1}^d$, via $\{\dyad{U_1}, \ldots,\dyad{U_t}, \mathds{1}-\Pi\}$, where $t:=d_{\Pi}$. We can now measure $\rho$ and $\sigma$ using this POVM to obtain a classical probability distribution $P$ over $[d_{\Pi}+1]$: for $i\in S:=[d_{\Pi}]$, $P(i):=\bra{U_i}\rho\ket{U_i}=\bra{U_i}\Pi\rho\Pi\ket{U_i}$ and $P(d_{\Pi}+1)=\mathds{1}-\tr[\Pi \rho]$. Analogously, we obtain another classical probability distribution $Q$ from $\sigma$.
    
We now want to test for equivalence on the subset $S$ defined by $\{U_i\}_{i=1}^d$. Let us define a random variable $Z$ as follows:
$$Z=\tr[\sum_i(U_i^{\dagger}MU_i)^2]=\|\subD{P}{S}{}-\subD{Q}{S}{}\|_2^2.$$ 

Applying \Cref{prel:lemma:exp_z_to_l2} to $M=\Pi(\rho-\sigma)\Pi$, gives us
    \begin{equation}
    \label{eq:equiv_l2_distr_dist}
        \E[Z]=\E[\|\subD{P}{S}{}-\subD{Q}{S}{}\|_2^2]=\frac{\tr[\Pi(\rho-\sigma)]^2+\|\Pi(\rho-\sigma)\Pi\|_2^2}{d_{\Pi}+1}.
    \end{equation}

Note that when $\Pi\rho\Pi=\Pi\sigma\Pi$, then the above \Cref{eq:equiv_l2_distr_dist} is zero. On the other hand, when $\|\Pi(\rho-\sigma)\Pi\|_2^2\geq \varepsilon$, \Cref{eq:equiv_l2_distr_dist} is lower bounded by $\varepsilon/(d_{\Pi}+1)$.

It is easy to show that this holds not only in expectation, but with high probability (see \Cref{lemma:helper_variance_z}), up to a constant factor.
    Further, with the choice $M=\Pi\sigma\Pi$, from \Cref{prel:lemma:exp_z_to_l2}, we also get that
    \begin{align}
        \E[\|\subD{Q}{S}{}\|_2^2]&= \frac{\tr[\Pi\sigma\Pi]^2+\|\Pi\sigma\Pi\|_2^2}{d_{\Pi}+1}\leq \frac{2\tr[\Pi\sigma\Pi]^2}{d_{\Pi}+1}.
    \end{align}
    Using Markov's inequality, we have $\|\subD{Q}{S}{}\|_2^2
        \leq 2000 \tr[\Pi\sigma\Pi]^2/d_{\Pi}$ with probability $0.999$.

Now we will apply the equivalence tester (\Cref{prelim:lemma:equivalence_l2}), which decides whether $P^S=Q^S$ or $\|P^S-Q^S\|_2^2\geq \eta$ using $\cHSdistr\max\{\|Q^S\|_2/\eta, \sqrt{1/\eta},\sqrt{d_{\Pi}}\}$ samples. For $\cHSquantum$ a suitably chosen universal constant, %
we get the following sample complexity (note that $\eta=\varepsilon/(d_{\Pi}+1)$ by \eqref{eq:equiv_l2_distr_dist}):
    \begin{align}
        N=\cHSdistr\max\left\{\frac{\|Q^S\|_2}{\eta}, \sqrt{\frac{1}{\eta}},\sqrt{d_{\Pi}}\right\}
        &\leq
        \cHSquantum\max\left\{\frac{\sqrt{d_{\Pi}}\tr[\Pi\sigma\Pi]}{\varepsilon},\sqrt{\frac{d_{\Pi}}{\varepsilon}}\right\} 
        \\
        &\leq
        \cHSquantum\max\left\{\frac{\sqrt{d_{\Pi}}\min\{d_{\Pi}\sigma_{\Pi}^{\max},1\}}{\varepsilon},\sqrt{\frac{d_{\Pi}}{\varepsilon}}\right\}.
    \end{align}
Using a union bound, our bound on $\|Q^S\|_2$ and the comparison of $P^S$ and $Q^S$ succeed with probability at least $0.99$.
\end{proof}

\begin{remark}[Minimal dimension]
\label{remark:min_dim}
    If $d_{\Pi}<d_0$, we perform a preprocessing in which we add a helper system $X$ of size $|X|=\lceil d_0/d_{\Pi}\rceil$ and consider states $\tilde\rho=\rho\otimes 1_X/d_X$, $\tilde\sigma=\sigma\otimes 1_X/d_X$. We let $\tilde\Pi=\Pi\otimes 1_X$. Clearly, $\|\tilde\Pi(\tilde\rho-\tilde\sigma)\tilde\Pi\|_2^2= \|\Pi(\rho-\sigma)\Pi\|_2^2/d_X$ and $d_{\tilde\Pi}\geq d_0$, such that we only need to increase the precision by a constant factor of $d_X$. We then pick the random unitaries over the joint system. For practical use, we do not need an additional system, since the POVM $\{\tr_X[\dyad{U_i}]/d_X\}_{i=1}^d\cup\{1-\Pi\}$ follows the same statistics, e.g., $\bra{U_i}(\tilde\rho-\tilde\sigma)\ket{U_i}=\tr\left[(\rho-\sigma)\tr_X[\dyad{U_i}/d_X]\right]$.
\end{remark}

\begin{lemma}
\label{lemma:helper_variance_z}
Let $U$ be a random unitary on a subspace characterized by a projection $\Pi$ of dimension $d_{\Pi}\geq d_0$, and denote the eigenvectors of $U$ by $\{\ket{U_i}\}_{i=1}^d$. Further, suppose that $\rho$ and $\sigma \in \mathbb{C}^{d \times d}$ are two quantum states such that $\|\Pi(\rho-\sigma)\Pi\|_2^2\geq\varepsilon$. Then, with probability at least $0.999$,
\begin{align}
\label{eq:l2test_whp_bound}
\tr[\sum_i(U_i^{\dagger}(\rho-\sigma) U_i)^2]\geq \frac{\varepsilon}{2(d_{\Pi}+1)}.
\end{align}
\end{lemma}

\begin{proof}
    We are interested in computing the variance of $\sum_i\bra{U_i}\tilde M\ket{U_i}$. However, we cannot use \eqref{eq_lemma_weingarten_moment_two} in \Cref{prel:lemma:exp_z_to_l2} directly, since the trace might not be zero. Instead, we define 
    \begin{align}
        \tilde M:=\Pi(\rho-\sigma)\Pi-\mathds{1}_{\Pi}\frac{\tr[\Pi(\rho-\sigma)\Pi]}{d_{\Pi}}=M-\mathds{1}_{\Pi}\frac{\tr[M]}{d_{\Pi}}.
    \end{align} Then $\tr[\tilde M]=0$. Now let us define the random variable $\tilde Z$ as follows:
    \begin{align}
        \tilde Z:=\tr[\sum_i(U_i^{\dagger}\tilde MU_i)^2]=\sum_i\bra{U_i}\tilde M\ket{U_i}^2.
    \end{align}
Following \Cref{prel:lemma:exp_z_to_l2}, we can say that
\begin{align}   
\textnormal{Var}[\tilde Z]=\E\left[\left(\sum_i\bra{U_i}\tilde M\ket{U_i}^2\right)^2\right]-\E\left[\left(\sum_i\bra{U_i}\tilde M\ket{U_i}^2\right)\right]^2\leq \frac{(1+o(1))\|\tilde M\|_2^4}{d_{\Pi}^2}-\frac{\|\tilde M\|_2^4}{(d_{\Pi}+1)^2}.
\end{align}

Now we have:
$$Z=\tr[\sum_i(U_i^{\dagger}MU_i)^2]=\|\subD{P}{S}{}-\subD{Q}{S}{}\|_2^2.$$ 
Note that $Z$ and $\tilde Z$ are simply shifted by a constant independent of $\{U_i\}$, 
\begin{equation}
\tilde Z=\tr[\sum_i(U_i^{\dagger}\tilde MU_i)^2]=\tr[\sum_i(U_i^{\dagger} MU_i)^2]-\frac{\tr[M]^2}{d_{\Pi}}=Z-\frac{\tr[M]^2}{d_{\Pi}}.
\end{equation}

Thus 
\begin{align}
\textnormal{Var}[Z]=\textnormal{Var}[\tilde Z]\leq o(1)\frac{\|\tilde M\|_2^4}{d_{\Pi}^2}&=o(1)\frac{\left(\tr[(M-1_{\Pi}\tr[M]/d_{\Pi})^2]\right)^2}{d_{\Pi}^2}
\\
&=o(1)\frac{(\|M\|_2^2-\tr[M]^2/d_{\Pi})^2}{d_{\Pi}^2}\leq o(1)\frac{\|M\|_2^4}{d_{\Pi}^2},
\end{align}
where we used Jensen's inequality in the last step ($d_{\Pi}\|M\|_2^2\geq \tr[M]^2$). By applying Chebyshev's inequality, we can say that
\begin{align}
\Pr[|Z-\E[Z]|\geq \E[Z]/2]\leq \frac{\textnormal{Var}[Z]}{(1/2)^2\E[Z]^2}\leq 4O\left(\frac{1}{d_{\Pi}}\right)\left(\frac{d_{\Pi}+1}{d_{\Pi}}\right)^2\leq O\left(\frac{1}{d_0}\right).
\end{align}
The lemma follows assuming $d_0$ is sufficiently larger than the implicit constant.
\end{proof}

\section{Certification in Fidelity}
\label{sec:id_testing}

In this section, we prove our main result of certification in fidelity:

\fididtesting*

\color{black}

We will prove the above theorem in several steps. In a first step (\Cref{lemma:fid_id_testing_general}), we give an algorithm which does not use knowledge of the rank of $\sigma$, i.e., assumes full-rank. We next use this as a subroutine in an algorithm that has a guarantee on the rank of $\sigma$ in \Cref{lemma:certification_rank_dependence}. Finally, in \Cref{lemma:certification_lower_bounds}, we discuss the lower bound.

Let $k$ be a constant to be defined. Then we partition the eigenvalues of $\sigma$ into $k+1$ buckets $\{S_i\}_{i=0}^{k}$, such that
    \begin{equation}
    \label{eq:id_buckets}
        0\leq i<k:~ S_i:=\{j |2^{-i-1}< \sigma_j\leq 2^{-i}\},\qquad S_k:=\{j| \sigma_j\leq 2^{-k}\}.
    \end{equation}
    For readability, we introduce $\bv_i^{\min}$ and $\bv_i^{\max}$ as bounds on the minimal and maximal values in $S_i$ (which, for all buckets except the last bucket, gives $\bv_i^{\min}=2^{-i-1}$ and $\bv_i^{\max}=2^{-i}$) ~\footnote{While here they are the same up to a factor of two, this will no longer hold true in the case of equivalence testing, which is why we emphasize the distinction between $\bv_i^{\min}$ and $\bv_i^{\max}$ here.}. We then define projectors  projecting to the corresponding subspaces as follows:
    \begin{equation}
        \Pi_{S_i}:=\sum_{j\in S_i}\dyad{\sigma_j}.
    \end{equation}

Let us start with our dimension dependent upper bound for fidelity certification.
    
\begin{lemma}[Certification, Dimensional Dependence]
    \label{lemma:fid_id_testing_general}
Let $\rho, \sigma \in \mathbb{C}^{d \times d}$ be quantum states and $\eps \in (0,1)$ be a parameter. We assume that $\sigma$ is known while $\rho$ is accessible via samples. Then, using non-adaptive single-copy measurements, deciding whether
\begin{align}
        \rho=\sigma
        \quad\text{or}\quad
        F(\rho,\sigma)\leq 1-\varepsilon
\end{align}
with probability at least $0.99$, can be done using $\widetilde{O}(d^{3/2}/\varepsilon)$ samples of $\rho$.

\end{lemma}

\begin{proof}

    In the following we use the aforementioned definitions for $S_i$ (see \eqref{eq:id_buckets}) and let $k=2\lceil\log(d/\varepsilon)\rceil$. Recall that by \Cref{rem:min_ev}, we may assume that the eigenvalues of $\rho$ and $\sigma$ are of order $\widetilde{\Omega}(\eps/d)$. We may also assume that for a fixed constant $c$, $S_u$, $\bv_u^{\max}\leq c\bv_u^{\min}$. Note that for $u\neq v$, $\Pi_{S_u\cup S_v}=\Pi_{S_u}+\Pi_{S_v}$. Define 
    \begin{equation}
        \mathcal{S}:=\{S_{uv}:=S_u\cup S_v\}_{u,v},
    \end{equation}
    where $|\mathcal{S}|\leq (k+1)^2$. We may then rewrite, using \Cref{lemma:fid_chi2_bound} and \Cref{lemma:l2_proj_sum} (where $R=S$)
\begin{align}
    \label{eq:bound_f_by_l2}
        1-F(\rho,\sigma)%
        \leq \frac{1}{2}D_{\chi^2}(\rho\|\sigma)&=\sum_{i,j\in [d]}\frac{|\bra{\sigma_i}\rho\ket{\sigma_j}-\delta_{ij}\sigma_{i}|^2}{\sigma_i+\sigma_j}
        \\
        &\leq \sum_{u, v}\sum_{(i,j)\in S_u\times S_v}\frac{|\bra{\sigma_i}\rho\ket{\sigma_j}-\delta_{ij}\sigma_{i}|^2}{\sigma_i+\sigma_j}
        \\
        &\leq \sum_{u, v}\frac{\sum_{(i,j)\in S_u\times S_v}|\bra{\sigma_i}\rho\ket{\sigma_j}-\delta_{ij}\sigma_{i}|^2}{\bv_u^{\min}+\bv_v^{\min}}
        \\
        &\leq \sum_{u, v}\frac{\sum_{(i,j)\in S_{u}\cup S_{v}\times S_{u}\cup S_{v}}|\bra{\sigma_i}\rho\ket{\sigma_j}-\delta_{ij}\sigma_{i}|^2}{\bv_u^{\min}+\bv_v^{\min}}
        \\
        \textnormal{(\Cref{lemma:l2_proj_sum})}\quad &\leq \sum_{u, v}\frac{\|\Pi_{S_{u}\cup S_{v}}(\rho-\sigma)\Pi_{S_{u}\cup S_{v}}\|_2^2}{\bv_u^{\min}+\bv_v^{\min}}. \label{eq:fidelitycertification}
    \end{align}

Thus, from \Cref{eq:fidelitycertification}, we can say that  if $1-F(\rho,\sigma)\geq \varepsilon$, then there exists an $S_{uv}\in\mathcal{S}$ for which we have the following:
    \begin{equation}
    \label{eq:identity:sc_formula}
        \|\Pi_{S_{uv}}(\rho-\sigma)\Pi_{S_{uv}}\|_2^2\geq \frac{\varepsilon(\bv_u^{\min}+\bv_v^{\min})}{(k+1)^2}=:\nu_{uv}.
    \end{equation}
    Using \Cref{idenl2:lemma:equiv_proj} we can test each $\|\Pi_{S_{uv}}(\rho-\sigma)\Pi_{S_{uv}}\|_2^2$ using %
    \begin{align}
       N_{S_{uv}}&=\cHSquantum\max\left\{\frac{\sqrt{|S_{uv}|}\min\{|S_{uv}|(\bv_{u}^{\max}+\bv_{v}^{\max}),1\}}{\nu_{uv}},\sqrt{\frac{|S_{uv}|}{\nu_{uv}}}\right\}
       \\
       &\leq 2\cHSquantum k^2c\max\left\{\frac{|S_{uv}|^{3/2}}{\varepsilon},\sqrt{\frac{|S_{uv}|}{\varepsilon (\bv_{u}^{\min}+\bv_{v}^{\min})}}\right\}%
    \end{align}
    samples, using $\bv_u^{\max}+\bv_v^{\max}\leq c(\bv_u^{\min}+\bv_v^{\min})$ for all buckets. Note that as we know $\sigma$ and the measurement, we can simply simulate sampling from $\sigma$. 
    
    The probability of success in the above is $0.99$. To ensure a high probability of success overall, we repeat each test $c_{\textnormal{rep}}\log(k^2)$ times and perform a majority vote, which boosts the success probability to at least $1-0.001/(k+1)^2$. Using the fact that $\sigma_{\Pi_S}^{\min}\geq \varepsilon/d$, the total sample complexity is:
    \begin{align}
        N&=c_{\textnormal{rep}}\log(k^2)\sum_{S_{uv}\in\mathcal{S}}2\cHSquantum k^2\max\left\{\frac{|S_{uv}|^{3/2}}{\varepsilon},\sqrt{\frac{|S_{uv}|}{\varepsilon (\bv_{u}^{\min}+\bv_{v}^{\min})}}\right\}
        \\
        &\leq 2\cHSquantum c_{\textnormal{rep}}\log(k^2)k^4\frac{d^{3/2}}{\varepsilon}%
        \leq \widetilde{O}\left(\frac{d^{3/2}}{\varepsilon}\right).
    \end{align}
    
We will perform this for all $O(\log^2(d/\varepsilon))$ buckets, which allows to test whether $F(\rho,\sigma)=1$ or $\leq 1-\widetilde{O}(\varepsilon)$, with a probability of success which is bounded by a union bound to $1-(k+1)^2\cdot 0.01/(k+1)^2=0.99$.
\end{proof}

\subsection{Certification in Fidelity with Rank Dependence}

Now we move on to proving our rank dependent certification result:

\begin{lemma}[Certification with Rank Dependence]
\label{lemma:certification_rank_dependence}
Let $\rho, \sigma \in \mathbb{C}^{d \times d}$ be quantum states and $\eps \in (0,1)$ be a parameter. We assume that $\sigma$ is known and of rank $r$ while $\rho$ is accessible via samples. Then, using non-adaptive single-copy measurements, deciding whether
\begin{align}
        \rho=\sigma
        \quad\text{or}\quad
        F(\rho,\sigma)\leq 1-\varepsilon
\end{align}
with probability at least $0.99$, $\widetilde{O}(r^{3/2}/\varepsilon)$ samples of $\rho$ are sufficient, where the polylogarithmic factors have no dimensional dependence.

\end{lemma}

The main idea to prove the above theorem is to perform a preselection. We then measure all samples using the POVM $\{\Pi_+,\mathds{1}-\Pi_+\}$, using $O(r^{3/2}/\varepsilon)$ samples, where $\Pi_+= \textnormal{supp}(\sigma)$. If there is any hit for $\mathds{1}-\Pi_+$, we reject immediately, as it implies that $\textnormal{supp}(\rho)\neq \textnormal{supp}(\sigma)$. Otherwise, we test for equivalence between $\rho$ and $\sigma$ on $\Pi_+$ using our algorithm from before. This will cost $\widetilde{O}(r^{3/2}/\varepsilon)$ by \Cref{lemma:fid_id_testing_general}. 
We argue that this is sufficient to ensure that the states are equal.

\begin{algorithm}[H]
\LinesNumbered
\DontPrintSemicolon
\setcounter{AlgoLine}{0}
\caption{Certification of $\rho$ with rank dependence}
\label{alg:identity_rank}
\KwIn{Sample access to $N=\widetilde{O}(r^{3/2}/\eps)$ samples of $\rho$, description of $\sigma$, parameter $1>\varepsilon> 0$,}
\KwOut{Accept if $\rho=\sigma$, reject if $F(\rho,\sigma)\leq 1-\varepsilon$.}

$\Pi_+\gets \textnormal{supp}(\sigma)$, $k\gets \lceil\log(10^3r/\eps)\rceil+1$

$\{\Pi_{S_0},...,\Pi_{S_k}\} \gets$ bucketing of $\Pi_+$ for $(1-\eps/10^3)\sigma+\eps\mathds{1}_r/(10^3r)$

Measure $N/2$ samples of $\rho$ using the POVM $\{\Pi_{S_0},...,\Pi_{S_{k}},\mathds{1}-\Pi_+\}$. Reject if $\mathds{1}-\Pi_+$ appears or $\exists i:\Pi_{S_i}$ appears more than $ 10\log(k)N\max\{\tr[\Pi_{S_i}\sigma],\eps/r^{3/2}\}$ times.

$\forall x,y:$ certify $\rho=\sigma$ on $\Pi_{S_x\cup S_y}$ in $\ell_2$ to precision $\alpha^2\eps(\bv_x^{\min}+\bv_y^{\min})$

\Return reject if any test failed. \ 
\end{algorithm}
We note that the testing in line 4 needs to be performed with amplified probability of success to $1-O(1/k^2)$ to ensure a high probability of success over all $O(k^2)$ tests. We now prove the correctness:

\begin{proof}[Proof of \Cref{lemma:certification_rank_dependence}]
First, note that if $\rho=\sigma$, then line 3 in the algorithm does not reject, with high probability (which follows from standard concentration inequalities and a union bound). In the following we assume without loss of generality that $r$ is larger than some constant $r_0$. Let $k:=\lceil\log(10^3 r/\eps)\rceil + 1$. To prove correctness, assume that we have performed the certification on the support of $\sigma$ successfully, which in particular implies,
\begin{align}
\label{eq:cert_rank_guarantee}
    \forall x,y\in [k]_0: %
    \|\Pi_{S_x\cup S_y}(\rho-\sigma)\Pi_{S_x\cup S_y}\|_2\leq  \alpha\sqrt{(\bv_x^{\max}+\bv_y^{\max})\varepsilon},
\end{align}
for $\alpha:= 1/(100c_0(k+2)^2)^2$. We now need to treat the minimal eigenvalues explicitly. Measuring $\rho$ using the POVM $\{\Pi_+,\mathds{1}-\Pi_+\}$ with $r^{3/2}/(\alpha \varepsilon)$ samples of $\rho$ has also not revealed any sample in $\mathds{1}-\Pi_+$. For a suitably large constant $c_0$, this implies that $\tr[(\mathds{1}-\Pi_+)\rho]=\tr[L\rho]\leq \alpha c_0\varepsilon/r^{3/2}$ with probability $0.999$ ~\footnote{Recall that testing whether a coin always shows heads compared to showing head in at least $\nu$ fraction of times requires $O(1/\nu)$ samples.}. We now write
\begin{align}
\sigma=\begin{bmatrix}
    \sigma_+ & 0 \\ 0 & 0
\end{bmatrix},\quad 
\rho=\begin{bmatrix}
    \rho_+ & \rho_{\sim}^* \\ \rho_{\sim} & \rho_0
\end{bmatrix},
\end{align}
where we split the $d\times d$ states in block matrices after $r$ rows and columns, respectively (i.e., $\sigma_+\in \mathbb{C}^{r \times r}$ and $\rho_0\in \mathbb{C}^{d-r \times d-r}$), such that $\text{supp}(\sigma_+)=\Pi_+$, and define the state
\begin{align}
    \tilde\sigma:=\begin{bmatrix}
    (1-\tr[\rho_0])\sigma_+ & 0 \\ 0 & \rho_0
\end{bmatrix}=(1-\tr[\rho_0])\sigma + 
\begin{bmatrix}
    0 & 0 \\ 0 & \rho_0
\end{bmatrix}
.
\end{align}
Our goal is to show that $F(\rho,\sigma)> 1-\eps$, which by assumption implies $\rho=\sigma$. For this, we use that the Bures distance satisfies the triangle inequality, such that we can bound
\begin{align}
\label{eq:cert_r_f_decomp}
    1-F(\rho,\sigma)\leq 2(1-F(\rho,\tilde\sigma))+2(1-F(\tilde\sigma,\sigma)).
\end{align}
First, note that
\begin{align}
    2(1-F(\tilde\sigma,\sigma))\leq 4\tr[(\sqrt{\tilde \sigma}-\sqrt{\sigma})^2]=4\tr\left[(\sqrt{1-\tr[\rho_0]}-1)^2\sigma\right]+4\tr[\rho_0]\leq 8\tr[\rho_0]\leq \frac{8\varepsilon}{r^{3/2}}.
\end{align}
For $r_0$ large enough, we can therefore conclude that in the farness case, $1-F(\rho,\tilde\sigma)\geq \eps/4$.
To bound $(1-F(\rho,\tilde\sigma))/2$, let $\Delta$ be arbitrarily small. For the mathematical analysis, we apply the following modified depolarization channel: the support of $\sigma$ is depolarized with parameter $\eps/10^3-r\Delta$, and the entire state is depolarized with parameter $d\Delta$, resulting in $\rho_{\Delta}$, i.e.,
\begin{align}
    \rho\mapsto \rho_{\Delta}=\left(1-\frac{\eps}{10^3}-\Delta\right)\rho+ \frac{\eps}{10^3r} 
    \begin{bmatrix}
        \mathds{1}_r & 0 \\ 0 & 0
    \end{bmatrix}
    + \frac{\Delta}{d}\mathds{1}_d
\end{align}
with minimal eigenvalues $\Omega(\varepsilon/r+\Delta/d)$ on the support of $\Pi_+$, and minimal eigenvalues $\Delta/d$ otherwise. We proceed analogous for $\sigma$, resulting in $\sigma_{\Delta}$. The bucketing of $\sigma$ on the support of $\sigma$ is the same as before, and the guarantees we tested carry over. There is then only one additional bucket with support $1-\textnormal{supp}(\sigma)$. We denote by $S$ the buckets used in testing equivalence on the support of $\Pi_+$, and we include one last bucket $L:=\mathds{1}-\Pi_+$. Note that we can assume that the buckets corresponding to $\Pi_+$ have minimal eigenvalues in $\widetilde{\Omega}(\eps/r)$. If $\eps/4\leq 1-F(\rho,\tilde \sigma)$, then, by \Cref{rem:min_ev} we can bound
\begin{align}
\label{eq:f_chi2_bound_rank}
    \frac{\eps}{16}& \leq \frac{1-F(\rho_{\Delta},\tilde \sigma_{\Delta})}{2}
    \\
    &\leq \sum_{\substack{x,y\in S\\ (x,y)\neq (L,L)}}\frac{\|\Pi_{S_x\cup S_{y}}(\tilde \sigma_{\Delta}-\rho_{\Delta})\Pi_{S_x\cup S_{y}}\|_2^2}{\bv_{x}^{\min}+\bv_{y}^{\min}}+\frac{\|\Pi_L(\tilde \sigma_{\Delta}-\rho_{\Delta})\Pi_L\|_2^2}{2\Delta/d}
    \label{eq:f_chi2_bound_rank_2}
    \\
    &\leq 2\sum_{\substack{x,y\in S\\ (x,y)\neq (L,L)}}\left(\frac{\|\Pi_{S_x\cup S_{y}}((1-\tr[\rho_0])\sigma-\rho)\Pi_{S_x\cup S_{y}}\|_2^2}{\bv_{x}^{\min}+\bv_{y}^{\min}}+\frac{\|\Pi_{S_x\cup S_{y}}\rho_0\Pi_{S_x\cup S_{y}}\|_2^2}{\bv_{x}^{\min}+\bv_{y}^{\min}}\right)
    \\
    &\leq 4\sum_{\substack{x,y\in S\\ (x,y)\neq (L,L)}}\left(\frac{\|\Pi_{S_x\cup S_{y}}(\sigma-\rho)\Pi_{S_x\cup S_{y}}\|_2^2}{\bv_{x}^{\min}+\bv_{y}^{\min}}+\frac{\tr[\rho_0]^2+\|\Pi_{S_x\cup S_{y}}\rho_0\Pi_{S_x\cup S_{y}}\|_2^2}{\bv_{x}^{\min}+\bv_{y}^{\min}}\right)
    \\
    &\leq 4\sum_{x,y\in S\setminus L}\frac{\|\Pi_{S_x\cup S_{y}}(\sigma-\rho)\Pi_{S_x\cup S_{y}}\|_2^2}{\bv_{x}^{\min}+\bv_{y}^{\min}}
    \\
    &\quad +4\sum_{\substack{x,y\in S \\ x= L \oplus y= L}}\left(\frac{\|\Pi_{S_x\cup S_{y}}(\sigma-\rho)\Pi_{S_x\cup S_{y}}\|_2^2}{\bv_{x}^{\min}+\bv_{y}^{\min}}\right)+\max_{(x,y)\neq (L,L)}\frac{8(k+1)^2\tr[\rho_0]^2}{\bv_{x}^{\min}+\bv_{y}^{\min}},
\end{align}
where $\oplus$ denotes the XOR condition. Note that we can discard the term corresponding to $(x,y)=(L,L)$ since $\Pi_L(\tilde\sigma_{\Delta}-\rho_{\Delta})\Pi_L=0$ by construction. This is because by construction of $\tilde\sigma$, it coincides with $\rho$ on $\text{supp}(\Pi_L)$. If all tests indicate closeness, we can bound the first term directly using the guarantees we get from \eqref{eq:cert_rank_guarantee},
\begin{align}
    4\sum_{x,y\in S\setminus L}\frac{\|\Pi_{S_x\cup S_{y}}(\sigma-\rho)\Pi_{S_x\cup S_{y}}\|_2^2}{\bv_{x}^{\min}+\bv_{y}^{\min}}\leq 4(k+1)^2\alpha^2\eps\leq \frac{\eps}{2^7}.
\end{align}
For the remaining terms, we note that
\begin{align}
    \|\Pi_{S_x\cup L}(\rho-\sigma)\Pi_{S_x\cup L}\|_2&\leq \|\Pi_{S_x}(\rho-\sigma)\Pi_{S_x}\|_2+2\|\Pi_{L}(\rho-\sigma)\Pi_{S_x}\|_2+\|\Pi_{L}(\rho-\sigma)\Pi_{L}\|_2
        \\
        &\leq \alpha\sqrt{\bv_x^{\max}\varepsilon} + 2\|\Pi_{L}\rho\Pi_{S_x}\|_2+ 2\|\Pi_{L}\sigma\Pi_{S_x}\|_2+\|\Pi_{L}\rho\Pi_{L}\|_2
        \\
        &\leq \alpha\sqrt{\bv_x^{\max}\varepsilon} + 2\|\Pi_{L}\sqrt{\rho}\|_{\opn}\|\sqrt{\rho}\Pi_{S_x}\|_{2}+\|\Pi_{L}\rho\Pi_{L}\|_2
        \\
        &\leq \alpha\sqrt{\bv_x^{\max}\varepsilon} + 2\sqrt{\|\Pi_{L}\rho\Pi_{L}\|_{\opn}\tr[\Pi_{S_x}\rho\Pi_{S_x}]}+\tr[\Pi_{L}\rho\Pi_{L}].
\end{align}
Further,
\begin{align}
    \tr[\Pi_{S_x}\rho\Pi_{S_x}]&\leq 30\log(k)(\tr[\Pi_{S_x}\sigma\Pi_{S_x}]+2\eps/r^{3/2})\leq 100\log(k)r\bv_x^{\max}
\end{align}
and 
\begin{align}
    \|\Pi_{L}\rho\Pi_{L}\|_{\opn}\leq \tr[\Pi_{L}\rho\Pi_{L}]\leq \frac{\alpha c_0\varepsilon}{r^{3/2}}.
\end{align}
Finally,
\begin{align}
    \|\Pi_{S_x\cup S_{y}}\rho_0\Pi_{S_x\cup S_{y}}\|_2\leq \tr[\Pi_{S_x\cup S_{y}}\rho_0\Pi_{S_x\cup S_{y}}]\leq \tr[\rho_0]
\end{align}
Then, since $\bv_x^{\max}\geq \varepsilon/(10^3r)$, 
\begin{align}
    \|\Pi_{S_x\cup L}(\rho-\sigma)\Pi_{S_x\cup L}\|_2&\leq 2\alpha\sqrt{\bv_x^{\max}\eps} + 20\sqrt{\frac{\alpha c_0\eps}{r^{3/2}}\log(k)r\bv_x^{\max}}+\frac{\alpha c_0\eps}{r^{3/2}}
    \leq 30c_0\sqrt{\alpha\log(k) \bv_x^{\max}\eps},
\end{align}
such that, using $\alpha= 1/(100c_0(k+2)^2)^2$ and $\bv_x^{\max}/\bv_x^{\min}\leq 2$,
\begin{align}
    4\sum_{\substack{x\in S\\ x\neq L}}\frac{\|\Pi_{S_x\cup L}(\rho-\sigma)\Pi_{S_x\cup L}\|_2^2}{\bv_x^{\min}+\Delta/d
    }
    \leq 
    8\sum_{\substack{x\in S\\ x\neq L}}30^2\log(k)\alpha c_0^2\varepsilon
    \leq 
    \sum_{\substack{x\in S\\ x\neq L }}\frac{\varepsilon}{8(k+2)^2}\leq \frac{\eps}{2^7},
\end{align}
as required, and analogously for the case with switched roles, where $y\notin L$ and $x\in L$. Finally, we use again that $\forall x\neq L: s_x^{\min}\geq \eps/(2\cdot 10^3r)$ to bound
\begin{align}
    \max_{(x,y)\neq (L,L)}\frac{8(k+1)^2\tr[\rho_0]^2}{\bv_{x}^{\min}+\bv_{y}^{\min}}\leq 
    \max_{(x,y)\neq (L,L)}\frac{8(k+1)^2\alpha^2c_0^2\eps^2}{r^3(\bv_{x}^{\min}+\bv_{y}^{\min})}
    \leq \frac{\eps}{2^7}.
\end{align}
Then the right-hand side of \eqref{eq:f_chi2_bound_rank_2} is bounded by $\varepsilon/32$ in case all tests accept, which guarantees $F(\rho,\sigma)=1$ by assumption.
\end{proof}

\subsection{Optimality of our Certification Bound}
We conclude our discussion on state certification with guarantees in fidelity by arguing that our algorithms are optimal up to polylogarithmic factors. This follows from the lower bounds in \cite{chen_toward_2021}, as any improvement would contradict their lower bounds on state certification with respect to trace distance.

\begin{lemma}
\label{lemma:certification_lower_bounds}
Let $\rho, \sigma \in \mathbb{C}^{d \times d}$ be unknown and known quantum states, respectively, and $0<\eps\leq \eps_0$ be a parameter, for $\eps_0$ a sufficiently small constant. Moreover, $\sigma$ is of rank $r$. Then there exists $\sigma$ such that, using single-copy adaptive measurements, deciding whether
\begin{align}
\rho=\sigma
\quad\text{or}\quad
F(\rho,\sigma)\leq 1-\varepsilon
\end{align}
with probability at least $0.99$, takes $\widetilde{\Omega}(r^{3/2}/\varepsilon)$ samples of $\rho$. Moreover, the lower bound without an assumption on the rank is $\Omega(d^{3/2}/\varepsilon)$.

\end{lemma}

\begin{proof}
Following the Fuchs-van de Graaf inequalities (\Cref{lem:fucsvandegraff}), it follows that our algorithm for certification in fidelity to precision $\eps/2$ also certifies with respect to the trace distance to precision $\sqrt{\eps}$. By \cite[Thm.\ 2.1]{chen2022tightboundsquantumstate} (using notation as in \cite[Thm.\ 5.1]{chen_toward_2021}), testing to precision $\sqrt{\varepsilon}$ will cost
 \begin{align}            
    \widetilde{\Omega}\left(\frac{d\sqrt{\textnormal{rank}(\hat\sigma^{**})}F(\hat\sigma^{**},\mathds{1}_d/d)^2}{\varepsilon}\right),
\end{align}
    where $\hat\sigma^{**}$ is a processed version of $\sigma$ in which small eigenvalues of at most $7\varepsilon$ have been zeroed out (and then renormalized), see \cite[Def.\ 5.2]{chen_toward_2021}. 
    Thus our claim follows if we choose $\sigma$ to be maximally mixed on either the full dimension $d$ (thereby obtaining the lower bound of $\Omega(d^{3/2}/\eps)$), or maximally mixed over a subspace of dimension $r$ (thereby obtaining a lower bound of $\Omega(r^{3/2}/\eps)$). It is easy to see that for our choice of $\sigma$, the processing from $\sigma$ to $\hat\sigma^{**}$ leaves $F(\hat\sigma^{**},\mathds{1}_d/d)$ at the order of $F(\sigma,\mathds{1}_d/d)$.
\end{proof}

\section{Equivalence Testing in Fidelity}
\label{sec:equivalence}

In this section, we prove the result about equivalence testing with respect to fidelity:

\fideqtesting*

As described in \Cref{sec:overview_equiv_fidelity}, the approach is to first approximately learn buckets, and then perform a piecewise comparison, similar to certification. For the approximate learning of one of the states $\rho$, we follow the approach of \cite{chen_when_2023} to show how we can perform a bucketing down to a certain threshold $\eta$, see \Cref{sec:equiv:learning}. For the testing, discussed in \Cref{sec:equiv:testing}, we will use \Cref{idenl2:lemma:equiv_proj}. Note that \Cref{idenl2:lemma:equiv_proj} does not need to know either $\sigma$ or $\rho$, but only an upper bound on the maximal eigenvalue in the respective bucket suffices. Finally, in \Cref{sec:equiv:balancing}, we balance the tasks of learning and testing, which are expressed as functions of $\eta$, to determine the optimal threshold with respect to the sample complexity. We will throughout assume that we have minimal eigenvalues $\widetilde{\Omega}(\eps/d)$, as outlined in \Cref{rem:min_ev}.

\begin{lemma}[{\cite[Lemma 5.7]{chen_when_2023}}]
\label{lemma:learn_buckets_chl}
Let $\rho \in \mathbb{C}^{d \times d}$ be an unknown quantum state. Given sample access to $\rho$ and a projection $\Pi$ to a subspace. Let $\mathcal{M}$ be the uniform POVM over that subspace. Then, we measure $\rho$ in the POVM $\mathcal{M}\cup \{\mathds{1}-\Pi\}$, and record all outcomes in $\mathcal{M}$ in a set $M$. Construct $G(\rho,\Pi,N)=\frac{1}{N}\sum_{v_i\in M}((d_{\Pi}+1)\dyad{v_i}- \mathbb{I}_{d_{\Pi}})$. Then
    \begin{equation}
        \|G(\rho,\Pi,N)-\Pi\rho\Pi\|_{\opn} \leq \nu= \cOpLearn\max\left\{\frac{d+\log(1/\delta)}{N},\sqrt{\tr[\Pi\rho\Pi]\frac{d+\log(1/\delta)}{N}}\right\},%
    \end{equation}
    with probability at least $1-\delta$.
\end{lemma}

We note that $G(\rho,\Pi,N)$ is not a state - since we will only use it to determine the bucketing, this technicality does \emph{not} affect us. %

Let us introduce some notation. Consider a set of pairwise orthogonal projections $\{Q_i\}$, $i\in [k]_0$, which sum to $\mathds{1}$ and let
\begin{align}
    Q_i^+:=\sum_{j=i}^kQ_j,\qquad Q_i^-:=\mathds{1}-Q_i^+=\sum_{j=0}^{i-1}Q_j.
\end{align}
In particular, note that for $j\geq i$, $Q_jQ_i^+=Q_i^+Q_j=Q_j$. We can define $\hat{\Pi}_j^+$ by only knowing $\{\Pi_l\}_{l=0}^{j-1}$, and not $\{\Pi_l\}_{l=j}^{k}$. 

Let $k:=\lceil\log(1/\eta)\rceil$. For readability, we will use (see \eqref{eq:id_buckets})
\begin{align}
    0\leq i<k:\bv_i^{\min}:=2^{-i-1}, \bv_i^{\max}:=2^{-i},\quad\text{and} \quad \bv_k^{\max}=2^{-k}=\Theta(\eta), \bv_k^{\min}=\Omega(\eps/d),
\end{align}
where $\bv_k^{\min}$ is given by the minimal eigenvalues due to depolarization. These now correspond to the \emph{target values} we aim to have in the normal buckets (since the buckets are only approximate, there will be deviations). Crucially, unlike for certification, for the last bucket, we will now \emph{not} have that $\bv_i^{\max}/\bv_i^{\min}\leq O(1)$.

\subsection{Approximate Learning of a State via Bucketing}
\label{sec:equiv:learning}

The algorithm to determine the buckets is analogous to the subroutine in \cite{chen_when_2023}. Specifically, the idea is to learn buckets iteratively. In every iteration, we project away the earlier buckets we already defined. Concretely, this means that at iteration $i$, we will use \Cref{lemma:learn_buckets_chl} with $\Pi=\mathds{1}_d-\sum_{j=0}^{i-1}\hat\Pi_j=\hat\Pi_i^+$. We present the algorithm in \Cref{alg:equiv_bucketing}. The overall number of samples required, $N_{\text{learn}}$, will be determined by the analysis in \Cref{lemma:bucketing_guarantee}.

\begin{algorithm}[H]
\LinesNumbered
\DontPrintSemicolon
\setcounter{AlgoLine}{0}
\caption{Determine bucketing}
\label{alg:equiv_bucketing}
\KwIn{$N_{\text{learn}}$ samples of the unknown state $\sigma \in \mathbb{C}^{d \times d}$, parameter $\eta \in (0, 1)$.}
\KwOut{A bucketing of $\sigma$ down to $\eta$: $\{\hat\Pi_0, \hat\Pi_1, \ldots, \hat\Pi_{k} \}$.}

$\hat\Pi_0^+ \gets \mathds{1}_d$, %

$k \gets \lceil\log(1/\eta)\rceil$ for a suitable constant $c$.

\For{$i\in [k-1]_0$}{
    $H(\sigma,\hat\Pi_i^+)\gets G(\sigma,\hat\Pi_i^+, N_{\text{learn}}/k)$

    $\{\ket{h_j}\}$, $\{h_j\}$ $\gets$ eigenvectors and eigenvalues ($> 0$) of $H(\sigma,\hat\Pi_i^+)$, respectively
    
    $\hat\Pi_{i}\gets$ $\sum_{j:h_j\geq \bv_i^{\min}}\dyad{h_j}$%
}

$\hat\Pi_{k} \gets$ $\mathds{1}_d-\sum_{j=0}^{k-1}\hat\Pi_{j}$.

\Return $\{\hat\Pi_i\}_{i=0}^k$. \

\end{algorithm}
\Cref{lemma:learn_buckets_chl} gives us the following necessary guarantee on the precision of the buckets.

\begin{lemma}
\label{lemma:bucketing_guarantee}
    For $\delta\in (0,1)$, \Cref{alg:equiv_bucketing} can learn an empirical bucketing $\{\hat\Pi_i\}_{i=0}^k$ with state approximations $\{H(\sigma,\hat\Pi_i)\}$, such that, with probability at least $1-\delta$ overall, with $\bv_{j}^{\max}/\bv_{j+1}^{\max}\leq 2$,
    \begin{equation}
        \forall j\in[k-1]_0:~ %
        \left\|H(\sigma,\hat\Pi_j^+)-\hat{\Pi}_j^+\sigma\hat{\Pi}_j^+\right\|_{\opn}\leq \alpha\sqrt{\eta \bv_j^{\max}},
    \end{equation}
    where $\alpha<1/2$ is a free parameter, as well as 
        \begin{align}
    \label{eq:op_bound}
        \forall j\in [k]_0:~\left\|\hat{\Pi}_j^+\sigma\hat{\Pi}_j^+\right\|_{\opn}\leq 2\bv_j^{\max}, 
    \end{align}
    using 
    \begin{align}
        N_{\textnormal{learn}}=\frac{\cOpG}{\alpha^2}k \frac{\min\{d\eta ,1\}}{\eta^2}(d+\log(1/\delta')).  
    \end{align}
samples of $\sigma$, where $\cOpG$ is a suitably chosen constant and $\delta'=\Theta(\delta/k)$.
\end{lemma}

\begin{proof}
    We proceed by induction. For the induction basis, first note that $\|\hat{\Pi}_0^+\sigma\hat{\Pi}_0^+\|_{\opn}\leq 2$ is trivial, and 
    \begin{align}
        \left\|H(\sigma,\hat\Pi_0^+)-\hat{\Pi}_0^+\sigma\hat{\Pi}_0^+\right\|_{\opn}
        &\leq \cOpLearn\max\left\{\frac{d+\log(1/\delta')}{N_{\textnormal{learn}}/k}, \sqrt{\frac{\tr[\sigma](d+\log(1/\delta'))}{N_{\textnormal{learn}}/k}}\right\}
        \\
        &\leq \alpha\cOpLearn\sqrt{\frac{\eta^2}{\min\{d\eta,1\}}}= \alpha\cOpLearn\max\{\eta,\sqrt{\eta/d}\}%
    \end{align}
    follows directly from \Cref{lemma:learn_buckets_chl}. Note that here and in the following, a bucket might also be empty. We will further drop the constant $\cOpLearn$, which we might absorb by redefining $\eta$ suitably. We now assume that for any $j$,
    \begin{align}
         \left\|H(\sigma,\hat\Pi_j^+)-\hat{\Pi}_j^+\sigma\hat{\Pi}_j^+\right\|_{\opn}\leq \alpha\sqrt{\eta \bv_j^{\max}}\quad\text{and}\quad \left\|\hat{\Pi}_j^+\sigma\hat{\Pi}_j^+\right\|_{\opn}\leq 2\bv_j^{\max}.
    \end{align}
    and step to $j+1$. Note that we can construct $\hat{\Pi}_{j+1}^+$ from $\{\hat{\Pi}_{0},...,\hat{\Pi}_{j}\}$ alone.
    \begin{enumerate}
        \item Using the Hölder inequality and $\hat{\Pi}_j^+\hat{\Pi}_{j+1}^+=\hat{\Pi}_{j+1}^+$,
        \begin{align}
            \alpha\sqrt{\eta \bv_j^{\max}}
            \geq 
            \left\|H(\sigma,\hat\Pi_j^+)-\hat{\Pi}_j^+\sigma\hat{\Pi}_j^+\right\|_{\opn}
            \geq
            \left\|\hat{\Pi}_{j+1}^+(H(\sigma,\hat\Pi_j^+)-\sigma)\hat{\Pi}_{j+1}^+\right\|_{\opn},
        \end{align}
        such that 
        \begin{align}
            \left\|\hat{\Pi}_{j+1}^+\sigma\hat{\Pi}_{j+1}^+\right\|_{\opn}
            \leq 
            \alpha\sqrt{\eta \bv_j^{\max}}+\|\hat{\Pi}_{j+1}^+H(\sigma,\hat\Pi_j^+)\hat{\Pi}_{j+1}^+\|_{\opn}
            \leq 
            \alpha\sqrt{\eta \bv_j^{\max}}+\bv_{j+1}^{\max}\leq 2\bv_{j+1}^{\max}.
        \end{align}
        The last inequality holds since $\hat\Pi_j$ is constructed to contain all eigenvalues of $H(\sigma,\hat\Pi_j^+)$ which are at least $\bv_j^{\min}$, leaving only $\bv_{j+1}^{\max}\leq\bv_j^{\min}$ in the support of $\hat\Pi_{j+1}^+$. Further, note that the eigenvalues of $H(\sigma,\hat \Pi_j^+)$ are lower bounded by $-\alpha\sqrt{\eta\bv_{j}^{\max}}$ since $\sigma\geq 0$.
        
        \item To prove the bound on the distance between $H(\sigma,\hat\Pi_{j+1}^+)$ and $\hat{\Pi}_{j+1}^+\sigma\hat{\Pi}_{j+1}^+$, we simply note that the precision we get from \Cref{lemma:learn_buckets_chl} can be bounded by using our bound on $\|\hat{\Pi}_{j+1}^+\sigma\hat{\Pi}_{j+1}^+\|_{\opn}$, such that
        \begin{align}
            \tr \big[\hat{\Pi}_{j+1}^+\sigma\big]\leq \min\{d\|\hat{\Pi}_{j+1}^+\sigma\hat{\Pi}_{j+1}^+\|_{\opn},1\}\leq \min\{2d\bv_{j+1}^{\max},1\},
    \end{align}
    such that, using \Cref{lemma:learn_buckets_chl}, the precision is bounded by 
    \begin{align}
        \left\|H(\sigma,\hat\Pi_{j+1}^+)-\hat{\Pi}_{j+1}^+\sigma\hat{\Pi}_{j+1}^+\right\|_{\opn}&\leq \cOpLearn\max\left\{\frac{d}{N_{\text{learn}}^{(j+1)}},\sqrt{\frac{d\tr \big[\hat{\Pi}_{j+1}^+\sigma\big]}{N_{\text{learn}}^{(j+1)}}}\right\}
        \\
        &\leq
        \cOpLearn\max\left\{\frac{d}{N_{\text{learn}}^{(j+1)}},\sqrt{\frac{d\min\{2d\bv_{j+1}^{\max},1\}}{N_{\text{learn}}^{(j+1)}}}\right\}
        \\
        &=
        \cOpLearn\sqrt{\frac{d}{N_{\text{learn}}^{(j+1)}}}\max\left\{\sqrt{\frac{d}{N_{\text{learn}}^{(j+1)}}},\sqrt{\min\{2d\bv_{j+1}^{\max},1\}}\right\}
        \\
        &=\cOpLearn\sqrt{\frac{d}{N_{\text{learn}}^{(j+1)}}\min\{2d\bv_{j+1}^{\max},1\}}
        \\
        &=\cOpLearn\sqrt{\frac{\alpha^2}{\cOpG}\frac{\eta^2\min\{2d\bv_{j+1}^{\max},1\}}{\min\{d\eta,1\}}}
        \\
        &\leq \cOpLearn\sqrt{\frac{\alpha^2}{\cOpG}\frac{\eta^22d\bv_{j+1}^{\max}}{d\eta}}\leq \alpha\sqrt{\eta \bv_{j+1}^{\max}},
    \end{align}
    as long as $\cOpG\geq 4(\cOpLearn)^2$. In the last inequality, we used that for $0<b\leq a$, $\frac{\min\{a,1\}}{\min\{b,1\}}\leq \frac{a}{b}$. %
    \end{enumerate}
    \end{proof}

Next, we will show a few helper lemmas which can be derived from the primary guarantee in \Cref{lemma:bucketing_guarantee}. They will bound the overlap between the empirical buckets $\hat\Pi_k$ and the correct buckets $\Pi_i$ in the case where $i$ and $k$ are gapped away from each other. We will use the following folklore facts:

\begin{fact}\label{lemma:op_prop_1}
 Let $\{R_i\}$ be a set of pairwise orthogonal projections, $Q$ some arbitrary projection, and $\{x_i\}$ and $\{y_i\}$ numbers such that $\forall i:0\leq x_i\leq y_i$. Then
    \begin{equation}
        \left\|Q\sum_ix_iR_iQ\right\|_{\opn}\leq \left\|Q\sum_iy_iR_iQ\right\|_{\opn}.
    \end{equation}
\end{fact}

\begin{fact}
\label{fact:bound_overlap}
    Let $A$ be a $d$-dimensional Hermitian matrix with eigenvalues $\lambda_i$ and eigenvectors $\{\ket{\lambda_i}\}$, and $S\subseteq[d]$ such that $\forall i\in S:\lambda_i\in [x,y]$, for some fixed $x,y\in\mathbb{R}^+$. Then, for $\Pi:=\sum_{i\in S}\dyad{\lambda_i}$ and any arbitrary projection $Q$, 
    \begin{align}
        x\|\Pi Q\|_{\opn}\leq \|\Pi A Q\|_{\opn}\leq y\|\Pi Q\|_{\opn}.
    \end{align}
\end{fact}

We first prove the overlap between the first bucket and any bucket with $i\geq 2$.

\begin{lemma}
\label{lemma:first_bucket}
    For the first bucket, if $\|H(\rho,\hat\Pi_0^+)-\rho\|_{\opn}\leq c\beta$ holds, and $2\bv_i^{\max}\leq \bv_0^{\min}$, then
    \begin{align}
    \label{eq:first_bucket_guarantee}
        \forall i\geq 2: \|\Pi_i\hat{\Pi}_0\|_{\opn}\leq 2c\frac{\beta}{\bv_0^{\min}},\quad \|\Pi_i^+\hat{\Pi}_0\|_{\opn}\leq 2c\frac{\beta}{\bv_0^{\min}}.%
    \end{align}
\end{lemma}
\begin{proof}
The proofs for both inequalities proceed in the exact same way, in every step, $\Pi_i$ may be replaced with $\Pi_i^+$. By definition of $\hat{\Pi}_0$, it is made from eigenvectors of $H(\rho,\hat\Pi_0^+)=H(\rho,1_d)$, such that we can apply \Cref{fact:bound_overlap} and find that $\bv_0^{\min}\| \Pi_i\hat{\Pi}_0\|_{\opn}\leq \|\Pi_iH(\rho,\hat\Pi_0^+)\hat{\Pi}_0\|_{\opn}$. Similarly, $\Pi_i$ maps to the eigenspace of $\rho$, such that $\|\Pi_i\rho \hat{\Pi}_0\|_{\opn}\leq \bv_i^{\max}\|\Pi_i\hat{\Pi}_0\|_{\opn}$.
    We then have 
    \begin{align}
        \bv_0^{\min}\| \Pi_i\hat{\Pi}_0\|_{\opn}&\leq \|\Pi_iH(\rho,\hat\Pi_0^+)\hat{\Pi}_0\|_{\opn}
        \\
        &\leq \|\Pi_i(H(\rho,\hat\Pi_0^+)-\rho)\hat{\Pi}_0\|_{\opn}+\|\Pi_i\rho \hat{\Pi}_0\|_{\opn}
        \\
        &\leq \|H(\rho,\hat\Pi_0^+)-\rho\|_{\opn}+\bv_i^{\max}\|\Pi_i\hat{\Pi}_0\|_{\opn}
        \\
        \implies \| \Pi_i\hat{\Pi}_0\|_{\opn}&\leq \frac{\|H(\rho,\hat\Pi_0^+)-\rho\|_{\opn}}{\bv_0^{\min}-\bv_i^{\max}}.
    \end{align}
    Lastly, we use that $\bv_i^{\max}\leq \bv_0^{\min}/2$  and $\|H(\rho,\hat\Pi_0^+)-\rho\|_{\opn}\leq c\beta$ by assumption.
\end{proof}

For technical reasons, we will need the guarantee from \Cref{lemma:first_bucket} also for the first bucket: for this we can technically learn the first bucket to the precision of $\hat\Pi_1$, and then 'manually' split it into $\hat\Pi_0$ and $\hat\Pi_1$, such that \eqref{eq:first_bucket_guarantee} also holds for $\hat\Pi_1$.

\begin{lemma}
\label{lemma:equiv:overlap}
    Assume we have a set $\{\hat{\Pi}_j\}$ and state approximations $\{H(\rho,\hat\Pi_j^+)\}$, such that, 
    \begin{equation}
        \forall j\geq 0: \|H(\rho,\hat\Pi_j^+)-\hat{\Pi}_j^+\rho\hat{\Pi}_j^+\|_{\opn}\leq c\sqrt{\eta\bv_j^{\max}},
        \quad \text{and} \quad
        \forall j\geq 2: \|\Pi_j^+\hat{\Pi}_0\|_{\opn}\leq c\sqrt{\eta}.    
    \end{equation}
    Let $k>m>0$ be arbitrary (i.e., $\bv_m^{\min}\geq \eta$). Suppose there exists a constant $\kappa$ such that $\forall j \leq m: \kappa\geq\bv_j^{\max}/\bv_j^{\min}$. Then, for any $i>m+1$,  the following holds:
    \begin{align}
    \label{eq:equiv_io_overlap}
        \|\Pi_i^+\hat{\Pi}_m\|_{\opn}\leq 4(1+2c\sqrt{\kappa})^mc\sqrt{\frac{\kappa\eta}{\bv_m^{\min}}}.
    \end{align}
\end{lemma}
\begin{proof}
    Recall that $\hat{\Pi}_m$ is constructed from eigenvectors of $H(\rho,\hat\Pi_m^+)$ corresponding to eigenvalues at least $\bv_m^{\min}$. Using \Cref{fact:bound_overlap} twice, we find that
    \begin{align}
        \bv_m^{\min}\|\Pi_i\hat{\Pi}_m\|_{\opn}
        \leq \|\Pi_iH(\rho,\hat\Pi_m^+)\hat{\Pi}_m\|_{\opn} 
        &\leq \|\Pi_i(H(\rho,\hat\Pi_m^+)-\rho)\hat{\Pi}_m\|_{\opn}
        +\|\Pi_i\rho\hat{\Pi}_m\|_{\opn}
        \\
        &\leq \|\Pi_i(H(\rho,\hat\Pi_m^+)-\rho)\hat{\Pi}_m\|_{\opn}
        +\bv_i^{\max}\|\Pi_i\hat{\Pi}_m\|_{\opn}.
    \end{align}
    Then
    \begin{align}
        (\bv_m^{\min}-\bv_i^{\max})\|\Pi_i\hat{\Pi}_m\|_{\opn}&\leq \|\Pi_i(H(\rho,\hat\Pi_m^+)-\rho)\hat{\Pi}_m\|_{\opn}
        \\
        &=\|\Pi_i(\hat{\Pi}_m^+ + (1-\hat{\Pi}_m^+))(H(\rho,\hat\Pi_m^+)-\rho)\hat{\Pi}_m\|_{\opn}
        \\
        &\leq \|\Pi_i\hat{\Pi}_m^+(H(\rho,\hat\Pi_m^+)-\rho)\hat{\Pi}_m\|_{\opn}+\|\Pi_i(\mathds{1}-\hat{\Pi}_m^+)(H(\rho,\hat\Pi_m^+)-\rho)\hat{\Pi}_m\|_{\opn}
        \\
        &= \|\Pi_i\hat{\Pi}_m^+(H(\rho,\hat\Pi_m^+)-\rho)\hat{\Pi}_m^+\hat{\Pi}_m\|_{\opn} +\|\Pi_i(\mathds{1}-\hat{\Pi}_m^+)\rho\hat{\Pi}_m\|_{\opn} %
        \\
        &\leq \|\hat{\Pi}_m^+(H(\rho,\hat\Pi_m^+)-\rho)\hat{\Pi}_m^+\|_{\opn} +\|\Pi_i(\mathds{1}-\hat{\Pi}_m^+)\rho\hat{\Pi}_m\|_{\opn}%
        \\
        &= \|H(\rho,\hat\Pi_m^+)-\hat{\Pi}_m^+\rho\hat{\Pi}_m^+\|_{\opn}+\|\Pi_i(\mathds{1}-\hat{\Pi}_m^+)\rho\hat{\Pi}_m\|_{\opn}.\label{eq:op_bound_intermediate} %
    \end{align}
    For $j<m$, $\hat\Pi_m=\hat\Pi_j^+\hat\Pi_m$. Further, by construction of $\hat\Pi_j$, we have that $H(\rho,\hat\Pi_j^+)\hat\Pi_j=\hat\Pi_jH(\rho,\hat\Pi_j^+)$. Thus %
    \begin{align}
        \|\Pi_i(1-\hat{\Pi}_m^+)\rho \hat{\Pi}_m\|_{\opn}&\leq \sum_{j=0}^{m-1}\|\Pi_i\hat{\Pi}_j\rho\hat{\Pi}_m\|_{\opn}
        \\
        &= \sum_{j=0}^{m-1}\|\Pi_i\hat{\Pi}_j(H(\rho,\hat\Pi_j^+)\hat{\Pi}_j-\rho)\hat{\Pi}_m\|_{\opn}
        \\
        &= \sum_{j=0}^{m-1}\|\Pi_i\hat{\Pi}_j(H(\rho,\hat\Pi_j^+)-\rho)\hat{\Pi}_m\|_{\opn}
        \\
        &=\sum_{j=0}^{m-1}\|\Pi_i\hat{\Pi}_j\hat{\Pi}_j^+(H(\rho,\hat\Pi_j^+)-\rho)\hat{\Pi}_j^+\hat{\Pi}_m\|_{\opn}
        \\
        &\leq \sum_{j=0}^{m-1}\|\Pi_i\hat{\Pi}_j\|_{\opn}\|\hat{\Pi}_j^+(H(\rho,\hat\Pi_j^+)-\rho)\hat{\Pi}_j^+\|_{\opn}.
    \end{align}
    Using our guarantee, and using that $\bv_m^{\min}-\bv_i^{\max}\geq \bv_m^{\min}/2$ we find that
    \begin{align}
        \|\Pi_i\hat{\Pi}_m\|_{\opn}&\leq 2\frac{\|\hat{\Pi}_m^+(H(\rho,\hat\Pi_m^+)-\rho)\hat{\Pi}_m^+\|_{\opn}}{\bv_m^{\min}}+2\sum_{j=0}^{m-1}\frac{\|\hat{\Pi}_j^+(H(\rho,\hat\Pi_j^+)-\rho)\hat{\Pi}_j^+\|_{\opn}}{\bv_m^{\min}}\|\Pi_i\hat{\Pi}_j\|_{\opn}
        \\
        &\leq 2c\sqrt{\frac{\kappa\eta}{\bv_m^{\min}}}+\sum_{j=0}^{m-1}\frac{2c\sqrt{\eta\bv_j^{\max}}}{\bv_m^{\min}}\|\Pi_i\hat{\Pi}_j\|_{\opn}.
    \end{align}
    To resolve this recursion, define for all $r$, $z_{r,i}:=\sqrt{\bv_r^{\min}}\|\Pi_i\hat{\Pi}_r\|_{\opn}$, then this is equal to
    \begin{align}
        z_{m,i}\leq 2c\sqrt{\kappa\eta}+2c\sum_{j=0}^{m-1}\sqrt{\frac{\eta}{\bv_m^{\min}}\frac{\bv_j^{\max}}{\bv_j^{\min}}}z_{j,i}\leq 2c\sqrt{\kappa\eta}+2c\sqrt{\kappa}\sum_{j=0}^{m-1}z_{j,i}.%
    \end{align}
    Here we used that, since $i>m+1$, $m$ does not denote the last bucket, and therefore we can bound $\eta/\bv_m^{\min}\leq 1$. Solving an upper bound, $z_{m,i}\leq \tilde{z}_{m,i}=2c\sqrt{\kappa\eta}+2c\sqrt{\kappa}\sum_{j=0}^{m-1}\tilde{z}_{j,i}$, we find 
    \begin{align}
        \tilde{z}_{m,i}-\tilde{z}_{m-1,i}= 2c\sqrt{\kappa}\tilde{z}_{m-1,i}\quad \implies\quad  \tilde{z}_{m,i}=(1+2c\sqrt{\kappa})\tilde{z}_{m-1,i},
    \end{align}
    such that $z_{m,i}\leq (1+2c\sqrt{\kappa})^{m-1}(2c\sqrt{\kappa\eta}+z_{0,i})$, i.e.,
    \begin{align}
        \|\Pi_i\hat{\Pi}_m\|_{\opn}\leq 4(1+2c\sqrt{\kappa})^{m-1}c\sqrt{\frac{\kappa\eta}{\bv_m^{\min}}},
    \end{align}
    Note that $\Pi_i$ can be replaced by $\Pi_i^+$ in all arguments.
\end{proof}

We now have an approximate bucketing, as well as some guarantees on the overlap between the true buckets and our estimate. We will now use these to show how we can perform the equivalence testing.

\subsection{Testing}
\label{sec:equiv:testing}
Our approach for certification bounds the fidelity as a sum of weighted `pieces' of the state in the Hilbert Schmidt distance, using the ideal bucketing $\{\Pi_i\}$ of the known state $\sigma$. However, since both states are unknown here, we will instead define approximate buckets $\{\hat\Pi_i\}$. Our goal would be to connect the expression using ideal buckets (to which we have no access) to an expression involving our empirical buckets, which we can actually use for equivalence testing. 
As before, $\Pi_{S_j\cup S_l}$ denotes the projector projecting to the subspaces corresponding to the buckets $S_j$ and $S_l$. Note that %
$\bv_k^{\min}=\Theta(\varepsilon/d)$. Following \eqref{eq:bound_f_by_l2}, we have the following:
\begin{align}
    \label{eqn:testfidelityequivalence}
    1-F(\rho,\sigma) \leq \frac{1}{2}D_{\chi^2}(\rho\|\sigma) \leq \sum_{j,l=0}^k\frac{\|\Pi_{S_{j}\cup S_{l}}(\rho-\sigma)\Pi_{S_{j}\cup S_{l}}\|_2^2}{\bv^{\min}_{j}+\bv^{\min}_{l}}\leq 2(k+1)\sum_{j=0}^k\frac{\|\Pi_{j}^+(\rho-\sigma)\Pi_{j}^+\|_2^2}{\bv^{\min}_{j}}.
\end{align}
Note that testing $\{\Pi_j^+\}$ instead of the $\{\Pi_{S_{j}\cup S_{l}}\}$ would not increase the sample complexity we derived for identity testing, since we bounded $d_{\Pi}\leq d$ anyways, and by stepping from $\Pi_{S_{j}\cup S_{l}}$ to $\Pi_{\min\{j,l\}}^+$, we only include smaller eigenvalues.

{%
The problem with \eqref{eqn:testfidelityequivalence} is that we do not know $\{\Pi_j^+\}$. The idea, instead, is to use the following projections as a proxy. Note in particular that the sum always starts two indices earlier.
\begin{equation}
    \label{eq:equiv:bucketing}
    \eqM_{i}:=\sum_{j=\max\{i-2,0\}}^{k}\hat\Pi_j.
\end{equation}
We will now argue how this works formally. If we could ensure that 
\begin{align}
\label{eq:equiv_coseness_ideal_pi}
    \forall j:\|\Pi_{j}^+(\rho-\sigma)\Pi_{j}^+\|_2^2< \frac{\eps s_j^{\min}}{2(k+1)^2},
\end{align}
then this would imply that $F(\rho,\sigma)>1-\eps$ by \eqref{eqn:testfidelityequivalence}. By assumption on the decision gap being at least $\eps$ means that $\rho=\sigma$. On the other hand, if for \emph{any} projection $\Pi$ (not necessarily in $\{\Pi_i\}$) $\|\Pi(\rho-\sigma)\Pi\|_2^2>0$, then we immediately know that $\rho\neq \sigma$, and hence $F(\rho,\sigma)\leq 1-\eps$. In the following, for readability, let $\Delta:=(\bv_k^{\max}/\bv_k^{\min})^{1/2}$.

We argue formally in \Cref{lemma:test_approximate} that testing for equivalence with respect to $\{\hat{M}_i\}$ indeed can be used to test instead of $\{\Pi_i^+\}$. Specifically, we will show that if, for $\cMtest\ll 1$ a suitably small constant, all tests
\begin{align}
    \|\eqM_i(\rho-\sigma)\eqM_i\|_2^2\geq \frac{\cMtest \bv_i^{\max}\eps}{(2\Delta k^2(k+1))^2}
    \quad\text{or}\quad
    \eqM_i\rho \eqM_i=\eqM_i\sigma \eqM_i
\end{align}
indicate equivalence, then, with high probability, \eqref{eq:equiv_coseness_ideal_pi} holds, and hence $\rho=\sigma$ with high probability, and if any test indicates farness, we know $\rho\neq \sigma$.

To determine the sample complexity from using \Cref{idenl2:lemma:equiv_proj}, we note the following facts:
\begin{enumerate}
    \item Here we test to precision $\varepsilon':=\varepsilon \cMtest\bv_i^{\max}/(2 \Delta k^2(k+1))^2$,
    \item The maximal eigenvalue in $\eqM_i\sigma \eqM_i$ is given by $\|\eqM_i\sigma \eqM_i \|_{\opn}=\|\hat\Pi_{\max\{i-2,0\}}^+\eqM_i\sigma \hat\Pi_{\max\{i-2,0\}}^+ \|_{\opn}\leq 2\bv_{\max\{i-2,0\}}^{\max}$, which follows from the bucketing approach in \Cref{lemma:bucketing_guarantee}. 
\end{enumerate}

The sample complexity of testing bucket $i$ is then, for $x_i:=\min\{2d \bv_i^{\max},1\}$,  
\begin{align}
\label{eq:equiv:sc_test_i}
    N_{\text{test}}^{(i)}&\leq \cHSquantum\max\left\{\frac{\sqrt{d_{\Pi}}\min\{d_{\Pi}\bv_{\max\{i-2,0\}}^{\max},1\}}{\varepsilon'},\sqrt{\frac{d_{\Pi}}{\varepsilon'}}\right\}
    \\
    &\leq \frac{\cHSquantum k^4(k+1)^2}{\cMtest}\max\left\{\frac{\sqrt{d_{\Pi}}\Delta^{2}\min\{d_{\Pi}\bv_{\max\{i-2,0\}}^{\max},1\}}{\varepsilon\bv_i^{\max}},\sqrt{\frac{d_{\Pi}\Delta^{2}}{\varepsilon\bv_i^{\max}}}\right\}
    \\
    &\leq \widetilde{O}\left(\max\left\{\frac{d^{1/2}\Delta^2x_i}{\varepsilon\bv_i^{\max}},\frac{d}{\eps}\right\}\right),
\end{align}
where we used that for the last bucket, $\bv_{k}^{\max}=\Theta(\eta)$, and the fact that $\bv_{\max\{i-2,0\}}^{\max}/\bv_{i}^{\max}=\Theta(1)$. It will be clear that the term $d/\eps$ can never dominate, and we will drop it in the remaining discussion.
    
We will later see that $\eta=\widetilde{\Theta}(\eps/d^{1/4})$, such that both terms are in $\widetilde{O}(\min\{d^{9/4}/\eps, d^{3/2}/\eps^2\})$.
}

\color{blue}

\begin{figure}[h!]
        \centering
        \begin{tikzpicture}[scale=0.75]

\definecolor{myyellow}{RGB}{230,230,100}
\definecolor{myblue}{RGB}{120,170,220}
\definecolor{mygreen}{RGB}{120,230,120}
\definecolor{myred}{RGB}{255,0,0}

\newcommand{\drawSquare}[1]{%
    \draw[draw=black] (\wF,0) rectangle (0,\wF-#1);
    \draw[draw=black] (\wF,0) rectangle (#1,\wF);
}
\newcommand{\drawSquareColor}[2]{%
    \fill[fill=#1] (\wF,0) rectangle (#2,\wF-#2);
    \draw[draw=black] (\wF,0) rectangle (0,\wF-#2);
    \draw[draw=black] (\wF,0) rectangle (#2,\wF);
}

\newcommand{\drawSquareStripe}[2]{%
    \fill[pattern={Lines[angle=45, line width=1.25]}, pattern color=#1 ] (\wF,0) rectangle (#2,\wF-#2);
    \draw[draw=black] (\wF,0) rectangle (0,\wF-#2);
    \draw[draw=black] (\wF,0) rectangle (#2,\wF);
}

\newcommand{\drawSquareOnly}[3]{%
    \fill[fill=#1] (#3,\wF-#3) rectangle (#2,\wF-#2);
    \draw[draw=black] (\wF,0) rectangle (0,\wF-#2);
    \draw[draw=black] (\wF,0) rectangle (#2,\wF);
}

\def\wA{0.4}   
\def\wB{1.1}  
\def\wC{1.4}
\def\wD{2.4}   
\def\wE{3}  
\def\wF{3.8}

\drawSquare{0}
\drawSquare{\wA}
\drawSquare{\wB}
\drawSquareOnly{myyellow}{\wC}{\wD}
\drawSquareColor{myblue}{\wD}{\wF}
\drawSquare{\wE}

\draw[thick,->] (4.8,1.9) -- (7,1.9);

\begin{scope}[xshift=8cm]
\drawSquare{0}
\drawSquareColor{myyellow}{\wA}
\drawSquareStripe{myred}{\wA}
\drawSquare{\wB}
\drawSquareColor{myyellow}{\wC}
\drawSquareStripe{myblue}{\wD}
\drawSquare{\wE}

\end{scope}

\end{tikzpicture}
                \caption{Comparing the buckets from identity testing and equivalence testing. A bucket $S$ (yellow) is expanded to cover the entire remainder of the state, i.e., the area with lower eigenvalues as well (blue). We further include two earlier buckets (red, which also covers the off-diagonals). The actual $M$ is then an approximation to this new bucket.}
        \label{fig:buckets}
\end{figure}
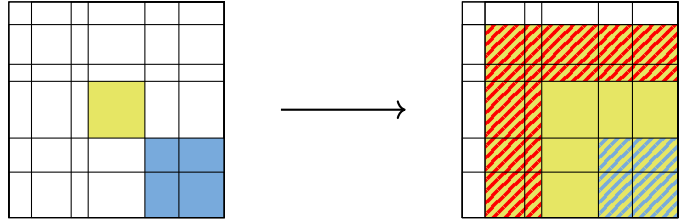

\color{black}

We now prove that we can indeed replace testing using $\{\Pi_i^+\}$ by testing with $\{\eqM_i\}$.
\begin{lemma}
\label{lemma:test_approximate}
    For a suitably large implicit constant in $N_{\textnormal{learn}}$, where $\alpha=O(1/k^3)$ as defined in \Cref{lemma:bucketing_guarantee}. Suppose we test, for all $\eqM_i$, whether $\|\eqM_i(\rho-\sigma)\eqM_i\|_2^2\geq \bv_i^{\max}\cMtest\varepsilon/(2\Delta k^2(k+1))^2$ or $\eqM_i\rho \eqM_i=\eqM_i\sigma \eqM_i$, for $\cMtest$ a universal constant, using the bucketing produced by \Cref{alg:equiv_bucketing}, to precision $\delta=10^{-4}/k$ each. Suppose all tests indicate equivalence. Then, with probability at least $0.99$, it holds that
    \begin{equation}
    \label{eq:test_approximate_eq}
        \forall i: \frac{\bv_i^{\min}\varepsilon}{4(k+1)^2} \geq \|\Pi_i^+(\rho-\sigma)\Pi_i^+\|_2^2.
    \end{equation}
\end{lemma}

\begin{proof}
    For the following, recall that by \eqref{eq:equiv:bucketing}, our bucketing is given by $\eqM_i:=\hat{\Pi}_{\max\{i-2,0\}}^+=\sum_{j=\max\{i-2,0\}}^{k}\hat{\Pi}_j$ and $\eqM_i^-:=\sum_{j=0}^{i-3}\hat{\Pi}_j$, such that $\mathds{1}_d=\eqM_i+\eqM_i^-$. We also note that by a union bound, with probability at least $0.999$, no test is faulty. We then bound
    \begin{align}
        \|\Pi_i^+(\rho-\sigma)\Pi_i^+\|_2^2&=\|\Pi_i^+(\eqM_i+(\mathds{1}-\eqM_i))(\rho-\sigma)(\eqM_i+(\mathds{1}-\eqM_i))\Pi_i^+\|_2^2%
        \\
        &\leq 4\left(\|\eqM_i(\rho-\sigma)\eqM_i\|_2^2+2\|\Pi_i^+\eqM_i^-(\rho-\sigma)\eqM_i\Pi_i^+\|_2^2+\|\Pi_i^+\eqM_i^-(\rho-\sigma)\eqM_i^-\Pi_i^+\|_2^2\right).
    \end{align}
    The first term is bounded by our assumption, and it remains to show that the other terms are small enough such that \eqref{eq:test_approximate_eq} holds. Using \Cref{lemma:first_bucket}, \Cref{lemma:equiv:overlap}, and the guarantee from \Cref{lemma:bucketing_guarantee} (with $\delta$ chosen sufficiently small). Then, for $j<i-1$,
    it is not difficult to see that by a sufficiently large constant $\cOpMix$, we can bound 
    \begin{align}
        \forall j<i-1:\quad \|\Pi_i^{+}\hat{\Pi}_j\|_{\opn}\leq \cOpMix\sqrt{\frac{\eta}{\bv_j^{\min}}},
        \qquad 
        \|\hat{M}_j (\rho-\sigma) \hat{M}_j\|_2\leq \frac{\sqrt{\bv_j^{\max}\cMtest\varepsilon}}{2\Delta k^2(k+1)}.
    \end{align}
    Here we assumed that the constants in \eqref{eq:equiv_io_overlap} are such that $(1+2c\sqrt{\kappa})^mc\leq \cOpMix$, which requires $c=O(1/(\sqrt{\kappa}m))$, which is only polylogarithmic in the parameters.
    Thus,
    \begin{align}
        \|\Pi_i^+\eqM_i^-(\rho-\sigma)\eqM_i\Pi_i^+\|_2&\leq \sum_{j=0}^{i-3}\|\Pi_i^{+}\hat{\Pi}_j (\rho-\sigma) \eqM_i\Pi_i^+\|_2
        \\
        &\leq \sum_{j=0}^{i-3}\|\Pi_i^{+}\hat{\Pi}_j\|_{\opn}\|\hat{\Pi}_j\eqM_j (\rho-\sigma) \eqM_j\eqM_i\Pi_i^+\|_2
        \\
        &\leq \sum_{j=0}^{i-3}\|\Pi_i^{+}\hat{\Pi}_j\|_{\opn}\|\eqM_j (\rho-\sigma) \eqM_j\|_2
        \\
        &\leq \sum_{j=0}^{i-3}\cOpMix\sqrt{\frac{\bv_j^{\max}}{\bv_j^{\min}}}\frac{\sqrt{\eta\cMtest\varepsilon}}{\Delta k^2(k+1)}\leq \frac{\cOpMix\sqrt{\eta\cMtest\varepsilon}}{\Delta(k+1)}.
    \end{align}
    To bound the remaining term, note that for any $j$, $l$, $\hat\Pi_j\eqM_{\min\{j,l\}}=\hat\Pi_j$ and $\hat\Pi_l\eqM_{\min\{j,l\}}=\hat\Pi_l$.
    \begin{align}
        \|\Pi_i^{+}\eqM_i^-(\rho-\sigma)\eqM_i^-\Pi_i^{+}\|_2&\leq \sum_{j,l=0}^{i-3}\|\Pi_i^{+}\hat{\Pi}_j (\rho-\sigma) \hat{\Pi}_l\Pi_i^{+}\|_2
        \\
        &\leq \sum_{j,l=0}^{i-3}\|\Pi_i^{+}\hat{\Pi}_j\|_{\opn} \|\hat{\Pi}_j(\rho-\sigma)\hat{\Pi}_l\|_2 \|\hat{\Pi}_l\Pi_i^{+}\|_{\opn}
        \\
        &\leq \sum_{j,l=0}^{i-3}\|\Pi_i^{+}\hat{\Pi}_j\|_{\opn} \|\eqM_{\min\{j,l\}}(\rho-\sigma)\eqM_{\min\{j,l\}}\|_2 \|\hat{\Pi}_l\Pi_i^{+}\|_{\opn}
        \\
        &\leq \sum_{j,l=0}^{i-3}\frac{(\cOpMix)^2}{k^2(k+1)}\sqrt{\frac{\eta}{\bv_j^{\min}} \frac{\cMtest\max\{\bv_j^{\max},\bv_l^{\max}\}\varepsilon}{\Delta^2} \frac{\eta}{\bv_l^{\min}}}
        \\
        & \leq \sum_{j,l=0}^{i-3}\frac{(\cOpMix)^2}{ k^2(k+1)}\sqrt{\frac{2\cMtest \eta}{\Delta^2\min\{\bv_j^{\min},\bv_l^{\min}\}}}\sqrt{\eta \varepsilon}\leq \frac{(\cOpMix)^2\sqrt{\eta\cMtest  \varepsilon}}{\Delta(k+1)}.
    \end{align}
    Finally, note that all tests accepting guarantees that $\|\eqM_i(\rho-\sigma)\eqM_i\|_2^2< \bv_i^{\min}\cMtest \varepsilon/(k+1)^2$
    Choosing the implicit constant small enough, we thus found that, for suitably chosen $\cMtest$, 
    \begin{align}
        \|\Pi_i^+(\rho-\sigma)\Pi_i^+\|_2^2\leq 4\|\eqM_i(\rho-\sigma)\eqM_i\|_2^2+4\cdot 9\left(\frac{[(\cOpMix)^2+\cOpMix]\sqrt{\eta\cMtest  \varepsilon}}{\Delta(k+1)}\right)^2\leq \frac{\bv_i^{\min}\eps}{4(k+1)^2}.
    \end{align}
    In the last step, we used that $\eta/\Delta^2\leq 2\bv_k^{\max}/\Delta^2= 2\bv_k^{\min}$. After a constant number of repetitions to boost the probability of success, this implies our claimed result.
\end{proof}

\subsection{Balancing the Costs of Learning and Testing}
\label{sec:equiv:balancing}
As mentioned before, our overall approach is to suitably tune the approximation threshold $\eta$ such that the sample complexity of approximate learning via bucketing and the sample complexity of testing are equal.

For the approximate bucketing, similar to the case of certification, the final sample complexity is dominated by the second term in \Cref{lemma:learn_buckets_chl} as is easily checked. For the following, let $x_i:=\min\{2d\bv_i^{\max},1\}\geq \tr[\hat\Pi_i^+\sigma]$ by \eqref{eq:op_bound}.  Then, by \Cref{lemma:learn_buckets_chl}, the cost of learning for bucket $i$ is
\begin{equation}\label{eqn:samplecomplexitylearning}
    \cOpLearn\sqrt{\frac{d \cdot \tr[\hat\Pi_i^+\sigma]}{N_{\text{learn}}^{(i)}}}\leq \cOpLearn\sqrt{\frac{x_i d}{N_{\text{learn}}^{(i)}}}\leq\sqrt{\eta\bv_i^{\max}}
    \quad\Leftrightarrow\quad 
    N_{\text{learn}}^{(i)}\geq \widetilde{O}\left(\frac{x_i d}{\eta\bv_i^{\max}}\right),
\end{equation}
and we may bound the maximal cost of learning a bucket by the cost of the last bucket we learn, where $\bv_i^{\max}=\eta$, which we simply denote by $N_{\text{learn}}\leq(\cOpLearn)^2dx/\eta^2$ (and let $x_k:=\min\{2d\eta,1\}$ be the $x$ of the last bucket).
We can similarly analyze the cost for testing the buckets. Now for the sample complexity of testing, by \Cref{idenl2:lemma:equiv_proj}, we can say that
\begin{align}\label{eqn:samplecomplexitytesting}
    N_{\text{test}}^{(i)}\leq \widetilde{O}\left(\frac{d^{1/2}x_i}{\varepsilon\bv_i^{\max}/\Delta^2}\right).
\end{align}
The bound for testing the last bucket will clearly dominate this sample complexity. Therefore, the respective costs of learning and testing are
\begin{align}\label{eqn:samplecomplexitytestinglearning}
    N_{\text{learn}}\leq \widetilde{O}\left(\frac{dx_{k-1}}{\eta^2}\right),
    \qquad
    N_{\text{test}}=\widetilde{O}\left(\frac{d^{1/2}x_k \eta/(\eps/d)}{\eps\eta}\right)=
    \widetilde{O}\left(\frac{d^{3/2}x_k}{\eps^2}\right).
\end{align}
Finally, to minimize the total sample complexity of our equivalence testing, we equate the sample complexities of learning and testing as in \eqref{eqn:samplecomplexitytestinglearning}, using that $x_{k-1}\leq \Theta(x_k)$
\begin{equation}
    N_{\text{test}}=N_{\text{learn}}\implies \eta=\frac{\cOpLearn}{(\cHSquantum)^{1/2}}\frac{\varepsilon}{d^{1/4}}.
\end{equation}
By setting the above value of $\eta$, we obtain the final sample complexity scales as:
\begin{equation}
    N=\widetilde{O}\left(\min\left\{\frac{d^{3/2}}{\varepsilon^2},\frac{d^{9/4}}{\varepsilon}\right\}\right),
\end{equation}
Note that the above mentioned sample complexity shows an improvement over testing in trace distance for the regime $\varepsilon<1/d^{3/4}$.

\color{blue}

\color{black}

We can now conclude our claimed result.
\fideqtesting*
\begin{proof}
    In \Cref{sec:equiv:learning}, we have shown how to learn the buckets approximately. Then, in \Cref{lemma:test_approximate}, we proved that we can use the approximate buckets $\{M_i\}$ as substitutes for the actual buckets to perform testing and find a test which fails in case $F(\rho,\sigma)\leq 1-\eps$, since then
    \begin{align}
        \eps\leq 1-F(\rho,\sigma) \leq 2(k+1)\sum_{j=0}^k\frac{\|\Pi_{j}^+(\rho-\sigma)\Pi_{j}^+\|_2^2}{\bv^{\min}_{j}},
    \end{align}
    whereas a successful testing according to \Cref{lemma:test_approximate} guarantees that $1-F(\rho,\sigma)\leq \eps/2$. The sample complexity derived above does indeed bound the costs of learning and testing as mentioned in \Cref{sec:equiv:balancing},
    \begin{align}
        \forall i\in [k-1]_0: N=\widetilde{O}\left(\min\left\{\frac{d^{3/2}}{\varepsilon^2},\frac{d^{9/4}}{\varepsilon}\right\}\right)=\widetilde{O}\left(\frac{d}{\eta}\min\left\{\frac{1}{\eta},d\right\}\right)\geq \widetilde{O}\left(\frac{d\cdot \min\left\{1,d\bv_i^{\max}\right\}}{\bv_i^{\min}\eta}\right)=N_{\text{learn}}^{(i)},
    \end{align}
    and $N\geq N_{\text{test}}^{(i)}$ (see \eqref{eq:equiv:sc_test_i}) always. 
    We test in total $k+1$ buckets. Since $k=\lceil \log \frac{d}{\eps} \rceil$, this only leads to additional polylogarithmic factors. This completes the proof of our equivalence testing.
\end{proof}

\section{Mutual Information Testing}
\label{sec:equivalence_mi}

In this section, we show how to perform mutual information testing. Our main result is as follows:

\mitesting*

Our approach is to first reduce the problem to independence testing in (in)fidelity. This problem in turn is solved with a specialized version of equivalence testing between $\rho_{AC}$ and $\rho_A\otimes\rho_C$, which makes use of the product structure of $\rho_A\otimes\rho_C$.

\color{blue}

\color{black}

\subsection{Reduction to Independence Testing in Fidelity}

Consider two quantum states $\rho, \sigma$ with spectral decompositions $\rho=\sum_i\rho_i\dyad{\rho_i}$ and $\sigma=\sum_i\sigma_i\dyad{\sigma_i}$. Then the following holds~\cite{Flammia2024quantumchisquared}: 
\footnote{We note that the original statements are written in terms of the \emph{squared Hellinger distance}, which we bound by the fidelity.},
\begin{align}
    D(\rho\|\sigma)\leq 4\left(2+D_{\infty}^{\textnormal{Ren}}(\rho\|\sigma)\right)(1-F(\rho,\sigma)),
\end{align}
where
\begin{align}
    D_{\infty}^{\textnormal{Ren}}(\rho\|\sigma):=\max_{i,j|\braket{\rho_i}{\sigma_j}\neq 0} \ln\left(\frac{\rho_i}{\sigma_j}\right).
\end{align}
This can be used to bound the mutual information by fidelity, which directly implies the following result:
\begin{lemma}[{Follows from \cite[Thm.\ 4.1]{Flammia2024quantumchisquared}}]
Let $\rho_{AC} \in \mathbb{C}^{d_{AC} \times d_{AC}}$ be an unknown bipartite quantum state with $d_A\geq d_C$ and $\eps \in (0,1)$ be a parameter. Suppose there exists an algorithm that can distinguish between $\rho_{AC}=\rho_A\otimes\rho_C$ and $F(\rho_{AC},\rho_A\otimes\rho_C)\leq 1-\eta$ with probability at least $0.99$ by performing measurement on samples of $\rho_{AC}$ and $\rho_A\otimes\rho_C$, using at most $N$ samples. Then, after the preprocessing discussed in \Cref{sec:preprocessmitesting}, the algorithm can distinguish between the following: 
 \begin{align}
        \rho_{AC}=\rho_A\otimes\rho_C
        \quad\text{or}\quad I(A:C)_{\rho_{AC}}\geq\varepsilon.
\end{align}
with probability at least $0.99$, where $\eta=O(\eps/\log(d_Ad_C/\eps))$.

\end{lemma}

In the following subsections, we will derive an algorithm which performs independence testing in fidelity, stated as follows:

\begin{lemma}
Let $\sigma_{AC},\sigma_A\otimes\sigma_C \in \mathbb{C}^{d_{AC} \times d_{AC}}$ be unknown bipartite quantum states and $\eps \in (0,1)$ be a parameter. Then, using single-copy adaptive measurements, in order to distinguish whether
    \begin{align}
        \sigma_{AC}=\sigma_A\otimes\sigma_C
        \quad\text{or}\quad F({\sigma_{AC}, \sigma_A\otimes\sigma_C})\leq 1-\varepsilon
    \end{align}
with probability at least $0.99$, $\widetilde{O}(\min\{(d_Ad_C)^{3/2}/\varepsilon^2,d_A^{9/4}d_C^{3/4}/\varepsilon\})$ samples of $\sigma_{AC}$ and $\sigma_A\otimes\sigma_C$ are sufficient.
\end{lemma}

The above lemma follows from \eqref{eq:mi:sample_complexity}.

Finally, our mutual information testing result follows from combining the above two lemmas.

\mitesting*

\begin{remark}
    Technically speaking, we are always using two copies of $\rho_{AC}$ to create a sample of $\rho_A\otimes\rho_C$, even though the measurements are performed on a single state of dimension $d_Ad_C$. We consider this a reasonable assumption, but we assume it is also possible to simply learn the smaller subsystem, and then simulate measurements on $\rho_A\otimes\hat\rho_C$ from measurements on $\rho_A$ alone. This in particular requires to consider a robust version of \Cref{idenl2:lemma:equiv_proj} (note that the underlying \cite{pmlr-v291-seyfried25a} allows for robustness too). See also \cite{Flammia2024quantumchisquared}. %
\end{remark}

\subsection{Preprocessing Step}\label{sec:preprocessmitesting}

We preprocess our samples, $\rho_{AC}$ and $\sigma_A\otimes\sigma_C$ as follows. We then apply a channel (dependent on $\nu_A$, $\nu_C \in (0,\varepsilon/4)$ to be determined), which maps the samples $\sigma_{AC}$ to $\sigma_{AC}^{[\nu]}$ as follows:
\begin{align*}\label{eqn:mitestingchannel}
    \sigma_{AC}\mapsto \sigma_{AC}^{[\nu]}:=(1-\nu_A)(1-\nu_C)\sigma_{AC}+\underbrace{(1-\nu_A)\nu_C\sigma_A\otimes \frac{\mathds{1}_C}{d_C}+\nu_A(1-\nu_C)\frac{\mathds{1}_A}{d_A}\otimes \sigma_C + \nu_A\nu_C\frac{\mathds{1}_{AC}}{d_Ad_C}}_{=: X_{AC}}.
\end{align*}
Note in particular that this channel maps
\begin{align}
    \sigma_A\otimes\sigma_C\mapsto \left((1-\nu_A)\sigma_A+\nu_A\frac{\mathds{1}_A}{d_A}\right)\otimes\left((1-\nu_C)\sigma_C+\nu_C\frac{\mathds{1}_C}{d_C}\right),%
\end{align}
which means that when $\sigma_{AC}$ is a product state, $\sigma_{AC}^{[\nu]}$ will be a product state as well. To obtain preprocessed samples from $\sigma_A\otimes\sigma_C$, we simply take partial traces, i.e., $\tr_{C}[\sigma_{AC}] \otimes \tr_{A}[\sigma_{AC}]$. From \eqref{eqn:mitestingchannel}, 
we note that the minimal eigenvalues of the respective marginals $\sigma_A^{[\nu]}$ and $\sigma_C^{[\nu]}$ are $\nu_A/d_A$ and $\nu_C/d_C$, respectively. 

Let $\nu:=\nu_A+\nu_C-\nu_A\nu_C$. Note that $X_{AC}/\nu$ is a quantum state. Then, by \Cref{rem:min_ev}, we have
\begin{align}
     F(\sigma_{AC},\sigma_A\otimes\sigma_C)\leq 1-\eps\Rightarrow F\left((1-\nu)\sigma_{AC}+\nu\frac{X_{AC}}{\nu}, (1-\nu)\sigma_{A}\otimes\sigma_C+\nu\frac{X_{AC}}{\nu}\right)&\leq 1-\frac{\varepsilon}{2},
\end{align}
as long as, $\nu\leq \varepsilon/100$ holds.
As mentioned, we will set $\nu_A$ and $\nu_C$ later, via
\begin{align}
\label{eq:mi:nu_def}
    \nu_A:=\frac{\eps}{10^3},%
    \qquad
    \nu_C:=\frac{\eps}{10^3}.%
\end{align}
In particular, the minimal eigenvalues of $\sigma_A^{[\nu]}$ and $\sigma_C^{[\nu]}$ are then $\eps/(10^3d_A)$ and $\eps/(10^3d_C)$, respectively.

For equivalence testing, we kept the threshold $\eta$ as a variable which we resolved in the end by equating the costs of learning and testing to intuitively explain why the choice of $\eta$ is optimal. The same can be done here, but for simplicity, we let
\begin{align}
    \eta_A=O\left(\frac{\eps}{d_A^{1/4}d_C^{3/4}}\right)\quad\text{and}\quad \eta_C=O\left(\frac{\eps}{d_C}\right)
\end{align}
for suitable implicit constants. Note in particular that $\eta_A(d_C/d_A)^{3/4}=O(\eps/d_A)=O(\bv^{\min}_{A,k_A})$ and $\bv^{\min}_{C,k_C}=\Theta(\bv^{\max}_{C,k_C})$.

\subsection{Learning}
\label{sec:mi:learning}
We will first approximately learn the marginals $\sigma_A$ and $\sigma_C$ of the unknown state $\sigma_{AC}$ to precision $\eta_A$ and $\eta_C$, respectively, each using \Cref{alg:equiv_bucketing}. Thus, we will obtain sets $\{\hat\Pi_j^{A}\}$, $\{\hat\Pi_j^{C}\}$, and estimators $\{H_A(\rho_A,\hat\Pi_j^{A})\}$, $\{H_C(\rho_C,\hat\Pi_j^{C})\}$ such that following \Cref{lemma:bucketing_guarantee}, we have the following guarantees:
\begin{align}
    \forall j\in [k_A-1]_0:\|\hat{\Pi}_j^{A,+}(H_{A}(\sigma_A,\hat\Pi_j^{A,+})-\sigma_A)\hat{\Pi}_j^{A,+}\|_{\opn}\leq c\sqrt{\eta_A \bv_{A,j}^{\max}},
    \\
    \forall j\in [k_C-1]_0:\|\hat{\Pi}_j^{C,+}(H_{C}(\sigma_C,\hat\Pi_j^{C,+})-\sigma_C)\hat{\Pi}_j^{C,+}\|_{\opn}\leq c\sqrt{\eta_C \bv_{C,j}^{\max}}.
\end{align}
Following the arguments in the previous section, we note that the guarantees from \Cref{lemma:equiv:overlap} carry over here as well, such that we have
\begin{align}
\label{eq:mi:overlap_bound_op}
    \forall j< i-1:\quad \|\Pi_i^{A,+}\hat\Pi_j^A\|_{\opn}\leq \cOpMix\sqrt{\frac{\eta_A}{\bv_{A,j}^{\min}}},
    \quad
    \|\Pi_i^{C,+}\hat\Pi_j^C\|_{\opn}\leq \cOpMix\sqrt{\frac{\eta_C}{\bv_{C,j}^{\min}}},
\end{align}
where we again introduced a constant $\cOpMix\geq (1+2c\kappa)^kc$ as in \Cref{sec:equiv:testing}.
Moreover, note that by \eqref{eq:op_bound}
\begin{align}
\label{eq:mi:ev_op_bound}
    \|(\hat{\Pi}_i^{A,+}\otimes \hat{\Pi}_j^{C,+})(\sigma_A\otimes\sigma_C) (\hat{\Pi}_i^{A,+}\otimes\hat{\Pi}_j^{C,+})\|_{\opn}=\|\hat{\Pi}_i^{A,+}\sigma_A\hat{\Pi}_i^{A,+}\|_{\opn}\| \hat{\Pi}_j^{C,+}\sigma_C\hat{\Pi}_j^{C,+}\|_{\opn}\leq 4\bv_{A,i}^{\max}\bv_{C,j}^{\max}.
\end{align}

\subsection{Testing}
For simplicity, let us assume for the moment that
we know the state $\sigma_A\otimes \sigma_C$. This will help us in defining the ideal projections we would like to measure. Later in a second step, we will argue that the approximate projections $\{\hat\Pi_j^{A}\}$ and $\{\hat\Pi_j^{C}\}$ obtained from \Cref{sec:mi:learning} will approximate the ideal projections well enough.

Let $\{\ket{a_i}\}$ and $\{a_i\}$ denote the set of eigenvectors and eigenvalues of $\sigma_A$, and $\{\ket{c_i}\}$ and $\{c_i\}$ denote the set of eigenvectors and eigenvalues of $\sigma_C$, respectively. We can also bucket them analogous to \Cref{sec:id_testing}, i.e.,
\begin{align}
&0\leq i<k_A: S_i^A:=\left\{j|2^{-i-1}< a_j\leq 2^{-i}\right\},\qquad S_{k_A}^A:=\left\{j| a_j\leq 2^{-k_A}\right\}, \\
& 0\leq i<k_C: S_i^C:=\left\{j|2^{-i-1}< c_j\leq 2^{-i}\right\},\qquad S_{k_C}^C:=\left\{j| c_j\leq 2^{-k_C}\right\}.
\end{align}
Note that here the thresholds and the number of buckets will in general \emph{differ} for $A$ and $C$, that is, $\eta_A\neq \eta_C$ and $k_A\neq k_C$ (see \Cref{sec:mi:balancing}). For the following, let $k_{AC}:=(k_A+1)(k_C+1)$.

We can now write
\begin{align}
    \sigma_A\otimes\sigma_C=\sum_{i=1}^{d_A}\sum_{j=1}^{d_C}a_ic_j\dyad{a_ic_j}.
\end{align}
Now, let $\rho_{AC}$ be an arbitrary bipartite quantum state. Following \Cref{lemma:l2_proj_sum}, the bound on the $F(\rho_{AC},\sigma_A\otimes\sigma_C)$ may be rewritten as follows:
\begin{align}
\label{eq:mi:intermediate_bound:beginning}
    \frac{1-F(\rho_{AC},\sigma_A\otimes\sigma_C)}{2} &\leq\sum_{i,k=1}^{d_A}\sum_{j,l=1}^{d_C}\frac{|\bra{a_ic_j}\rho_{AC}\ket{a_kc_l}-\delta_{ik}\delta_{jl}a_ic_j|^2}{a_ic_j+a_kc_l}
    \\
    &= \sum_{u,v=0}^{k_A}\sum_{r,s=0}^{k_C}\sum_{i\in S^A_u,j\in S^C_r,k\in S^A_v,l\in S^C_s}\frac{|\bra{a_ic_j}\rho_{AC}\ket{a_kc_l}-\delta_{ik}\delta_{jl}a_ic_j|^2}{a_ic_j+a_kc_l}
    \\
    &\leq \sum_{u,v=0}^{k_A}\sum_{r,s=0}^{k_C}\frac{\sum_{i\in S^A_u,j\in S^C_r,k\in S^A_v,l\in S^C_s}|\bra{a_ic_j}\rho_{AC}\ket{a_kc_l}-\delta_{ik}\delta_{jl}a_ic_j|^2}{\bv_{A,u}^{\min}\bv_{C,r}^{\min}+\bv_{A,v}^{\min}\bv_{C,s}^{\min}}
    \\
    &=\sum_{u,v=0}^{k_A}\sum_{r,s=0}^{k_C}\frac{\|\Pi_u^A\otimes\Pi_r^C(\rho_{AC}-\sigma_A\otimes\sigma_C)\Pi_v^A\otimes\Pi_s^C\|_2^2}{\bv_{A,u}^{\min}\bv_{C,r}^{\min}+\bv_{A,v}^{\min}\bv_{C,s}^{\min}}.\label{eq:mi:intermediate_bound}
\end{align}
We will now bound \eqref{eq:mi:intermediate_bound} using suitable terms such that we can apply \Cref{idenl2:lemma:equiv_proj}. For this, we first define $\Pi_{ij}:=\Pi_i^A\otimes \Pi_j^C$ and $\Pi_{ij}^+:=\sum_{x=i}^{k_A}\sum_{y=j}^{k_C}\Pi_x^A\otimes \Pi_y^C=\Pi_i^{A,+}\otimes \Pi_j^{C,+}$. Let
\begin{align}
    \label{eq:mi:operator_def}\Pi^+_{ij,kl}:= \mathds{1}[\textnormal{supp}(\Pi_{ij}^++\Pi_{kl}^+)]=\Pi_{ij}^++\Pi_{kl}^+-\Pi_{ij}^+\Pi_{kl}^+,
    \quad\text{and}\quad  
    \bv_{ur,vs}^{\min}=\bv_{A,u}^{\min}\bv_{C,r}^{\min}+\bv_{A,v}^{\min}\bv_{C,s}^{\min},
\end{align}
and define $\bv_{ur,vs}^{\max}$ analogously.
\begin{figure}[h!]
    \centering
    \begin{tikzpicture}[scale=0.75]%

\colorlet{gridc}{black}
\colorlet{bluefill}{blue!50}
\colorlet{redfill}{red!50}
\colorlet{greenfill}{green!65}

\def\wA{1.3}   
\def\wB{2.2}
\def\wC{1.5}
\def\hA{1.3}
\def\hB{1.7}
\def\hC{2}

\def\W{\wA+\wB+\wC}
\def\H{\hA+\hB+\hC}

\begin{scope}[shift={(0,0)}]
    \draw[thick, ->, black] (0,\W) -- (0,-0.3);
    \draw[thick, ->, black] (0,\W) -- (\W+0.3,\W);
  \node[] at (\W+0.15,\W+0.45) {$C$};
  \node[] at (-0.45,-0.15) {$A$};

  \fill[fill=orange!15] (\wA,\hC)--(\wA,\hC+\hB)--(\wA+\hB,\hC+\hB);
  \draw[thick, ->, orange] (\wA,\hC+\hB) -- (\wA+0.85,\hC+\hB-0.85);
  \draw[thick, dashed, orange] (\wA-1.25,\hC-1.25) -- (\wA+3,\hC+3);
  \fill[redfill,opacity=.75] (\wA,\hC) rectangle (\wA+\wB,0);          %
  \fill[bluefill,opacity=.75] (\wA+\wB,\hC+\hB) rectangle (\wA+\wB+\wC,\hC); %
  \fill[bluefill,opacity=.75] (\wA+\wB,\hC) rectangle (\wA+\wB+\wC,0); %
  \fill[pattern={Lines[angle=45, line width=1.25]},pattern color=redfill] (\wA+\wB,\hC) rectangle (\wA+\wB+\wC,0);

  \draw[gridc] (0,0) rectangle ({\wA+\wB+\wC},{\hA+\hB+\hC});
  \draw[gridc] (\wA,0) -- (\wA,\hA+\hB+\hC);
  \draw[gridc] (\wA+\wB,0) -- (\wA+\wB,\hA+\hB+\hC);
  \draw[gridc] (0,\hC) -- (\wA+\wB+\wC,\hC);
  \draw[gridc] (0,\hC+\hB) -- (\wA+\wB+\wC,\hC+\hB);

  \node[] at (\wA,\hA+\hB+\hC+0.45) {$l$};
  \node[] at (\wA+\wB,\hA+\hB+\hC+0.45) {$j$};
  \node[] at (-0.45,\hC+\hB) {$i$};
  \node[] at (-0.45,\hC) {$k$};

\end{scope}

\begin{scope}[shift={({\wA+\wB+\wC+2.0},0)}]
     \draw[thick, ->, black] (0,\W) -- (0,-0.3);
    \draw[thick, ->, black] (0,\W) -- (\W+0.3,\W);
  \node[] at (\W+0.15,\W+0.45) {$C$};
  \node[] at (-0.45,-0.15) {$A$};
  
  \fill[bluefill,opacity=.75] (\wA,\hC+\hB) rectangle (\wA+\wB,\hC);        %
  \fill[bluefill,opacity=.75] (\wA+\wB,\hC+\hB) rectangle (\wA+\wB+\wC,\hC);%
  \fill[bluefill,opacity=.75] (\wA,\hC) rectangle (\wA+\wB,0);              %
\fill[bluefill,opacity=.75] (\wA+\wB,\hC) rectangle (\wA+\wB+\wC,0); %
  \fill[pattern={Lines[angle=45, line width=1.25]},pattern color=redfill] (\wA+\wB,\hC) rectangle (\wA+\wB+\wC,0);     %

  \draw[gridc] (0,0) rectangle ({\wA+\wB+\wC},{\hA+\hB+\hC});
  \draw[gridc] (\wA,0) -- (\wA,\hA+\hB+\hC);
  \draw[gridc] (\wA+\wB,0) -- (\wA+\wB,\hA+\hB+\hC);
  \draw[gridc] (0,\hC) -- (\wA+\wB+\wC,\hC);
  \draw[gridc] (0,\hC+\hB) -- (\wA+\wB+\wC,\hC+\hB);

  \node[] at (\wA,\hA+\hB+\hC+0.45) {$j$};
  \node[] at (\wA+\wB,\hA+\hB+\hC+0.45) {$l$};
  \node[] at (-0.45,\hC+\hB) {$i$};
  \node[] at (-0.45,\hC) {$k$};

\end{scope}

\end{tikzpicture}
    \caption{Illustration of $\Pi_{kl}^+$ (red), and $\Pi_{ij}^+$ (blue). $\Pi_{ij,kl}^+$ corresponds to the entire colored area, the subtracted term $\Pi_{kl}^+\Pi_{ij}^+$, the overlap of $\Pi_{kl}^+$ and $\Pi_{ij}^+$, is marked in red and blue. We cannot just choose $\Pi^+_{\min\{i,k\},\min\{l,j\}}$, as this might introduce eigenvalues which are much larger than $\max\{a_kc_l,a_ic_j\}$ in some instances (orange line, note that diagonal lines correspond to same eigenvalues). %
    }
    \label{fig:placeholder}
\end{figure}
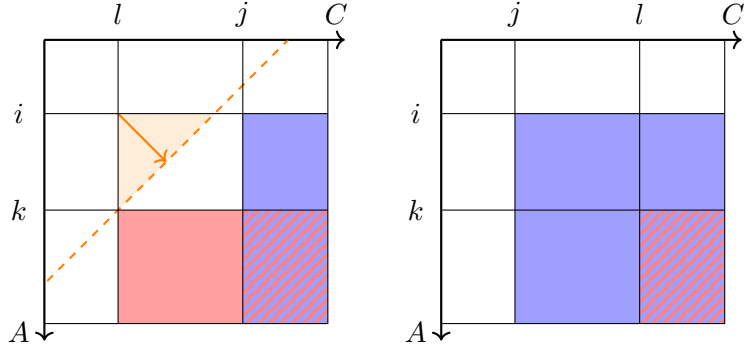

Note that the above relation holds for exact bucketing. However, as we do not know $\sigma_A \otimes \sigma_C$, we will again use an approximate bucketing $\{\hat{M}_{ij,kl}\}$ instead of $\{\Pi_{ij,kl}^+\}$. For this purpose, similar to equivalence testing, letting $\hat\Pi_{x,y}^+:=\hat\Pi_x^{A,+}\otimes \hat\Pi_y^{C,+}$, we define the following:
\begin{equation}
    \eqM_{ij}:=\hat\Pi_{\max\{0,i-2\},\max\{j-2,0\}}^+,\qquad \eqM_{ij,kl}:= \mathds{1}[\textnormal{supp}(\eqM_{ij}+\eqM_{kl})]=\eqM_{ij}+\eqM_{kl}-\eqM_{ij}\eqM_{kl}.
\end{equation}

In the same spirit as before, we aim to show that testing with our approximate buckets instead of the ideal ones is fine. The proof is in \Cref{lemma:mi_test_approximate}, and the underlying idea is equivalent to \Cref{lemma:test_approximate}, only a bit more care with the different indices is required. It is important to note that applying the testing routine from \Cref{idenl2:lemma:equiv_proj} would be costly if eigenvalues are too small. To avoid this, we apply a global depolarization channel to both $\sigma_{AC}$ and $\sigma_A\otimes\sigma_C$ of strength $\widetilde{O}(\eps)$. Letting $\zeta:=\widetilde{\Omega}(\eps/(d_Ad_C))$, this results in $\rho_{AC,\zeta}=(1-d_Ad_C\zeta)\rho_{AC}+\zeta \mathds{1}$ and $(\sigma_A\otimes\sigma_C)_{\zeta}=(1-d_Ad_C\zeta)\sigma_A\otimes\sigma_C+\zeta \mathds{1}$ with minimal eigenvalues of the order $\widetilde{O}(\eps/(d_Ad_C))$. We note that this doesn't change the eigenvectors of $\sigma_A\otimes\sigma_C$, therefore the bucketing for buckets above the threshold remains valid (up to constants), and all buckets below are now also of weight $\widetilde{\Omega}(\zeta)$. As in \eqref{eq:mi:intermediate_bound:beginning}, if either $F(\rho_{AC},\sigma_A\otimes\sigma_C)\leq 1-\eps$ or the two states are equal, we may bound 
\begin{align}
    \frac{1-F(\rho_{AC},\sigma_A\otimes\sigma_C)}{4} 
    &\leq \frac{1-F(\rho_{AC,\zeta},(\sigma_A\otimes\sigma_C)_{\zeta})}{2} 
    \\
    &\leq 
    \sum_{u,v=0}^{k_A}\sum_{r,s=0}^{k_C}\frac{(1-d_Ad_C\zeta)^2\|\Pi_{ur}(\rho_{AC}-\sigma_A\otimes\sigma_C)\Pi_{vs}\|_2^2}{\max\{\bv_{A,u}^{\min}\bv_{C,r}^{\min}+\bv_{A,v}^{\min}\bv_{C,s}^{\min},\zeta\}}
    \\
    &\leq 2\sum_{u,v=0}^{k_A}\sum_{r,s=0}^{k_C}\frac{\|\Pi_{ur}\Pi_{ur,vs}^+(\rho_{AC}-\sigma_A\otimes\sigma_C)\Pi_{ur,vs}^+\Pi_{vs}\|_2^2}{\bv_{ur,vs}^{\min}+\zeta}
    \\
    &\leq 2\sum_{u,v=0}^{k_A}\sum_{r,s=0}^{k_C}\frac{\|\Pi_{ur,vs}^+(\rho_{AC}-\sigma_A\otimes\sigma_C)\Pi_{ur,vs}^+\|_2^2}{\bv_{ur,vs}^{\min}+\zeta},
\end{align}
where we applied the Hölder inequality in the last step. If all tests are successful, allows us to directly conclude that with high probability, $F(\rho_{AC},\sigma_A\otimes\sigma_C)\geq 1-\eps/2$, which implies $\rho_{AC}=\sigma_A\otimes\sigma_C$. Note in particular that this depolarization does not need to be physical but only allowed us to skip handling small eigenvalues. We will be interested in the case where $\rho_{AC}=\sigma_{AC}$.
For simplicity we will also drop the $\zeta$ in the following. 
\begin{lemma}
    \label{lemma:mi_test_approximate}
    Assume it holds $\forall ((x,y),(u,v))\in ([k_A]_0\times [k_C]_0)^{\times 2}$ that %
    \begin{align}
    \label{eq:mi_equiv_bucketing}
        \|\eqM_{xy,uv}(\rho-\sigma)\eqM_{xy,uv}\|_2\leq \frac{\sqrt{\varepsilon\cMtest \max\{\bv_{xy,uv}^{\max}(d_C/d_A)^{3/4},\zeta\}}}{10^3k_{AC}^4},
    \end{align}
    and $\delta=0.001/k_{AC}^2$. Then it holds that with high probability,
    \begin{align}
    \label{eq:mi_equiv}
    \forall ((i,j),(k,l))\in ([k_A]_0\times [k_C]_0)^{\times 2}: 
    \quad 
    \|\Pi_{ij,kl}^+(\rho-\sigma)\Pi_{ij,kl}^+\|_2^2\leq \frac{\varepsilon \max\{\bv_{ij,kl}^{\min},\zeta\}}{16 k_{AC}^2}.
\end{align}
    
\end{lemma}
\begin{proof}
The chosen value of $\delta$ guarantees with a union bound that all tests are reliable with probability at least $0.999$. We first note that for any $X$,
\begin{align}
    \|\Pi_{ij,kl}^+X\|_{2}\leq \|\Pi_{ij}^+X\|_{2}+\|\Pi_{kl}^+X\|_{2}+\|\Pi_{kl}^+\|_{\opn}\|\Pi_{ij}^+X\|_{2}\leq 2\left(\|\Pi_{ij}^+X\|_{2}+\|\Pi_{kl}^+X\|_{2}\right).
\end{align}
Next, define 
\begin{align}
    Z:=\{(i,j),(k,l)\}\quad\text{and}\quad K:=\{(x,y)|\textnormal{supp}(\hat\Pi_{xy})\subseteq \textnormal{supp}(1-\eqM_{ij,kl})\}.
\end{align}
We then bound:
\begin{align}
    \|\Pi_{ij,kl}^+(\rho-\sigma)\Pi_{ij,kl}^+\|_2&= \|\Pi_{ij,kl}^+(\eqM_{ij,kl}+(\mathds{1}-\eqM_{ij,kl}))(\rho-\sigma)(\eqM_{ij,kl}+(\mathds{1}-\eqM_{ij,kl}))\Pi_{ij,kl}^+\|_2\label{eq:mi:buckets_0}
    \\
    &\leq \|\eqM_{ij,kl}(\rho-\sigma)\eqM_{ij,kl}\|_2
    \\
    &\quad+2\|\Pi_{ij,kl}^+(1-\eqM_{ij,kl})(\rho-\sigma)\eqM_{ij,kl}\Pi_{ij,kl}^+\|_2
    \\
    &\quad +\|\Pi_{ij,kl}^+(\mathds{1}-\eqM_{ij,kl})(\rho-\sigma)(\mathds{1}-\eqM_{ij,kl})\Pi_{ij,kl}^+\|_2
    \\
    &\leq \|\eqM_{ij,kl}(\rho-\sigma)\eqM_{ij,kl}\|_2
    \\
    &\quad+4\sum_{(x,y),(u,v),(r,s)\in Z}\sum_{(\alpha,\beta)\in K}
    \|\Pi_{xy}^+\hat\Pi_{\alpha\beta}(\rho-\sigma)\eqM_{rs}\Pi_{uv}^+\|_2 \label{eq:mi:buckets_1}
    \\
    &\quad+2\sum_{(x,y),(u,v)\in Z}\sum_{(\alpha,\beta),(\gamma,\zeta)\in K}\|\Pi_{xy}^+\hat\Pi_{\alpha\beta}(\rho-\sigma)\hat\Pi_{\gamma\zeta}\Pi_{uv}^+\|_2 \label{eq:mi:buckets_2}.
\end{align}
Analogous to before, we aim to show that the latter terms can be bounded by $c\varepsilon \sigma_{ij,kl}^{\min}$ for a suitably small constant $c$. %
Now, by assumption on our testing,
\begin{align}
    \|\hat \Pi_{xy}(\rho-\sigma)\eqM_{uv}\|_2\leq \|\eqM_{xy,uv}(\rho-\sigma)\eqM_{xy,uv}\|_2\leq \frac{\sqrt{\varepsilon\cMtest \max\{\bv_{xy,uv}^{\max}(d_C/d_A)^{3/4},\zeta\}}}{10^3k_{AC}^4}.
\end{align}
For the following, let $(x,y),(u,v),(r,s)\in Z$ and $(\alpha,\beta)\in K$ be arbitrary. Then
\begin{enumerate}
    \item First to bound \eqref{eq:mi:buckets_1}, using that $\Pi_{xy}^+:=\Pi_x^{A,+}\otimes \Pi_y^{C,+}$,
    \begin{align}
         &\|\Pi_{xy}^+(\hat\Pi_{\alpha}^{A}\otimes \hat\Pi_{\beta}^{C})(\rho-\sigma)\eqM_{rs}\Pi_{uv}^+\|_2 
         \\
         &\leq \|\Pi_{x}^{A,+}\hat\Pi_{\alpha}^{A}\|_{\opn}\|\Pi_{y}^{C,+}\hat\Pi_{\beta}^{C}\|_{\opn}\|\hat\Pi_{\alpha\beta}(\rho-\sigma)\eqM_{rs}\|_2
         \\
         &\leq \|\Pi_{x}^{A,+}\hat\Pi_{\alpha}^{A}\|_{\opn}\|\Pi_{y}^{C,+}\hat\Pi_{\beta}^{C}\|_{\opn}\frac{\sqrt{\varepsilon\cMtest\max\{\bv_{\alpha\beta, rs}^{\max}(d_C/d_A)^{3/4},\zeta\}}}{10^3k_{AC}^4}
         \\
         &\leq \frac{\sqrt{\eps\cMtest \|\Pi_{x}^{A,+}\hat\Pi_{\alpha}^{A}\|_{\opn}^2\|\Pi_{y}^{C,+}\hat\Pi_{\beta}^{C}\|_{\opn}^2\max\{(\bv_{A,\alpha}^{\max}\bv_{C,\beta}^{\max}+\bv_{A,r}^{\max}\bv_{C,s}^{\max})(d_C/d_A)^{3/4},\zeta\}}}{10^3k_{AC}^4}.
    \end{align}
    We have for all $y$, including the last bucket, that $\bv_{C,y}^{\max}\leq 2\bv_{C,y}^{\min}$ because of the choice for $\eta_C$.
    Since $y\in\{j,l\}$, we have $\bv_{C,y}^{\min}\leq \bv_{C,j}^{\min}+\bv_{C,l}^{\min}$. Next, we make a case distinction.
    \begin{itemize}
        \item Either $x-2\geq\alpha$, in which case we can use \eqref{eq:mi:overlap_bound_op} to bound
        \begin{align}
            \|\Pi_{x}^{A,+}\hat\Pi_{\alpha}^{A}\|_{\opn}\sqrt{\bv_{A,\alpha}^{\max}}
            &\leq \cOpMix\sqrt{\frac{\eta_A}{\bv_{A,\alpha}^{\min}}}\sqrt{\bv_{A,\alpha}^{\max}}
            \leq 2\cOpMix\sqrt{\eta_A}.%
        \end{align} 
        Here, we used that $\bv_{A,\alpha}^{\max}/\bv_{A,\alpha}^{\min}$ is bounded since $\alpha$ cannot be the last bucket. With $\eta_A(d_C/d_A)^{3/4}=\eps/d_A\leq \min\{\bv_{A,i}^{\min},\bv_{A,k}^{\min}\}$.
        \item Otherwise, if $x-2<\alpha$, we can simply bound 
        \begin{align}
            \|\Pi_{x}^{A,+}\hat\Pi_{\alpha}^{A}\|_{\opn}\sqrt{\bv_{A,\alpha}^{\max}}\leq 2\sqrt{\bv_{A,\alpha}^{\max}}.
        \end{align}
        Since $x\in\{i,k\}$ ($x=i$ if $y=j$ and $x=k$ when $y=l$), we have that $\bv_{A,\alpha}^{\max}\bv_{C,y}^{\min}(d_C/d_A)^{3/4}\leq 4(\bv_{A,i}^{\min}\bv_{C,j}^{\min}+\bv_{A,k}^{\min}\bv_{C,l}^{\min})$.
        And for $C$ similarly. In particular. Combined, we find
        \begin{align}
            \|\Pi_{x}^{A,+}\hat\Pi_{\alpha}^{A}\|_{\opn}\|\Pi_{y}^{C,+}\hat\Pi_{\beta}^{C}\|_{\opn}\sqrt{\bv_{A,\alpha}^{\max}\bv_{C,\beta}^{\max}(d_C/d_A)^{3/4}}%
            &\leq 4\sqrt{\bv_{A,i}^{\min}\bv_{C,j}^{\min}+\bv_{A,k}^{\min}\bv_{C,l}^{\min}}.
        \end{align}
    \end{itemize}
    Thus, in both cases, we can bound
    \begin{align}
        \|\Pi_{x}^{A,+}\hat\Pi_{\alpha}^{A}\|_{\opn}\|\Pi_{y}^{C,+}\hat\Pi_{\beta}^{C}\|_{\opn}\sqrt{\bv_{A,\alpha}^{\max}\bv_{C,\beta}^{\max}(d_C/d_A)^{3/4}}&\leq 4\sqrt{\bv_{A,i}^{\min}\bv_{C,j}^{\min}+\bv_{A,k}^{\min}\bv_{C,l}^{\min}}.
    \end{align}
    With $(r,s)\in\{(i,j),(k,l)\}$ we also directly have that $\bv_{A,r}^{\max}\bv_{C,s}^{\max}(d_C/d_A)^{3/4}\leq 4(\bv_{A,i}^{\min}\bv_{C,j}^{\min}+\bv_{A,k}^{\min}\bv_{C,l}^{\min})$. Thus, we have that
    \begin{align}
        \|\Pi_{xy}^+\hat\Pi_{\alpha\beta}(\rho-\sigma)\eqM_{rs}\Pi_{uv}^+\|_2%
        &\leq 
        4\frac{\sqrt{\eps\cMtest\max\{\bv_{A,i}^{\min}\bv_{C,j}^{\min}+\bv_{A,k}^{\min}\bv_{C,l}^{\min},\zeta\}}}{10^3k_{AC}^4}%
    \end{align}
    \item For bounding the term in \eqref{eq:mi:buckets_2}, by a similar argument,
    \begin{align}
        &\|\Pi_{xy}^+\hat\Pi_{\alpha\beta}(\rho-\sigma)\hat\Pi_{\gamma\zeta}\Pi_{uv}^+\|_2
        \\
        &\leq \|\Pi_{x}^{A,+}\hat\Pi_{\alpha}^{A}\|_{\opn}\|\Pi_{y}^{C,+}\hat\Pi_{\beta}^{C}\|_{\opn}\|\hat\Pi_{\alpha\beta}(\rho-\sigma)\hat\Pi_{\gamma\zeta}\|_2\|\Pi_{u}^{A,+}\hat\Pi_{\gamma}^{A}\|_{\opn}\|\Pi_{v}^{C,+}\hat\Pi_{\zeta}^{C}\|_{\opn}
        \\
        &\leq \|\Pi_{x}^{A,+}\hat\Pi_{\alpha}^{A}\|_{\opn}\|\Pi_{y}^{C,+}\hat\Pi_{\beta}^{C}\|_{\opn}\frac{\sqrt{\varepsilon\cMtest\max\{\bv_{\alpha\beta,\gamma\zeta}^{\max}(d_C/d_A)^{3/4},\zeta\}}}{10^3k_{AC}^4}\|\Pi_{u}^{A,+}\hat\Pi_{\gamma}^{A}\|_{\opn}\|\Pi_{v}^{C,+}\hat\Pi_{\zeta}^{C}\|_{\opn}
        \\
        & \leq 16\frac{\sqrt{\varepsilon\cMtest\max\{\bv_{ij,kl}^{\min},\zeta\}}}{10^3k_{AC}^4}.
    \end{align}
\end{enumerate}
Inserted into \eqref{eq:mi:buckets_0}, we find
\begin{align}
    \|\Pi_{ij,kl}^+(\rho-\sigma)\Pi_{ij,kl}^+\|_2&\leq \|\hat M_{ij,kl}^+(\rho-\sigma)\hat M_{ij,kl}^+\|_2+256k_{AC}^2\frac{\sqrt{\varepsilon\cMtest\max\{\bv_{ij,kl}^{\min},\zeta\}}}{10^3k_{AC}^4}
    \\
    &\leq \frac{1}{20}\sqrt{\frac{\varepsilon \max\{\bv_{ij,kl}^{\min},\zeta\}}{2k_{AC}^2}}.
\end{align}

\end{proof}

\subsection{The Costs of Learning and Testing}
\label{sec:mi:balancing}

We bound the cost of testing first, using \Cref{idenl2:lemma:equiv_proj}. We will later see we can ignore the second case, '$\sqrt{d_{\Pi}/\eps'}$', $\eps':=\eps \bv_{ij,kl}^{\max} (d_C/d_A)^{3/4}$ in the sample complexity of \Cref{idenl2:lemma:equiv_proj}, as it will never be dominating, which we do here directly for readability. First, note that we have
\begin{align}
    \tr[\eqM_{ij,kl}\sigma_A\otimes\sigma_C]&\leq O\left(\max_{(x,y)\in\{(i,j),(k,l)\}}\tr[\hat\Pi_{\max\{x-2,0\}}^{A,+}\sigma_A]\tr[\hat\Pi_{\max\{y-2,0\}}^{C,+}\sigma_C]\right)
    \\
    &\leq O\Bigg(\max_{(x,y)\in\{(i,j),(k,l)\}}\min\left\{d_A\|\hat\Pi_{\max\{x-2,0\}}^{A,+}\sigma_A\hat\Pi_{\max\{x-2,0\}}^{A,+}\|_{\opn},1\right\}
    \\
    &\qquad\qquad\cdot \min\left\{d_C\|\hat\Pi_{\max\{y-2,0\}}^{C,+}\sigma_C\hat\Pi_{\max\{y-2,0\}}^{C,+}\|_{\opn},1\right\}\Bigg)
    \\
    &\leq O\left(\max_{(x,y)\in\{(i,j),(k,l)\}}\min\left\{d_A\bv_{A,x}^{\max},1\right\}\min\left\{d_C\bv_{C,y}^{\max},1\right\}\right).
\end{align}
We can use this to bound the cost of testing in \Cref{idenl2:lemma:equiv_proj} by
\begin{align}
    N_{\textnormal{test}}^{ij,kl}&\leq \widetilde{O}\left(\frac{\sqrt{d_Ad_C}\tr[\eqM_{ij,kl}\sigma_A\otimes\sigma_C]}{\varepsilon\bv_{ij,kl}^{\max}(d_C/d_A)^{3/4}}\right)
    \\
    &\leq \widetilde{O}\left(\frac{\sqrt{d_Ad_C}}{\varepsilon}\max_{(x,y)\in\{(i,j),(k,l)\}}\frac{\min\{d_A\bv_{A,x}^{\max},1\}\min\{d_C\bv_{C,y}^{\max},1\}}{\bv_{A,x}^{\max}\bv_{C,y}^{\max}(d_C/d_A)^{3/4}}\right)
    \\
    &\leq \widetilde{O}\left(\frac{\sqrt{d_Ad_C}}{\varepsilon}\max_{(x,y)\in\{(i,j),(k,l)\}}\frac{d_C\min\{d_A\bv_{A,x}^{\max},1\}}{\bv_{A,x}^{\max}(d_C/d_A)^{3/4}}\right)
    \\
    &\leq \widetilde{O}\left(\frac{d_A^{5/4}d_C^{3/4}}{\varepsilon}\max_{(x,y)\in\{(i,j),(k,l)\}}\min\left\{d_A,\frac{1}{\eta_A}\right\}\right)
    \leq \widetilde{O}\left(\min\left\{\frac{(d_Ad_C)^{3/2}}{\eps^2},\frac{d_A^{9/4}d_C^{3/4}}{\eps}\right\}\right).\label{eq:mi_sc_raw}
\end{align}

The cost of learning is bounded analogously as for equivalence testing. With $x_{A,l}:=O(\min\{d_A\bv_{A,l}^{\max},1\})$,
\begin{equation}
    \sqrt{\eta_A\bv_{A,l}^{\max}}\geq \widetilde{O}\left(\sqrt{\frac{x_{A,l}d_A}{N_{\text{learn}}^{A, (l)}}}\right)\quad\Leftrightarrow\quad N_{\text{learn}}^{A, (l)}=\widetilde{O}\left(\frac{x_{A,l}d_A}{\eta_A\bv_{A,l}^{\max}}\right),
\end{equation}
it is clear that the last bucket we learn (with precision $\eta_A$) is most costly, with
\begin{align}
    N_{\text{learn}}^{A,(k_A-1)}=\widetilde{O}\left(\frac{x_Ad_A}{\eta_A^2}\right)&=\widetilde{O}\left(d_A\min\{d_A/\eta_A,1/\eta_A^2\}\right)
    \\
    &\leq \widetilde{O}\left(\min\left\{\frac{(d_Ad_C)^{3/2}}{\varepsilon^2},\frac{d_A^{9/4}d_C^{3/4}}{\varepsilon}\right\}\right).
\end{align}
In particular, this always dominates the cost of learning $\sigma_C$, which is in $\widetilde{O}(d_C^3/\eps)$.
The costs between the learning and testing tasks are therefore balanced, up to logarithmic factors, and we found an overall sample complexity of
\begin{equation}
    \label{eq:mi:sample_complexity}N=\widetilde{O}\left(\min\left\{\frac{(d_Ad_C)^{3/2}}{\varepsilon^2},\frac{d_A^{9/4}d_C^{3/4}}{\varepsilon}\right\}\right).
\end{equation}

\section{Nonadaptive Lower Bounds for Equivalence Testing in Qubits}
\label{sec:lower_equiv_fidelity_qubits}
In this section, we provide lower bounds for testing whether two qubit states are the same or are at least gapped away by $\eps$ in fidelity.

\lbnonadaptive*

The general setup follows \cite{bubeck_entanglement_2020} and is stated below:
\begin{enumerate}
    \item A measurement schedule $M$ defines a series of $N$ measurements, without loss of generality we can assume them to be rank one states. Thus, this is a list of $N$ measurements $M_i$. The set of all possible measurement schedules is denoted by $\mathcal{M}$.
    \item For a given measurement schedule $M$, a result $R(M)$ is a list of length $N$, containing possible measurement outcomes, i.e., $R_i\in M_i$. The set of possible results for a given $M$ is $\mathcal{R}(M)$.
    \item A property $\mathcal{S}_x$ is a subset of all states.
    \item For a given prior, the distribution over results for a property $\mathcal{S}_x$ is given by
    \begin{align}
        \left(p_x^N(M)\right)(R)&=\E_{\rho\in \mathcal{S}_x}\left[\Pr[\text{$M$ measured on $\rho$ results in $R$}]\right]
        \\
        &=\E_{\rho\in \mathcal{S}_x}\left[\prod_{i=1}^N\Pr[\text{$M_i$ measured on $\rho$ results in $R_i$}]\right].
    \end{align}
\end{enumerate}
For readability, we abbreviate $\Pr[R_i|\rho]:=\Pr[\text{$M_i$ measured on $\rho$ results in $R_i$}]$.
\begin{lemma}[{\cite[Lemma 2.8]{bubeck_entanglement_2020}}]
\label{lemma:dist_lower_bound}
    Let $\mathcal{M}$ be a set of measurement schedules, and let $\mathcal{S}_0$ and $\mathcal{S}_1$ be two properties. If
    \begin{align}
        \forall M\in \mathcal{M}:\|p_0^N(M)-p_1^N(M)\|_1\leq 1/3,
    \end{align}
    then it is not possible to reliably distinguish property $\mathcal{S}_0$ from $\mathcal{S}_1$ using $N$ samples.
\end{lemma}
In our case, this setting is slightly modified: we consider \emph{tuples} of states. The properties we consider contain elements of the form $\{(\rho_{\text{I}},\rho_{\text{II}})\}$. For $\mathcal{S}_0$, $\rho_{\text{I}}=\rho_{\text{II}}$, and for $\mathcal{S}_1$, we have $F(\rho_{\text{I}},\rho_{\text{II}})\leq 1-\eps$. Specifically, for a unitary $U$, we define
\begin{align}
    \rho_{0,U}&:=U\textnormal{diag}(1-\eps,\eps)U^{\dagger}=U\dyad{0}U^{\dagger}+\rho_{\eps,U},
    \\
    \rho_{1,U}&:=U\textnormal{diag}(1-2\eps,2\eps)U^{\dagger}=U\dyad{0}U^{\dagger}+2\rho_{\eps,U},
\end{align}
for $\rho_{\eps,U}=U\textnormal{diag}(-\eps,\eps)U^{\dagger}$. The distance in fidelity is simply 
\begin{align}
    F(\rho_{0,U},\rho_{1,U})&=\sqrt{(1-\eps)(1-2\eps)}+\sqrt{\eps(2\eps)}
    \\
    &=\sqrt{1-3\eps+2\eps^2}+\sqrt{2\eps^2}=1-\left(\frac{3}{2}-\sqrt{2}\right)\eps +O(\eps^2)=1-\Omega(\eps).
\end{align}
Thus, $1-F(\rho_{0,U},\rho_{1,U})\geq \Omega(\eps)$ and a rescaling satisfies the setting in \Cref{theo:lbnonadaptive}. To get an instance from $\mathcal{S}_0$, we sample a Haar-random unitary $U$ and then return samples of $(\rho_{0,U},\rho_{0,U})$, and $(\rho_{0,U},\rho_{1,U})$ for $\mathcal{S}_1$.

Each step of the measurement schedule now consists of two measurements. These are independent, such that 
\begin{align}
    \Pr[R_i|\rho_{\text{I}},\rho_{\text{II}}]=\Pr[R_{i,\text{I}}|\rho_{\text{I}}]\Pr[R_{i,\text{II}}|\rho_{\text{II}}].    
\end{align}
Note that due to the non-adaptive setting, we can naturally see the measurement schedule $\mathcal{A}$ as a tuple as well, $\mathcal{A}=(\mathcal{A}_{\text{I}},\mathcal{A}_{\text{II}})$. Together, we can bound
\begin{align}
    \|p_0^N(M)-p_1^N(M)\|_1&=\sum_{R\in \mathcal{R}(\mathcal{A})}\left|\int_{\mathcal{S}_0}\prod_{i}\PrOutcome{R_i}{\rho_{\text{I}},\rho_{\text{II}}}d\rho-\int_{\mathcal{S}_1}\prod_{i}\PrOutcome{R_i}{\rho_{\text{I}},\rho_{\text{II}}}d\rho\right|
    \\
    &=\sum_{R\in \mathcal{R}(\mathcal{A})}\left|\int_{U}\prod_{i}\PrOutcome{R_i}{\rho_{0,U},\rho_{0,U}}d\rho-\prod_{i}\PrOutcome{R_i}{\rho_{0,U},\rho_{1,U}}d\rho\right|
    \\
    &\leq \sum_{R\in \mathcal{R}(\mathcal{A})}\int_{U}\left|\prod_{i}\PrOutcome{R_i}{\rho_{0,U},\rho_{0,U}}-\prod_{i}\PrOutcome{R_i}{\rho_{0,U},\rho_{1,U}}\right|
    \\
    &= \sum_{R_{\text{I}}\in \mathcal{R}(\mathcal{A}_{\text{I}})}\sum_{R_{\text{II}}\in \mathcal{R}(\mathcal{A}_{\text{II}})}\int_{U}\prod_{i}\PrOutcome{R_{i,\textnormal{I}}}{\rho_{0,U}}\left|\prod_{i}\PrOutcome{R_{i,\textnormal{II}}}{\rho_{0,U}}-\prod_{i}\PrOutcome{R_{i,\textnormal{II}}}{\rho_{1,U}}\right|
    \\
    &= \sum_{R_{\text{II}}\in \mathcal{R}(\mathcal{A}_{\text{II}})}\int_{U}\left|\prod_{i}\PrOutcome{R_{i,\textnormal{II}}}{\rho_{0,U}}-\prod_{i}\PrOutcome{R_{i,\textnormal{II}}}{\rho_{1,U}}\right|.\label{eq:lower_bound_tvd_in_R}
\end{align}

Without loss of generality, we can assume that the measurement consists of weighted rank one projections. Note that we may write the state $\rho$ as $U((1-\eps)\dyad{0}+\eps\dyad{1})U^{\dagger}$. Given a POVM $M_i$, then, for each element $c_{i,j}\dyad{M_i,j}$, we can define the $\emph{bias}$
\begin{align}
    y_{M_i,j,U}:=|\bra{M_i,j}U\ket{0}|^2.%
\end{align}
Note that $y_{M_i,j,U}\in [0,1]$. %
In the following, for notational convenience, we will write for a set of measurements $\{M_i\}$, 
\begin{align}
    y_{i,j,U}:=y_{M_i,j,U},\qquad \eps_{i,j,U}=\bra{M_i,j}\rho_{\eps,U}\ket{M_i,j}.
\end{align}
The basic intuition for why the lower bound is $\widetilde{\Omega}(1/\eps^2)$ as opposed to the $O(1/\eps)$ in distribution testing is that using a non-adaptive schedule implies that with high probability, only a small fraction of the measurements is very well aligned with the actual basis: if only a $O(\eps)$ fraction of measurement is very strongly biased, which is very useful to detect the `noise' distinguishing $\rho_{0,U}$ from $\rho_{1,U}$, then we need $O((1/\eps)\cdot 1/\eps)=\widetilde{\Omega}(1/\eps^2)$ samples in order to distinguish. We will split the measurements into buckets, depending on their bias, and then bound the terms individually. We use the following fact:
\begin{align}
    \text{Let $\ket{\psi}$ be Haar random. Then, for $x\in[0,1]$ $\Pr[|\braket{0}{\psi}|^2\leq x]=x$}
\end{align}
Then we have the following lemma.
\begin{lemma}
\label{lemma:lower_nonadaptive_buckets}
    Suppose we are given an unknown, Haar random unitary $U$. Fix any measurement schedule consisting of $N$ measurements. Let $S_{0,U}:=\{(i,j)|  y_{i,j,U}+\eps_{i,j,U}\leq 2\eps\}$ and, for $k>0$, $S_{k,U}:=\{(i,j)| 2^{k}\eps\leq y_{i,j,U}+\eps_{i,j,U}\leq 2^{k+1}\eps\}$. Then, with probability $1-p_a$, it holds for all $k\in[1,2,...,\lceil\log(1/\eps)\rceil]$ that
    \begin{align}
        \sum_{(i,j)\in S_{k,U}}c_{i,j}\leq 12N\min\{2^{k+1}\eps,1-2^{k-1}\eps\}\log(1/\eps)/p_a,
    \end{align}
    and $\sum_{(i,j)\in S_{0,U}}c_{i,j}\leq 12N\eps\log(1/\eps)/p_a$.
\end{lemma}
\begin{proof}

Fix a $k>0$. Then 
\begin{align}
    \E_U\left[\sum_{(i,j)\in S_{k,U}}c_{i,j}\right]&= \sum_{i=1}^N\sum_{j\in M_i}c_{i,j}\E_U[1[2^{k}\eps \leq y_{i,j,U}+\eps_{i,j,U}\leq 2^{k+1}\eps]]
    \\
    &=\sum_{i=1}^N\sum_{j\in M_i}c_{i,j}\Pr[2^{k}\eps \leq y_{i,j,U}+\eps_{i,j,U}\leq 2^{k+1}\eps].
\end{align}
Since $U$ is Haar random, the precise choice of $\ket{M_i,j}$ is irrelevant. Instead, we can directly inspect an arbitrary Haar random state $\ket{\phi}$. Note that $\bra{\phi}\rho_{0,U}\ket{\phi}=|\braket{\phi}{0}|^2(1-2\eps)+\eps$ such that 
\begin{align}
    \E_U\left[\sum_{(i,j)\in S_{k,U}}c_{i,j}\right]&=\sum_{i=1}^N\sum_{j\in M_i}c_{i,j}\Pr[\frac{2^{k}-1}{1-2\eps}\eps \leq |\bra{\phi}\ket{0}|^2\leq \frac{2^{k+1}-1}{1-2\eps}\eps].
\end{align}
Therefore,
\begin{align}
    &\Pr[\frac{2^{k}-1}{1-2\eps}\eps \leq |\bra{\phi}\ket{0}|^2\leq \frac{2^{k+1}-1}{1-2\eps}\eps]\leq 
    \begin{cases}
        \Pr[ |\bra{\phi}\ket{0}|^2\leq 2^{k+2}\eps]\leq 2^{k+2}\eps,
        \\
        \Pr[1-|\bra{\phi}\ket{0}|^2\leq 1-2^{k-1}\eps]\leq 1-2^{k-1}\eps.
    \end{cases}
\end{align}
Thus,
\begin{align}
    \E_U\left[\sum_{(i,j)\in S_{k,U}}c_{i,j}\right]\leq \sum_{i=1}^N\sum_{j\in M_i}c_{i,j}\min\{2^{k+2}\eps,1-2^{k-1}\eps\}\leq 4N\min\{2^{k+1}\eps,1-2^{k-1}\eps\}.
\end{align}
Now we use Markov such that 
\begin{align}
    \Pr\left[\sum_{(i,j)\in S_{k,U}}c_{i,j}\geq 12N\min\{2^{k+1}\eps,1-2^{k-1}\eps\}\log(1/\eps)/p_a\right]\leq \frac{p_a}{3\log(1/\eps)}.
\end{align}
A union bound over all $\lceil\log(1/\eps)\rceil +1\leq 3\log(1/\eps)$ sets then gives the desired result, the statement for $k=0$ follows analogously.
\end{proof}

Throughout, we assume that 
\begin{align}
\label{eq:n_lower_qubit_bound}
    N\leq \frac{c_Np_a}{\eps^2\log^4(1/\eps)},
\end{align}
for $c_N$ a small constant. Since we aim to show an asymptotic statement, we may also assume that $\eps\leq \eps_0$ for some $\eps_0$ of our choice.

Thus, we can write \eqref{eq:lower_bound_tvd_in_R}
\begin{align}
\label{eq:lower_rewrite_tv}
    \|p_0^N(M)-p_1^N(M)\|_1&\leq\sum_{R_{\text{II}}\in \mathcal{R}(\mathcal{A}_{\text{II}})}\int_{U}\left|\prod_{i=1}^N\PrOutcome{R_{i,\text{II}}}{\rho_{0,U}}-\prod_{i=1}^N\PrOutcome{R_{i,\textnormal{II}}}{\rho_{1,U}}\right|
    \\
    &=\sum_{R_{\text{II}}\in \mathcal{R}(\mathcal{A}_{\text{II}})}\int_U\underbrace{\prod_{i=1}^N\PrOutcome{R_{i,\text{II}}}{\rho_{0,U}}}_{\Pr(R_{\text{II}},U):=}\underbrace{\left|\prod_{i=1}^N\frac{\PrOutcome{R_{i,\textnormal{II}}}{\rho_{1,U}}}{\PrOutcome{R_{i,\text{II}}}{\rho_{0,U}}}-1\right|}_{f(R_{\text{II}},U):=}.
    \label{eq:lower_rewrite_tv_2}
\end{align}
For readability, we will write $T:=R_{\text{II}}$, and $\mathcal{T}(U)$ the set of all possible outcomes $T$ for which $(U,T)$ is typical. The basic idea will now be to perform a case distinction. We denote all pairs $(U,T)$ for which  $\bra{M_i,j}\rho_{0,U}\ket{M_i,j}=y_{i,j,U}+\eps_{i,j,U}\geq\eps\log(1/\eps)$ and the bucketing of \Cref{lemma:lower_nonadaptive_buckets} holds as \emph{typical}. Then
\begin{align}
\label{eq:f_T_bound}
    \int_U dU \sum_{T\in \mathcal{R}(\mathcal{A}_{\text{II}})}f(T,U)\Pr[T,U]&\leq \int_U dU \sum_{T\in \mathcal{T}(U)}f(T,U)\Pr[T,U]+\int_U dU \sum_{T\notin \mathcal{T}(U)}f(T,U)\Pr[T,U]
\end{align}
To do this, we rewrite
\begin{align}
f(R_{\text{II}},U)=\left|\prod_{i=1}^N\frac{\PrOutcome{R_{i,\text{II}}}{\rho_{1,U}}}{\PrOutcome{R_{i,\text{II}}}{\rho_{0,U}}}-1\right|&=\left|\exp\left(\sum_{i=1}^N\log\left(\frac{\PrOutcome{R_{i,\textnormal{II}}}{\rho_{1,U}}}{\PrOutcome{R_{i,\text{II}}}{\rho_{0,U}}}\right)\right)-1\right|,
\\
&=\left|\exp\left(\sum_{i=1}^N\log\left(1+\frac{c_{i,j}\eps_{i,j,U}}{\PrOutcome{R_{i,\text{II}}}{\rho_{0,U}}}\right)\right)-1\right|.
\end{align}
In the following, we assume a fixed $U$. We denote the set of such $T_i$'s as $\mathcal{T}_i(U)$, and by $\mathcal{T}(U)$ the set of all typical transcripts, i.e., the product of the $\mathcal{T}_i(U)$. Each term
\begin{align}
    \forall T_i\in \mathcal{T}_i(U):~X_i(T_i,U):=\log\left(1+\frac{\eps_{i,j,U}}{y_{i,j,U}+\eps_{i,j,U}}\right)
\end{align}
is a random variable (we set $X_i(T_i,U)=0$ for $T_i\notin \mathcal{T}_i(U)$), and we will either apply the Bernstein inequality, (\Cref{lemma:bernstein} below) over $\sum_{i=1}^NX_i$, or argue that we don't see too heavily biased elements (see \eqref{eq:no_heavy_biased}). %
\begin{align}
\label{eq:lower_f}
\forall T\in \mathcal{T}(U):~f(T,U)&=%
    \Bigg|\underbrace{\exp\left(\E_{T'}\left[\sum_{i=1}^NX_i(T_i',U)\right]\right)}_{g_{\textnormal{typical}}^U:=}\prod_{i=1}^N\underbrace{\exp(X_i(T_i,U)-\E_{T'}[X_i(T'_i,U)])}_{g_{\textnormal{bias~}i}^U}-1\Bigg|.
\end{align}
It remains to show that the contribution of typical pairs to \eqref{eq:lower_rewrite_tv_2} can be made arbitrarily small by suitably choosing the implicit small constant in $N$. 
For convenience, we drop the $U$ in the following also note that $\sum_{T_i\in \mathcal{T}_i}\PrOutcome{R_{i,\textnormal{II}}}{\rho_{0,U}}=1-o(1)\neq 1$, in general.

\begin{lemma}
\label{lemma:helper_lower}
    Fix a $U$ for which \Cref{lemma:lower_nonadaptive_buckets} holds. Then it holds that
    \begin{align}
        \sum_{i=1}^N\sum_{T_i\in \mathcal{T}_i}\frac{(c_{i,j}\eps_{i,j,U})^2}{\PrOutcome{R_{i,\textnormal{II}}}{\rho_{0,U}}}\leq 2^8\eps^2\log(1/\eps)^2N/p_a.
    \end{align}
    Further, the probability that there is \emph{any} occurrence of a light element can be bounded by $O(\eps^2\log(1/\eps)^3N/p_a)$.
\end{lemma}
\begin{proof}
    \begin{align}
        \sum_{i=1}^N\sum_{T_i\in \mathcal{T}_i}\frac{(c_{i,j}\eps_{i,j,U})^2}{\PrOutcome{R_{i,\text{II}}}{\rho_{0,U}}}
        &\leq  2\sum_{i=1}^N\sum_{j\in M_i}\frac{c_{i,j}\eps_{i,j,U}^2}{y_{i,j,U}+\eps_{i,j,U}}%
        \\
        &\leq 2\sum_{k}\sum_{(i,j)\in S_k}\frac{c_{i,j}\eps^2}{y_{i,j,U}+\eps_{i,j,U}}%
        \\
        &\leq 4\eps^2\sum_{k}\sum_{(i,j)\in S_k}\frac{c_{i,j}}{2^{k}\eps}
        \leq 
        4\eps^2 \sum_{k} \frac{24\log(1/\eps)N2^k\eps}{2^{k}p_a\eps}\leq 
        \frac{2^8\eps^2\log(1/\eps)^2N}{p_a}.
    \end{align}
    For the second statement, the probability that there is \emph{any} occurrence of a light element can then be bounded by (let $k_*:=\lceil\log\log(1/\eps)\rceil-1$, such that $2^{k_*}\eps< \eps\log(1/\eps)\leq 2^{k_*+1}\eps$)
    \begin{align}
        &\sum_{i,j}c_{i,j}(y_{i,j,U}+\eps_{i,j,U})1[y_{i,j,U}+\eps_{i,j,U}\leq \eps\log(1/\eps)]
        \\
        &\leq \eps\log(1/\eps)\sum_{i,j}c_{i,j}1[y_{i,j,U}+\eps_{i,j,U}\leq \eps\log(1/\eps)]
        \\
        &\leq\eps\log(1/\eps)\sum_{k=0}^{k_*}\sum_{(i,j)\in S_k}c_{i,j}
        \\
        &\leq 48\eps^2\log(1/\eps)^3N/p_a, \label{eq:no_heavy_biased}
    \end{align}
    using \Cref{lemma:lower_nonadaptive_buckets} in the last step.
\end{proof}

For $g_{\textnormal{typical}}$, it holds that
\begin{restatable}[]{lemma}{lowerqubittypical}
\label{lemma:lower_qubit_typical}
Fix a $U$ for which \Cref{lemma:lower_nonadaptive_buckets} holds. Then
    \begin{align}
        g_{\textnormal{typical}}^U:=\exp\left(\E_{T}\left[\sum_{i=1}^NX_i(T_i,U)\right]\right)\in [1-4c_N,1+4c_N],
    \end{align}
    for $c_N$ the small constant in the definition of $N$ in \eqref{eq:n_lower_qubit_bound}.
\end{restatable}
\begin{proof}
    We bound 
    \begin{align}
        \E_{T}\left[\sum_{i=1}^NX_i(T_i,U)\right]%
        &= \sum_{i=1}^N\sum_{T_i\in \mathcal{T}_i}\PrOutcome{T_i}{\rho_{0,U}}\log\left(1+\frac{c_{i,j}\eps_{i,j,U}}{\PrOutcome{T_i}{\rho_{0,U}}}\right) %
        \\
        &\leq \sum_{i=1}^N\sum_{T_i\in \mathcal{T}_i}c_{i,j}\eps_{i,j,U}
        \\
        &\leq %
        \sum_{i=1}^N\sum_{j\in M_i}\tr[c_{i,j}\dyad{M_i,j}\rho_{\eps}]+\sum_{i=1}^N\sum_{T_i\notin \mathcal{T}_i}c_{i,j}|\eps_{i,j,U}|
        \\
        &\leq%
        \sum_{i=1}^N\tr[\rho_{\eps}]+\sum_{i=1}^N\sum_{T_i\notin \mathcal{T}_i}\PrOutcome{T_i}{\rho_{0,U}}\leq 48\frac{\eps^2\log(1/\eps)^3N}{p_a}.
    \end{align}
    On the other hand, using that for $x\geq -0.68$, $\log(1+x)\geq x-x^2$,
    \begin{align}
        &\sum_{i=1}^N\E_{T}\left[X_i(T_i,U)\right]
        \\
        &\geq \sum_{i=1}^N\sum_{T_i\in \mathcal{T}_i(U)}\PrOutcome{T_i}{\rho_{0,U}}\left(\frac{c_{i,j}\eps_{i,j,U}}{\PrOutcome{T_i}{\rho_{0,U}}}-\frac{(c_{i,j}\eps_{i,j,U})^2}{\PrOutcome{T_i}{\rho_{0,U}}^2}\right)\label{eq:lower_bound_negative_log}
        \\
        &=%
        \sum_{i=1}^N\sum_{T_i\in \mathcal{T}_i(U)}\left(c_{i,j}\eps_{i,j,U}-\frac{(c_{i,j}\eps_{i,j,U})^2}{\PrOutcome{T_i}{\rho_{0,U}}}\right)
        \\
        &=%
        -\sum_{i=1}^N\left(\sum_{T_i\notin \mathcal{T}_i(U)}c_{i,j}\eps_{i,j,U}+\sum_{T_i\in \mathcal{T}_i(U)}\frac{(c_{i,j}\eps_{i,j,U})^2}{\PrOutcome{T_i}{\rho_{0,U}}}\right)\geq -\frac{64\eps^2\log(1/\eps)^3N}{p_a},
    \end{align}
    using \Cref{lemma:helper_lower}.
\end{proof}
For buckets with a less severe bias, we use Bernstein's inequality.
\begin{lemma}[Bernstein inequality~{\cite{conc_ineqs_2013}}]
\label{lemma:bernstein}
    Let $X_1, \ldots ,X_n$ be independent random variables with finite variance such that $X_i\leq b$ for some $b>0$ almost surely for all $i \in [n]$. Let 
    \begin{align}
        X:=\sum_{i=1}^nX_i-\E[X_i],\quad v:=\sum_{i=1}^n\E[X_i^2].
    \end{align}
    Then, for any $y>0$,
    \begin{align}
        \Pr[X\geq y]\leq \exp\left(\frac{-y^2}{2v+2by/3}\right).
    \end{align}
\end{lemma}
In our case, we let $c$ be a constant, say $c=1/100$, and recall that
\begin{align}
    \forall T_i\in \mathcal{T}_i(U):~  X_i(T_i,U):=\log\left(1+\frac{c_{i,j}\eps_{i,j,U}}{\PrOutcome{T_i}{\rho_{0,U}}}\right).%
\end{align}
Note that the $X_i$ are bounded because
\begin{align}
     \forall T_i\in \mathcal{T}_i:~ X_i\leq \log\left(1+\frac{|c_{i,j}|\eps}{\PrOutcome{T_i}{\rho_{0,U}}}\right)\leq \frac{|c_{i,j}|\eps}{\PrOutcome{T_i}{\rho_{0,U}}}\leq \frac{1}{\log(1/\eps)},
\end{align}
and $X_i$ can also be lower bounded by $-2/\log(1/\eps)$.
\begin{align}
    v=\sum_{i=1}^N\E_{T}[X_i^2]
    &\leq 2\sum_{i=1}^N\sum_{T_i\in \mathcal{T}_i}\PrOutcome{T_i}{\rho_{0,U}}\left(\frac{c_{i,j}\eps_{i,j,U}}{\PrOutcome{T_i}{\rho_{0,U}}}\right)^2\leq %
    512\eps^2\log(1/\eps)^2N/p_a,
\end{align}
by \Cref{lemma:helper_lower}. The factor $2$ appears when bounding $X_i^2$ if $\eps_{i,j,U}$ is negative, see \eqref{eq:lower_bound_negative_log}. Note that
\begin{align}
\label{eq:x_to_bias}
    X=\log\left(\prod_{i=1}^N\exp(X_i(T_i,U)-\E_{T'}[X_i(T'_i,U)])\right)
\end{align}
We have 
\begin{align}
    \Pr[X\geq c]&\leq \exp\left(\frac{-c^2}{2v+2bc/3}\right)
    \\
    &\leq \exp\left(\frac{-c^2}{1024\eps^2\log(1/\eps)^2N/p_a+c/\log(1/\eps)}\right)
    \\
    &\leq \exp\left(-\frac{c\log(1/\eps)}{2}\right)=\eps^{c/2},
\end{align}
which can be made arbitrarily small by choosing $\eps_0$ appropriately as a function of $c$. By replacing $c$ with $kc$, we can find $\Pr[X\geq kc]\leq \eps^{kc/2}$, and by applying Bernstein to $-X_i$, we obtain $\Pr[-X\geq kc]\leq \eps^{kc/2}$ (note that $-X_i\leq 2/\log(1/\eps)$). Therefore, using \eqref{eq:x_to_bias}, we find  
\begin{align}
    \Pr[\prod_{i=1}^N\exp(X_i-\E[X_i])\geq e^{kc}]\leq \frac{1}{100^k}\quad\text{and}\quad\Pr[\prod_{i=1}^N\exp(X_i-\E[X_i])\leq e^{-kc}]\leq \frac{1}{100^k} .
\end{align}
Combining the above with \Cref{lemma:lower_qubit_typical}, we can bound \eqref{eq:f_T_bound} using a step-wise bound on the tail, which decays exponentially (recall the definition of $g_{\text{typical}}$ from \eqref{eq:lower_f}). For integers $k\geq 0$ define $A_k(U):=\{T|(T,U)\text{~typical},kc \leq |\sum_i X_i-\E[X_i]|\leq (k+1)c\}$. Then, with $1-4c_N\leq g_{\text{typical}}^U\leq 1+4c_N$,
\begin{align}
    \int_U dU \sum_{T\in\mathcal{T}(U)}f(T,U)\Pr[T,U]
    &\leq \int_U dU \sum_{k=0}\sum_{A_k(U)}f(T,U)\Pr[T,U] 
    \\
    &\leq 2\int_UdU\sum_{k=0}\sum_{A_k(U)}\Pr[T,U]\left| (1+4c_N)e^{(k+1)c}-1\right|
    \\
    &\leq 2\int_U dU\left| (1+4c_N)e^{c}-1\right| 
    \\& + 2(1+4c_N)\int_U dU \sum_{k=1}\Pr[\Big|\sum_iX_i-\E[X_i]\Big|\geq kc]e^{(k+1)c}
    \\
    &\leq 2((1+4c_N)\left(1+\frac{1}{20}\right)-1) + 2(1+4c_N)\sum_{k=1}\frac{2e^{(k+1)c}}{100^k}
    \\
    &\leq \frac{1}{10}+O(c+c_N).
\end{align}
Lastly, note that
\begin{align}
    \int_U dU \sum_{T\notin \mathcal{T}(U)}\Pr[T,U]f(T,U)
    &= \int_U dU\sum_{T\notin \mathcal{T}(U)}\left|\prod_{i=1}^Nc_{i,j}(y_{i,j,U}+\eps_{i,j,U})-\prod_{i=1}^Nc_{i,j}(y_{i,j,U}+2\eps_{i,j,U})\right|
    \\
    &\leq \int_U dU\sum_{T\notin \mathcal{T}(U)}\left(\prod_{i=1}^Nc_{i,j}(y_{i,j,U}+\eps_{i,j,U})+\prod_{i=1}^Nc_{i,j}(y_{i,j,U}+2\eps_{i,j,U})\right)
    \\
    &\leq O\left(\Pr[T~\text{not typical for }\rho_0]+\Pr[T~\text{not typical for }\rho_1]\right)
    \\
    &\leq O\left(\Pr[T~\text{not typical for }\rho_0]\right)
    \\
    &\leq O\left(\Pr[\text{\Cref{lemma:lower_nonadaptive_buckets} doesn't hold}]+\Pr[\exists~\text{light element}]\right)
    \\
    &=O(p_a)+o(1),
\end{align}
where we used that $\rho_{1,U}\leq 2\rho_{0,U}$ such that the bucket assignments are shifted by at most one. Following the proof of \Cref{lemma:helper_lower}, this means that the probability of a $T$ being non-typical for $\rho_{1,U}$ can also change by at most a constant factor. Therefore, we bounded (see \eqref{eq:lower_bound_tvd_in_R} and \eqref{eq:lower_rewrite_tv}) $\|p_0^N(M)-p_1^N(M)\|_1\leq 1/3$, and can conclude our claimed lower bound by \Cref{lemma:dist_lower_bound} if $c_N$, $p_a$, and $c$ are chosen suitably small. %

\section{Certification \& Equivalence Testing of Quantum Markov Chains}\label{sec:cmitesting}

In this section, we prove the following result for testing quantum Markov chains, that is quantum states which satisfy $I(A:C|B)_{\rho}:=D(\rho_{ABC}\|\rho_{AB}\otimes 1_C)-D(\rho_{BC}\|\rho_{B}\otimes 1_C)=0$.

\lemmarkov*

Our results are analogous to the work of \cite{DBLP:journals/tit/GaoY25}, who studied the same problems under multi-copy measurements.

In order to prove the above theorem, we need the following lemma from \cite{DBLP:journals/tit/GaoY25}, which states that for certification and equivalence testing, it is sufficient to test marginals of the Markov chain.
\begin{lemma}[{\cite[Thm.\ 3.7]{DBLP:journals/tit/GaoY25}}]
\label{lemma:cmi_decomposition}
For states $\rho_{AB}$, $\sigma_{AB}$, $\rho_{BC}$ and $\sigma_{BC}$. Let $\rho_B:=\tr_A[\rho_{AB}]$ and $\sigma_B:=\tr_C[\sigma_{BC}]$. Let $\eps_1,\eps_2,\eps_3\in [0,1]$ be such that
    \begin{align}
        F(\rho_{AB},\sigma_{AB})\geq 1-\eps_1,\quad 
        F(\rho_{BC},\sigma_{BC})\geq 1-\eps_2,\quad 
        F(\rho_{B},\sigma_{B})\geq 1-\eps_3.
    \end{align}
    Then
    \begin{align}
        F(\rho_{BC}^{1/2}\rho_{B}^{-1/2}\rho_{AB}\rho_{B}^{-1/2}\rho_{BC}^{1/2},\sigma_{BC}^{1/2}\sigma_{B}^{-1/2}\sigma_{AB}\sigma_{B}^{-1/2}\sigma_{BC}^{1/2})\geq 1-2(\sqrt{\eps_1}+\sqrt{\eps_2}+\sqrt{\eps_3})^2.
    \end{align}
\end{lemma}

Now we are ready to prove \Cref{lem:markov}. 

\begin{proof}[Proof of \Cref{lem:markov}]
From \Cref{lemma:cmi_decomposition}, it is clear that it is sufficient to test the marginals of the Markov chain for our purpose, to a precision in $\Theta(\eps)$. For (i), we perform the certification with respect to the fidelity of the marginals $A,B$ and $B,C$ with respect to the known state $\sigma_{ABC}$ and unknown state $\rho_{ABC}$. We use our certification algorithm from \Cref{lemma:fid_id_testing}. The sample complexity of this testing problem then follows directly.

For (ii), since both the states $\sigma_{ABC}$ and $\rho_{ABC}$ are unknown, we perform equivalence testing with respect to the fidelity  of the marginals $A,B$ and $B,C$. The final sample complexity follows from our equivalence testing bound in \Cref{lemma:fid_eq_testing}.
\end{proof}

\section{Acknowledgements}
We thank the reviewers for their comments and suggestions which improved the presentation of the paper.
This project is supported by the National Research Foundation, Singapore through the National Quantum Office, hosted in A*STAR, under its Centre for Quantum Technologies Funding Initiative (S24Q2d0009) and by the NRF Investigatorship award (NRF-NRFI10-2024-0006). J.S.\ is supported by the National Research Foundation, Singapore, through the National Quantum Office, hosted in A*STAR, under its National Quantum Scholarship Scheme (NQSS) 
Funding Initiative (PhD). J.S.\ would like to thank Paco Aghamalyan-Lim for his encouragement throughout this project.

\paragraph{Statement on AI usage:}
The motivations, ideas, and technical results in this work were derived and proven by the authors.
ChatGPT Plus was used for basic feedback on writing and for retrieving basic inequalities and literature pointers, particularly regarding the $\ell_{\infty}$-norm (e.g., \Cref{lemma:op_prop_1} or \Cref{fact:bound_overlap}, and advice which led the authors to derive a more efficient use of the add-zero trick in \eqref{eq:op_bound_intermediate}). ChatGPT-5.6 Sol Pro was used for technical proof-checking, and found typos and a few minor technical issues, as well as a conditioning and normalization issue in \Cref{sec:lower_equiv_fidelity_qubits}, the latter requiring a modification to the original argument while the other issues were very straightforward to repair. The issues were resolved by the authors.

\bibliographystyle{alpha} 
\bibliography{biblio}
\newpage
\appendix

\end{document}